\documentclass[paper,notoc,aps,floats,nofootinbib,tightenlines,superscriptaddress]{JHEP3}
\pdfoutput=1

\usepackage{epsfig}
\usepackage{graphicx}
\usepackage{feynmp}

\usepackage{amsmath}
\usepackage{amsthm}
\usepackage{amsfonts}
\usepackage{amssymb}
\usepackage{graphicx}
\usepackage{lscape}
\usepackage{color}
\usepackage{comment}

\usepackage{units}
\usepackage{rotating}

\usepackage{pb-diagram}
\usepackage{amstext, xxcolor, xcolor, color, float}
\usepackage{tikz, pgf,pgflibrarysnakes}
\usetikzlibrary{decorations,arrows,calc}
\usetikzlibrary{decorations.markings}
\usepackage{lscape}

\let\normalcolor\relax 

\def\beq{\begin{equation}}
\def\eeq{\end{equation}}
\newcommand{\bea}{\begin{eqnarray}}
\newcommand{\eea}{\end{eqnarray}}

\def\bsp#1\esp{\begin{split}#1\end{split}}

\newcommand{\rdh}{d}

\tikzset{dashdotted/.style={dash pattern=on .8pt off 1.5pt on 4pt off 1.5pt}}

\def \sha{{\,\amalg\hskip -3.6pt\amalg\,}}

\newtheorem{Theorem}    {Theorem}[section]

\newtheorem{Remark}      [Theorem]{Remark}

\newtheorem{Definition} [Theorem]{Definition}

\newtheorem{Proposition} [Theorem]{Proposition}

\newcommand{\dualpart}[5]{
\draw node[dot,label = below:$2$] (P1) at (#1-{2*#5},#2-{2*#5}) {};
\draw node[dot,label = below:$2$] (P2) at (#1,#2-{2.5*#5}) {};
\draw node[circle,draw] (L1) at (#1-{1.2*#5},#2-{1.2*#5}) {#3};
\draw node[circle,draw] (L2) at (#1,#2-{1.5*#5}) {#4};
\draw [-] (#1,#2) to (L1) {};
\draw [-] (#1,#2) to (L2) {};
\draw [-] (L1) to (P1) {};
\draw [-] (L2) to (P2) {};
}

\newcommand{\zerogroup}[4]{
\draw node[dot] (Z) at (#1 - #2:1.15cm) {};
\draw node (L1) at (#1 - #2:1.25cm) {$0$};
\midarrow{Z}{#3};
\draw node (L2) at (#1+{#2}/4:1.25cm) {$\cdot$};
\draw node (L2) at (#1-{#2}/4:1.25cm) {$\cdot$};
\draw node (L2) at (#1+{#2}/4:1.15cm) {$\cdot$};
\draw node (L2) at (#1-{#2}/4:1.15cm) {$\cdot$};
\draw node[dot] (Z) at (#1 + #2:1.15cm) {};
\draw node (L3) at (#1 + #2:1.25cm) {$0$};
\midarrow{Z}{#3};
\draw node(a1) at (#1 - #2-2:1.35cm) {};
\draw node (a2) at (#1 + #2+2:1.35cm) {};
\draw[(-)] (a1) to (a2) {};
\draw node (L4) at (#1:1.45cm) {#4};
}

\newcommand{\ot}{\otimes}

\newcommand{\midarrow}[2]{\draw[-,decoration = {markings,mark=at position 0.5 with \arrow[thick]{>}},postaction={decorate}] (#1) to (#2)}

\newcommand{\octexample}[5]{

\begin{tikzpicture}[inner sep=1.2pt,scale = #5, dot/.style={fill,circle,minimum size=1pt}]

\draw node[label=above:$2$] (H1) at (0.7,-0.9){};
\draw node[label=left:$0$] (H2) at (0,-3){};
\draw node[label=below:$2$] (H3) at (0.7,-5.1){};
\draw node[label=below:$0$] (H4) at (3,-6){};
\draw node[label=below:$0$] (H5) at (5.3,-5.1){};
\draw node[label=right:$2$] (H6) at (6,-3){};
\draw node[label=above:$0$] (H7) at (5.3,-0.9){};
\draw node[label=above:$1$] (H8) at (3,0){};

\color{black!#1!white}
\path(1.8,0) node[circle,draw,fill,minimum size = 5pt] (Q1){};
\path(4.2,0) coordinate (Q8);
\path(6,-1.8) coordinate (Q7);
\path(6,-4.2) coordinate(Q6);
\path(4.2,-6) coordinate(Q5);
\path(1.8,-6) coordinate(Q4);
\path(0,-4.2) coordinate(Q3);
\path(0,-1.8) coordinate(Q2);
\begin{scope}
\draw [double distance = 3pt] (Q8)--(Q1)node[pos=0.33](thirdq41){};
\end{scope}
\draw(Q1)-- (Q2)--(Q3)--(Q4)--(Q5)--(Q6)--(Q7)--(Q8);

\colorlet{darkgreen}{green!50!black}
\color{darkgreen!#3!white}
\draw node[draw,circle,minimum size = 8pt] (C18) at (1.95,-0.45){};
\draw node[dot] (M18) at (1.95,-0.45){};
\draw node[dot] (M38) at (1.95,-2.55){};
\draw node[dot] (M23) at (0.45,-4.05){};
\draw node[dot] (M34) at (1.95,-5.55){};
\draw node[dot] (M56) at (5.55,-4.05){};
\draw node[dot] (M67) at (5.55,-1.95){};
\draw node[dot] (M68) at (4.5,-1.5){};

\draw [-,dashdotted] (M18) to (M38) {};
\draw [-,dashdotted] (M38) to (M23) {};
\draw [-,dashdotted] (M38) to (M34) {};
\draw [-,dashdotted] (M38) to (M68) {};
\draw [-,dashdotted] (M68) to (M56) {};
\draw [-,dashdotted] (M68) to (M67) {};

\color{blue!#2!white}
\draw node[dot] (H1) at (0.9,-0.9){};
\draw node[dot] (H2) at (0,-3){};
\draw node[dot] (H3) at (0.9,-5.1){};
\draw node[dot] (H4) at (3,-6){};
\draw node[dot] (H5) at (5.1,-5.1){};
\draw node[dot] (H6) at (6,-3){};
\draw node[dot] (H7) at (5.1,-0.9){};
\draw node[dot] (H8) at (3,0){};

\midarrow{H1}{H8};
\midarrow{H3}{H8};
\midarrow{H2}{H3};
\midarrow{H4}{H3};
\midarrow{H5}{H6};
\midarrow{H7}{H6};
\midarrow{H6}{H8};

\color{red!#4!white}
\draw [->,densely dashed] (Q2) to (H8) {};
\draw [->,densely dashed] (Q2) to (H3) {};
\draw [->,densely dashed] (Q5) to (H3) {};
\draw [->,densely dashed] (Q5) to (H8) {};
\draw [->,densely dashed] (Q5) to (H6) {};
\draw [->,densely dashed] (Q8) to (H6) {};

\end{tikzpicture}
}

\newcommand{\octexampletwo}[1]{

\begin{tikzpicture}[inner sep=1.2pt,scale = #1, dot/.style={fill,circle,minimum size=1pt}]

\draw node[label=above:$2$] (H1) at (0.7,-0.9){};
\draw node[label=left:$0$] (H2) at (0,-3){};
\draw node[label=below:$2$] (H3) at (0.7,-5.1){};
\draw node[label=below:$0$] (H4) at (3,-6){};
\draw node[label=below:$0$] (H5) at (5.3,-5.1){};
\draw node[label=right:$2$] (H6) at (6,-3){};
\draw node[label=above:$0$] (H7) at (5.3,-0.9){};
\draw node[label=above:$1$] (H8) at (3,0){};

\color{black}
\path(1.8,0) node[circle,draw,fill,minimum size = 5pt] (Q1){};
\path(4.2,0) coordinate (Q8);
\path(6,-1.8) coordinate (Q7);
\path(6,-4.2) coordinate(Q6);
\path(4.2,-6) coordinate(Q5);
\path(1.8,-6) coordinate(Q4);
\path(0,-4.2) coordinate(Q3);
\path(0,-1.8) coordinate(Q2);
\begin{scope}
\draw [double distance = 3pt] (Q8)--(Q1)node[pos=0.33](thirdq41){};
\end{scope}
\draw(Q1)-- (Q2)--(Q3)--(Q4)--(Q5)--(Q6)--(Q7)--(Q8);

\draw node (H1) at (0.9,-0.9){};
\draw node (H2) at (0,-3){};
\draw node (H3) at (0.9,-5.1){};
\draw node (H4) at (3,-6){};
\draw node (H5) at (5.1,-5.1){};
\draw node (H6) at (6,-3){};
\draw node (H7) at (5.1,-0.9){};
\draw node (H8) at (3,0){};

\color{red}
\draw [->,densely dashed] (Q2) to (H8) {};
\draw [->,densely dashed] (Q2) to (H3) {};
\draw [->,densely dashed] (Q5) to (H3) {};
\draw [->,densely dashed] (Q5) to (H8) {};
\draw [->,densely dashed] (Q5) to (H6) {};
\draw [->,densely dashed] (Q8) to (H6) {};

\end{tikzpicture}
}

\newcommand{\octexamplethree}[3]{

\begin{tikzpicture}[inner sep=1.2pt,scale = #1, dot/.style={fill,circle,minimum size=1pt}]

\draw node[label=above:$2$] (H1) at (0.7,-0.9){};
\draw node[label=left:$0$] (H2) at (0,-3){};
\draw node[label=below:$2$] (H3) at (0.7,-5.1){};
\draw node[label=below:$0$] (H4) at (3,-6){};
\draw node[label=below:$0$] (H5) at (5.3,-5.1){};
\draw node[label=right:$2$] (H6) at (6,-3){};
\draw node[label=above:$0$] (H7) at (5.3,-0.9){};
\draw node[label=above:$1$] (H8) at (3,0){};

\color{black}
\path(1.8,0) node[circle,draw,fill,minimum size = 5pt] (Q1){};
\path(4.2,0) coordinate (Q8);
\path(6,-1.8) coordinate (Q7);
\path(6,-4.2) coordinate(Q6);
\path(4.2,-6) coordinate(Q5);
\path(1.8,-6) coordinate(Q4);
\path(0,-4.2) coordinate(Q3);
\path(0,-1.8) coordinate(Q2);
\begin{scope}
\draw [double distance = 3pt] (Q8)--(Q1)node[pos=0.33](thirdq41){};
\end{scope}
\draw(Q1)-- (Q2)--(Q3)--(Q4)--(Q5)--(Q6)--(Q7)--(Q8);

\draw node (H1) at (0.9,-0.9){};
\draw node (H2) at (0,-3){};
\draw node (H3) at (0.9,-5.1){};
\draw node (H4) at (3,-6){};
\draw node (H5) at (5.1,-5.1){};
\draw node (H6) at (6,-3){};
\draw node (H7) at (5.1,-0.9){};
\draw node (H8) at (3,0){};

\color{red}
\draw [->,densely dashed] (Q2) to (H8) {};
\draw [->,densely dashed] (Q2) to (H3) {};
\draw [->,densely dashed] (Q5) to (H3) {};
\draw [->,densely dashed] (Q5) to (H8) {};
\draw [->,densely dashed] (Q5) to (H6) {};
\draw [->,densely dashed] (Q8) to (H6) {};

\color{blue!#2!white}
\draw node[dot] (H1) at (0.9,-0.9){};
\draw node[dot] (H2) at (0,-3){};
\draw node[dot] (H3) at (0.9,-5.1){};
\draw node[dot] (H4) at (3,-6){};
\draw node[dot] (H5) at (5.1,-5.1){};
\draw node[dot] (H6) at (6,-3){};
\draw node[dot] (H7) at (5.1,-0.9){};
\draw node[dot] (H8) at (3,0){};
\color{black}

\color{blue!#3!white}
\draw [-] (H1) to (H8) {};
\draw [-] (H3) to (H8) {};
\draw [-] (H2) to (H3) {};
\draw [-] (H4) to (H3) {};
\draw [-] (H5) to (H6) {};
\draw [-] (H7) to (H6) {};
\draw [-] (H6) to (H8) {};

\end{tikzpicture}
}

\newcommand{\octexamplefour}[5]{

\begin{tikzpicture}[inner sep=1.2pt,scale = #5, dot/.style={fill,circle,minimum size=1pt}]

\draw node[label=above:$2$] (H1) at (0.7,-0.9){};
\draw node[label=left:$0$] (H2) at (0,-3){};
\draw node[label=below:$2$] (H3) at (0.7,-5.1){};
\draw node[label=below:$0$] (H4) at (3,-6){};
\draw node[label=below:$0$] (H5) at (5.3,-5.1){};
\draw node[label=right:$2$] (H6) at (6,-3){};
\draw node[label=above:$0$] (H7) at (5.3,-0.9){};
\draw node[label=above:$1$] (H8) at (3,0){};

\color{black!#1!white}
\path(1.8,0) node[circle,draw,fill,minimum size = 5pt] (Q1){};
\path(4.2,0) coordinate (Q8);
\path(6,-1.8) coordinate (Q7);
\path(6,-4.2) coordinate(Q6);
\path(4.2,-6) coordinate(Q5);
\path(1.8,-6) coordinate(Q4);
\path(0,-4.2) coordinate(Q3);
\path(0,-1.8) coordinate(Q2);
\begin{scope}
\draw [double distance = 3pt] (Q8)--(Q1)node[pos=0.33](thirdq41){};
\end{scope}
\draw(Q1)-- (Q2)--(Q3)--(Q4)--(Q5)--(Q6)--(Q7)--(Q8);

\colorlet{darkgreen}{green!50!black}
\color{darkgreen!#3!white}
\draw node[draw,circle,minimum size = 8pt] (C18) at (1.95,-0.45){};
\draw node[dot] (M18) at (1.95,-0.45){};
\draw node[dot] (M38) at (1.95,-2.55){};
\draw node[dot] (M23) at (0.45,-4.05){};
\draw node[dot] (M34) at (1.95,-5.55){};
\draw node[dot] (M56) at (5.55,-4.05){};
\draw node[dot] (M67) at (5.55,-1.95){};
\draw node[dot] (M68) at (4.5,-1.5){};

\draw [-,dashdotted] (M18) to (M38) {};
\draw [-,dashdotted] (M38) to (M23) {};
\draw [-,dashdotted] (M38) to (M34) {};
\draw [-,dashdotted] (M38) to (M68) {};
\draw [-,dashdotted] (M68) to (M56) {};
\draw [-,dashdotted] (M68) to (M67) {};

\color{red!#4!white}
\draw [->,densely dashed] (Q2) to (H8) {};
\draw [->,densely dashed] (Q2) to (0.95,-5.05) {};
\draw [->,densely dashed] (Q5) to (0.95,-5.05) {};
\draw [->,densely dashed] (Q5) to (H8) {};
\draw [->,densely dashed] (Q5) to (H6) {};
\draw [->,densely dashed] (Q8) to (H6) {};

\end{tikzpicture}
}

\newcommand{\eps}{\epsilon}

\newcommand{\ord}{\begin{cal}O\end{cal}}
\newcommand{\cI}{\begin{cal}I\end{cal}}

\newcommand{\cA}{{\cal A}}
\newcommand{\cB}{{\cal B}}

\newcommand{\cW}{{\cal W}}
\newcommand{\cL}{{\cal L}}

\def\Q{{\cal Q} }

\def\C{{\cal C} }

\def \Z{{\mathbb Z}}
\def\C{{\mathbb C} }
\def\Q{{\mathbb Q} }

\usepackage[all]{xy}
\xyoption{poly}
\xyoption{arc}

\def \fourgon#1#2#3#4{{
\xy
\POS(10,4) \ar@{=} +(-10,0)_#4
\ar@{-} +(0,-10)^#3
\POS(10,-6) \ar@{-} +(-10,0)^#2
\POS(0,4) \ar@{-} +(0,-10)_#1
\POS(0,4) *+{\bullet}
\endxy
}}

\def \fourgondotted#1#2#3#4{{
\xy
\POS(10,4) \ar@{=} +(-10,0)_#4
\ar@{-} +(0,-10)^#3
\POS(10,-6) \ar@{-} +(-10,0)^#2
\POS(0,4) \ar@{.} +(0,-10)_#1
\endxy
}}

\def \fourgona{{
\POS(10,5) \ar@{=} +(-10,0)_x
\ar@{-} +(0,-10)^a
\POS(10,-5) \ar@{-} +(-10,0)^b
\POS(0,5) \ar@{-} +(0,-10)_c
\POS(0,5) *+{\bullet}
}}

\def \threegon#1#2#3{{
\xy
\POS(10,4) \ar@{=} +(-10,0)_{#3} 
\ar@{-} +(-5,-8)^{#2}  
\POS(0,4) \ar@{-} +(5,-8)_{#1}
\POS(0,4) *+{\bullet}
\endxy
}}

\def \threegonarrowc#1#2#3{{
\xy  
\POS(10,4) \ar@{=} +(-10,0)_{#3}
\ar@{-} +(-5,-8)^{#2}
\POS(0,4) \ar@{-} +(5,-8)_{#1}
\POS(5,4) \ar@{<<-} +(0,-8)
\POS(0,4) *+{\bullet}
\endxy
}}

\def \threegonarrowb#1#2#3{{
\xy  
\POS(10,4) \ar@{=} +(-10,0)_{#3}
\ar@{-} +(-5,-8)^{#2}
\POS(0,4) \ar@{-} +(5,-8)_{#1}
\POS(7.5,0) \ar@{<<-} +(-7.5,4)
\POS(0,4) *+{\bullet}
\endxy
}}

\def \threegonarrowa#1#2#3{{
\xy  
\POS(10,4) \ar@{=} +(-10,0)_{#3}
\ar@{-} +(-5,-8)^{#2}
\POS(0,4) \ar@{-} +(5,-8)_{#1}
\POS(2.5,0) \ar@{<<-} +(7.5,4)
\POS(0,4) *+{\bullet}
\endxy
}}

\def \twogon#1#2{\hskip 5pt\xy
\POS(10,0) \ar@{=} +(-10,0)_{#2}
\POS(0,0)*{};
\POS(0,0) *+{\bullet}
\POS(10,0)*{};
**\crv{(5,-5)};
\POS(5,-5) *+{\scriptstyle #1}
\endxy \hskip 5pt}

\def \twogonright#1#2{\hskip 5pt\xy
\POS(10,0) \ar@{=} +(-10,0)_#2
\POS(0,0)*{};
\POS(10,0) *+{\bullet}
\POS(10,0)*{};
**\crv{(5,-5)};
\POS(5,-5) *+{\scriptstyle #1}
\endxy \hskip 5pt}

\def \fivegon#1#2#3#4#5{{
\xy
\POS(10,4) \ar@{=} +(-10,0)_{#5}
\ar@{-} +(3,-7)^{#4}
\POS(0,4) \ar@{-} +(-3,-7)_{#1}
\POS(0,4) *+{\bullet}
\POS(-3,-3) \ar@{-} +(8,-7)_{#2}
\POS(5,-10)\ar@{-} +(8,7)_{#3}
\endxy
}}

\def \fivegonarrowtCDa#1#2#3#4#5{{
\xy
\POS(10,4) \ar@{=} +(-10,0)_{#5}
\ar@{-} +(3,-7)^{#4}
\POS(0,4) \ar@{-} +(-3,-7)_{#1}
\POS(0,4) *+{\bullet}
\POS(-3,-3) \ar@{-} +(8,-7)_{#2}
\POS(5,-10)\ar@{-} +(8,7)_{#3}
\POS(-3,-3)\ar@{->>} +(5.5,7)
\POS(5,-10)\ar@{->>} +(0,14)
\POS(13,-3)\ar@{->>} +(-5.5,7)
\endxy
}}

\def \fivegonarrowtCDb#1#2#3#4#5{{
\xy
\POS(10,4) \ar@{=} +(-10,0)_{#5}
\ar@{-} +(3,-7)^{#4}
\POS(0,4) \ar@{-} +(-3,-7)_{#1}
\POS(0,4) *+{\bullet}
\POS(-3,-3) \ar@{-} +(8,-7)_{#2}
\POS(5,-10)\ar@{-} +(8,7)_{#3}
\POS(13,-3)\ar@{->>} +(-5.5,7)
\POS(5,-10)\ar@{->>} +(0,14)
\POS(0,4)\ar@{->>} +(1,-10.5)
\endxy
}}

\def \fivegonarrowtCDc#1#2#3#4#5{{
\xy
\POS(10,4) \ar@{=} +(-10,0)_{#5}
\ar@{-} +(3,-7)^{#4}
\POS(0,4) \ar@{-} +(-3,-7)_{#1}
\POS(0,4) *+{\bullet}
\POS(-3,-3) \ar@{-} +(8,-7)_{#2}
\POS(5,-10)\ar@{-} +(8,7)_{#3}
\POS(-3,-3)\ar@{->>} +(5.5,7)
\POS(5,-10)\ar@{->>} +(0,14)
\POS(10,4)\ar@{->>} +(-1,-10.5)
\endxy
}}

\def \fivegonarrowtCDd#1#2#3#4#5{{
\xy
\POS(10,4) \ar@{=} +(-10,0)_{#5}
\ar@{-} +(3,-7)^{#4}
\POS(0,4) \ar@{-} +(-3,-7)_{#1}
\POS(0,4) *+{\bullet}
\POS(-3,-3) \ar@{-} +(8,-7)_{#2}
\POS(5,-10)\ar@{-} +(8,7)_{#3}
\POS(0,4)\ar@{->>} +(1,-10.5)
\POS(5,-10)\ar@{->>} +(0,14)
\POS(10,4)\ar@{->>} +(-1,-10.5)
\endxy
}}

\def \fivegonarrowtCEa#1#2#3#4#5{{
\xy
\POS(10,4) \ar@{=} +(-10,0)_{#5}
\ar@{-} +(3,-7)^{#4}
\POS(0,4) \ar@{-} +(-3,-7)_{#1}
\POS(0,4) *+{\bullet}
\POS(-3,-3) \ar@{-} +(8,-7)_{#2}
\POS(5,-10)\ar@{-} +(8,7)_{#3}
\POS(0,4)\ar@{->>} +(10.75,-1.75)
\POS(-3,-3)\ar@{->>} +(14.5,3.5)
\POS(5,-10)\ar@{->>} +(7.25,8.75)
\endxy
}}

\def \fivegonarrowtCEb#1#2#3#4#5{{
\xy
\POS(10,4) \ar@{=} +(-10,0)_{#5}
\ar@{-} +(3,-7)^{#4}
\POS(0,4) \ar@{-} +(-3,-7)_{#1}
\POS(0,4) *+{\bullet}
\POS(-3,-3) \ar@{-} +(8,-7)_{#2}
\POS(5,-10)\ar@{-} +(8,7)_{#3}
\POS(10,4)\ar@{->>} +(-11.5,-3.5)
\POS(-3,-3)\ar@{->>} +(14.5,3.5)
\POS(5,-10)\ar@{->>} +(7.25,8.75)
\endxy
}}

\def \fivegonarrowtCEc#1#2#3#4#5{{
\xy
\POS(10,4) \ar@{=} +(-10,0)_{#5}
\ar@{-} +(3,-7)^{#4}
\POS(0,4) \ar@{-} +(-3,-7)_{#1}
\POS(0,4) *+{\bullet}
\POS(-3,-3) \ar@{-} +(8,-7)_{#2}
\POS(5,-10)\ar@{-} +(8,7)_{#3}
\POS(0,4)\ar@{->>} +(10.75,-1.75)
\POS(-3,-3)\ar@{->>} +(14.5,3.5)
\POS(13,-3)\ar@{->>} +(-12,-3.5)
\endxy
}}

\def \fivegonarrowtCEd#1#2#3#4#5{{
\xy
\POS(10,4) \ar@{=} +(-10,0)_{#5}
\ar@{-} +(3,-7)^{#4}
\POS(0,4) \ar@{-} +(-3,-7)_{#1}
\POS(0,4) *+{\bullet}
\POS(-3,-3) \ar@{-} +(8,-7)_{#2}
\POS(5,-10)\ar@{-} +(8,7)_{#3}
\POS(10,4)\ar@{->>} +(-11.5,-3.5)
\POS(-3,-3)\ar@{->>} +(14.5,3.5)
\POS(13,-3)\ar@{->>} +(-12,-3.5)
\endxy
}}

\def \fivegonarrowtCCa#1#2#3#4#5{{
\xy
\POS(10,4) \ar@{=} +(-10,0)_{#5}
\ar@{-} +(3,-7)^{#4}
\POS(0,4) \ar@{-} +(-3,-7)_{#1}
\POS(0,4) *+{\bullet}
\POS(-3,-3) \ar@{-} +(8,-7)_{#2}
\POS(5,-10)\ar@{-} +(8,7)_{#3}
\POS(10,4)\ar@{->>} (-0.75,2.75)
\POS(13,-3)\ar@{->>} +(-14.5,3.5)
\POS(5,-10)\ar@{->>} (-2.25,-1.25)
\endxy
}}

\def \fivegonarrowtCCb#1#2#3#4#5{{
\xy
\POS(10,4) \ar@{=} +(-10,0)_{#5}
\ar@{-} +(3,-7)^{#4}
\POS(0,4) \ar@{-} +(-3,-7)_{#1}
\POS(0,4) *+{\bullet}
\POS(-3,-3) \ar@{-} +(8,-7)_{#2}
\POS(5,-10)\ar@{-} +(8,7)_{#3}
\POS(0,4)\ar@{->>} (11.5,0.5)
\POS(13,-3)\ar@{->>} +(-14.5,3.5)
\POS(5,-10)\ar@{->>} (-2.25,-1.25)
\endxy
}}

\def \fivegonarrowtCCc#1#2#3#4#5{{
\xy
\POS(10,4) \ar@{=} +(-10,0)_{#5}
\ar@{-} +(3,-7)^{#4}
\POS(0,4) \ar@{-} +(-3,-7)_{#1}
\POS(0,4) *+{\bullet}
\POS(-3,-3) \ar@{-} +(8,-7)_{#2}
\POS(5,-10)\ar@{-} +(8,7)_{#3}
\POS(10,4)\ar@{->>} (-0.75,2.75)
\POS(13,-3)\ar@{->>} +(-14.5,3.5)
\POS(-3,-3)\ar@{->>} (9,-6.5)
\endxy
}}

\def \fivegonarrowtCCd#1#2#3#4#5{{
\xy
\POS(10,4) \ar@{=} +(-10,0)_{#5}
\ar@{-} +(3,-7)^{#4}
\POS(0,4) \ar@{-} +(-3,-7)_{#1}
\POS(0,4) *+{\bullet}
\POS(-3,-3) \ar@{-} +(8,-7)_{#2}
\POS(5,-10)\ar@{-} +(8,7)_{#3}
\POS(0,4)\ar@{->>} (11.5,0.5)
\POS(13,-3)\ar@{->>} +(-14.5,3.5)
\POS(-3,-3)\ar@{->>} (9,-6.5)
\endxy
}}

\def \fivegonarrowtCBa#1#2#3#4#5{{
\xy
\POS(10,4) \ar@{=} +(-10,0)_{#5}
\ar@{-} +(3,-7)^{#4}
\POS(0,4) \ar@{-} +(-3,-7)_{#1}
\POS(0,4) *+{\bullet}
\POS(-3,-3) \ar@{-} +(8,-7)_{#2}
\POS(5,-10)\ar@{-} +(8,7)_{#3}
\POS(0,4)\ar@{->>} (-1,-4.75)
\POS(10,4)\ar@{->>} (1,-6.5)
\POS(13,-3)\ar@{->>} (3,-8.25)
\endxy
}}

\def \fivegonarrowtCBb#1#2#3#4#5{{
\xy
\POS(10,4) \ar@{=} +(-10,0)_{#5}
\ar@{-} +(3,-7)^{#4}
\POS(0,4) \ar@{-} +(-3,-7)_{#1}
\POS(0,4) *+{\bullet}
\POS(-3,-3) \ar@{-} +(8,-7)_{#2}
\POS(5,-10)\ar@{-} +(8,7)_{#3}
\POS(-3,-3)\ar@{->>} (5,4)
\POS(10,4)\ar@{->>} (1,-6.5)
\POS(13,-3)\ar@{->>} (3,-8.25)
\endxy
}}

\def \fivegonarrowtCBc#1#2#3#4#5{{
\xy
\POS(10,4) \ar@{=} +(-10,0)_{#5}
\ar@{-} +(3,-7)^{#4}
\POS(0,4) \ar@{-} +(-3,-7)_{#1}
\POS(0,4) *+{\bullet}
\POS(-3,-3) \ar@{-} +(8,-7)_{#2}
\POS(5,-10)\ar@{-} +(8,7)_{#3}
\POS(0,4)\ar@{->>} (-1,-4.75)
\POS(10,4)\ar@{->>} (1,-6.5)
\POS(5,-10)\ar@{->>} (11.5,0.5)
\endxy
}}

\def \fivegonarrowtCBd#1#2#3#4#5{{
\xy
\POS(10,4) \ar@{=} +(-10,0)_{#5}
\ar@{-} +(3,-7)^{#4}
\POS(0,4) \ar@{-} +(-3,-7)_{#1}
\POS(0,4) *+{\bullet}
\POS(-3,-3) \ar@{-} +(8,-7)_{#2}
\POS(5,-10)\ar@{-} +(8,7)_{#3}
\POS(-3,-3)\ar@{->>} (5,4)
\POS(10,4)\ar@{->>} (1,-6.5)
\POS(5,-10)\ar@{->>} (11.5,0.5)
\endxy
}}

\def \fivegonarrowtCAa#1#2#3#4#5{{
\xy
\POS(10,4) \ar@{=} +(-10,0)_{#5}
\ar@{-} +(3,-7)^{#4}
\POS(0,4) \ar@{-} +(-3,-7)_{#1}
\POS(0,4) *+{\bullet}
\POS(-3,-3) \ar@{-} +(8,-7)_{#2}
\POS(5,-10)\ar@{-} +(8,7)_{#3}
\POS(10,4)\ar@{->>} (11,-4.75)
\POS(0,4)\ar@{->>} (9,-6.5)
\POS(-3,-3)\ar@{->>} (7,-8.25)
\endxy
}}

\def \fivegonarrowtCAb#1#2#3#4#5{{
\xy
\POS(10,4) \ar@{=} +(-10,0)_{#5}
\ar@{-} +(3,-7)^{#4}
\POS(0,4) \ar@{-} +(-3,-7)_{#1}
\POS(0,4) *+{\bullet}
\POS(-3,-3) \ar@{-} +(8,-7)_{#2}
\POS(5,-10)\ar@{-} +(8,7)_{#3}
\POS(13,-3)\ar@{->>} (5,4)
\POS(0,4)\ar@{->>} (9,-6.5)
\POS(-3,-3)\ar@{->>} (7,-8.25)
\endxy
}}

\def \fivegonarrowtCAd#1#2#3#4#5{{
\xy
\POS(10,4) \ar@{=} +(-10,0)_{#5}
\ar@{-} +(3,-7)^{#4}
\POS(0,4) \ar@{-} +(-3,-7)_{#1}
\POS(0,4) *+{\bullet}
\POS(-3,-3) \ar@{-} +(8,-7)_{#2}
\POS(5,-10)\ar@{-} +(8,7)_{#3}
\POS(13,-3)\ar@{->>} (5,4)
\POS(0,4)\ar@{->>} (9,-6.5)
\POS(5,-10)\ar@{->>} (-1.5,0.5)
\endxy
}}

\def \fivegonarrowtCAc#1#2#3#4#5{{
\xy
\POS(10,4) \ar@{=} +(-10,0)_{#5}
\ar@{-} +(3,-7)^{#4}
\POS(0,4) \ar@{-} +(-3,-7)_{#1}
\POS(0,4) *+{\bullet}
\POS(-3,-3) \ar@{-} +(8,-7)_{#2}
\POS(5,-10)\ar@{-} +(8,7)_{#3}
\POS(10,4)\ar@{->>} (11,-4.75)
\POS(0,4)\ar@{->>} (9,-6.5)
\POS(5,-10)\ar@{->>} (-1.5,0.5)
\endxy
}}

\def \fivegonarrowtBEa#1#2#3#4#5{{
\xy
\POS(10,4) \ar@{=} +(-10,0)_{#5}
\ar@{-} +(3,-7)^{#4}
\POS(0,4) \ar@{-} +(-3,-7)_{#1}
\POS(0,4) *+{\bullet}
\POS(-3,-3) \ar@{-} +(8,-7)_{#2}
\POS(5,-10)\ar@{-} +(8,7)_{#3}

\POS(-3,-3)\ar@{->>} (5,4)
\POS(-3,-3)\ar@{->>} (10.75,2.75)
\POS(5,-10)\ar@{->>} (12.25,-1.25)

\endxy
}}

\def \fivegonarrowtBEb#1#2#3#4#5{{
\xy
\POS(10,4) \ar@{=} +(-10,0)_{#5}
\ar@{-} +(3,-7)^{#4}
\POS(0,4) \ar@{-} +(-3,-7)_{#1}
\POS(0,4) *+{\bullet}
\POS(-3,-3) \ar@{-} +(8,-7)_{#2}
\POS(5,-10)\ar@{-} +(8,7)_{#3}

\POS(-3,-3)\ar@{->>} (5,4)
\POS(-3,-3)\ar@{->>} (10.75,2.75)
\POS(13,-3)\ar@{->>} (1,-6.5)

\endxy
}}

\def \fivegonarrowtBEc#1#2#3#4#5{{
\xy
\POS(10,4) \ar@{=} +(-10,0)_{#5}
\ar@{-} +(3,-7)^{#4}
\POS(0,4) \ar@{-} +(-3,-7)_{#1}
\POS(0,4) *+{\bullet}
\POS(-3,-3) \ar@{-} +(8,-7)_{#2}
\POS(5,-10)\ar@{-} +(8,7)_{#3}

\POS(-3,-3)\ar@{->>} (9,-6.5)
\POS(-3,-3)\ar@{->>} (12.25,-1.5)
\POS(0,4)\ar@{->>} (10.75,2.25)

\endxy
}}

\def \fivegonarrowtBEd#1#2#3#4#5{{
\xy
\POS(10,4) \ar@{=} +(-10,0)_{#5}
\ar@{-} +(3,-7)^{#4}
\POS(0,4) \ar@{-} +(-3,-7)_{#1}
\POS(0,4) *+{\bullet}
\POS(-3,-3) \ar@{-} +(8,-7)_{#2}
\POS(5,-10)\ar@{-} +(8,7)_{#3}

\POS(-3,-3)\ar@{->>} (9,-6.5)
\POS(-3,-3)\ar@{->>} (11.5,0.5)
\POS(10,4)\ar@{->>} (-1.5,0.5)

\endxy
}}

\def \fivegonarrowtBEe#1#2#3#4#5{{
\xy
\POS(10,4) \ar@{=} +(-10,0)_{#5}
\ar@{-} +(3,-7)^{#4}
\POS(0,4) \ar@{-} +(-3,-7)_{#1}
\POS(0,4) *+{\bullet}
\POS(-3,-3) \ar@{-} +(8,-7)_{#2}
\POS(5,-10)\ar@{-} +(8,7)_{#3}

\POS(-3,-3)\ar@{->>} (7,-8.25)
\POS(-3,-3)\ar@{->>} (5,4)
\POS(10,4)\ar@{->>} (11,-4.75)

\endxy
}}

\def \fivegonarrowtBEf#1#2#3#4#5{{
\xy
\POS(10,4) \ar@{=} +(-10,0)_{#5}
\ar@{-} +(3,-7)^{#4}
\POS(0,4) \ar@{-} +(-3,-7)_{#1}
\POS(0,4) *+{\bullet}
\POS(-3,-3) \ar@{-} +(8,-7)_{#2}
\POS(5,-10)\ar@{-} +(8,7)_{#3}

\POS(-3,-3)\ar@{->>} (9,-6.5)
\POS(-3,-3)\ar@{->>} (2.5,4)
\POS(13,-3)\ar@{->>} (7.5,4)

\endxy
}}

\def \fivegonarrowtBDa#1#2#3#4#5{{
\xy
\POS(10,4) \ar@{=} +(-10,0)_{#5}
\ar@{-} +(3,-7)^{#4}
\POS(0,4) \ar@{-} +(-3,-7)_{#1}
\POS(0,4) *+{\bullet}
\POS(-3,-3) \ar@{-} +(8,-7)_{#2}
\POS(5,-10)\ar@{-} +(8,7)_{#3}

\POS(10,4)\ar@{->>} (-1.5,0.5)
\POS(10,4)\ar@{->>} (-1,-4.75)
\POS(13,-3)\ar@{->>} (3,-8.75)

\endxy
}}

\def \fivegonarrowtBDb#1#2#3#4#5{{
\xy
\POS(10,4) \ar@{=} +(-10,0)_{#5}
\ar@{-} +(3,-7)^{#4}
\POS(0,4) \ar@{-} +(-3,-7)_{#1}
\POS(0,4) *+{\bullet}
\POS(-3,-3) \ar@{-} +(8,-7)_{#2}
\POS(5,-10)\ar@{-} +(8,7)_{#3}

\POS(10,4)\ar@{->>} (-1.5,0.5)
\POS(10,4)\ar@{->>} (-1,-4.75)
\POS(5,-10)\ar@{->>} (11.5,0.5)

\endxy
}}

\def \fivegonarrowtBDc#1#2#3#4#5{{
\xy
\POS(10,4) \ar@{=} +(-10,0)_{#5}
\ar@{-} +(3,-7)^{#4}
\POS(0,4) \ar@{-} +(-3,-7)_{#1}
\POS(0,4) *+{\bullet}
\POS(-3,-3) \ar@{-} +(8,-7)_{#2}
\POS(5,-10)\ar@{-} +(8,7)_{#3}

\POS(10,4)\ar@{->>} (-1.5,0.5)
\POS(10,4)\ar@{->>} (11,-4.75)
\POS(-3,-3)\ar@{->>} (7,-8.75)

\endxy
}}

\def \fivegonarrowtBDd#1#2#3#4#5{{
\xy
\POS(10,4) \ar@{=} +(-10,0)_{#5}
\ar@{-} +(3,-7)^{#4}
\POS(0,4) \ar@{-} +(-3,-7)_{#1}
\POS(0,4) *+{\bullet}
\POS(-3,-3) \ar@{-} +(8,-7)_{#2}
\POS(5,-10)\ar@{-} +(8,7)_{#3}

\POS(10,4)\ar@{->>} (-0.75,2.25)
\POS(10,4)\ar@{->>} (11,-4.75)
\POS(5,-10)\ar@{->>} (-2.25,-1.25)

\endxy
}}

\def \fivegonarrowtBDe#1#2#3#4#5{{
\xy
\POS(10,4) \ar@{=} +(-10,0)_{#5}
\ar@{-} +(3,-7)^{#4}
\POS(0,4) \ar@{-} +(-3,-7)_{#1}
\POS(0,4) *+{\bullet}
\POS(-3,-3) \ar@{-} +(8,-7)_{#2}
\POS(5,-10)\ar@{-} +(8,7)_{#3}

\POS(10,4)\ar@{->>} (9,-6.5)
\POS(10,4)\ar@{->>} (1,-6.5)
\POS(-3,-3)\ar@{->>} (5,4)

\endxy
}}

\def \fivegonarrowtBDf#1#2#3#4#5{{
\xy
\POS(10,4) \ar@{=} +(-10,0)_{#5}
\ar@{-} +(3,-7)^{#4}
\POS(0,4) \ar@{-} +(-3,-7)_{#1}
\POS(0,4) *+{\bullet}
\POS(-3,-3) \ar@{-} +(8,-7)_{#2}
\POS(5,-10)\ar@{-} +(8,7)_{#3}

\POS(10,4)\ar@{->>} (9,-6.5)
\POS(10,4)\ar@{->>} (3,-8.25)
\POS(0,4)\ar@{->>} (-1,-4.75)

\endxy
}}

\def \fivegonarrowtBCa#1#2#3#4#5{{
\xy
\POS(10,4) \ar@{=} +(-10,0)_{#5}
\ar@{-} +(3,-7)^{#4}
\POS(0,4) \ar@{-} +(-3,-7)_{#1}
\POS(0,4) *+{\bullet}
\POS(-3,-3) \ar@{-} +(8,-7)_{#2}
\POS(5,-10)\ar@{-} +(8,7)_{#3}

\POS(5,-10)\ar@{->>} (11.5,0.5)
\POS(5,-10)\ar@{->>} (7.5,4)
\POS(-3,-3)\ar@{->>} (2.5,4)

\endxy
}}

\def \fivegonarrowtBCb#1#2#3#4#5{{
\xy
\POS(10,4) \ar@{=} +(-10,0)_{#5}
\ar@{-} +(3,-7)^{#4}
\POS(0,4) \ar@{-} +(-3,-7)_{#1}
\POS(0,4) *+{\bullet}
\POS(-3,-3) \ar@{-} +(8,-7)_{#2}
\POS(5,-10)\ar@{-} +(8,7)_{#3}

\POS(5,-10)\ar@{->>} (11.5,0.5)
\POS(5,-10)\ar@{->>} (7.5,4)
\POS(0,4)\ar@{->>} (1,-6.5)

\endxy
}}

\def \fivegonarrowtBCc#1#2#3#4#5{{
\xy
\POS(10,4) \ar@{=} +(-10,0)_{#5}
\ar@{-} +(3,-7)^{#4}
\POS(0,4) \ar@{-} +(-3,-7)_{#1}
\POS(0,4) *+{\bullet}
\POS(-3,-3) \ar@{-} +(8,-7)_{#2}
\POS(5,-10)\ar@{-} +(8,7)_{#3}

\POS(5,-10)\ar@{->>} (5,4)
\POS(5,-10)\ar@{->>} (-1.5,0.5)
\POS(10,4)\ar@{->>} (9,-6.5)

\endxy
}}

\def \fivegonarrowtBCd#1#2#3#4#5{{
\xy
\POS(10,4) \ar@{=} +(-10,0)_{#5}
\ar@{-} +(3,-7)^{#4}
\POS(0,4) \ar@{-} +(-3,-7)_{#1}
\POS(0,4) *+{\bullet}
\POS(-3,-3) \ar@{-} +(8,-7)_{#2}
\POS(5,-10)\ar@{-} +(8,7)_{#3}

\POS(5,-10)\ar@{->>} (2.5,4)
\POS(5,-10)\ar@{->>} (-1.5,0.5)
\POS(13,-3)\ar@{->>} (7.5,4)

\endxy
}}

\def \fivegonarrowtBCe#1#2#3#4#5{{
\xy
\POS(10,4) \ar@{=} +(-10,0)_{#5}
\ar@{-} +(3,-7)^{#4}
\POS(0,4) \ar@{-} +(-3,-7)_{#1}
\POS(0,4) *+{\bullet}
\POS(-3,-3) \ar@{-} +(8,-7)_{#2}
\POS(5,-10)\ar@{-} +(8,7)_{#3}

\POS(5,-10)\ar@{->>} (12.25,-1.25)
\POS(5,-10)\ar@{->>} (-1.5,0.5)
\POS(0,4)\ar@{->>} (10.75,2.25)

\endxy
}}

\def \fivegonarrowtBCf#1#2#3#4#5{{
\xy
\POS(10,4) \ar@{=} +(-10,0)_{#5}
\ar@{-} +(3,-7)^{#4}
\POS(0,4) \ar@{-} +(-3,-7)_{#1}
\POS(0,4) *+{\bullet}
\POS(-3,-3) \ar@{-} +(8,-7)_{#2}
\POS(5,-10)\ar@{-} +(8,7)_{#3}

\POS(5,-10)\ar@{->>} (11.5,0.5)
\POS(5,-10)\ar@{->>} (-2.25,-1.25)
\POS(10,4)\ar@{->>} (-0.75,2.25)

\endxy
}}

\def \fivegonarrowtBBa#1#2#3#4#5{{
\xy
\POS(10,4) \ar@{=} +(-10,0)_{#5}
\ar@{-} +(3,-7)^{#4}
\POS(0,4) \ar@{-} +(-3,-7)_{#1}
\POS(0,4) *+{\bullet}
\POS(-3,-3) \ar@{-} +(8,-7)_{#2}
\POS(5,-10)\ar@{-} +(8,7)_{#3}

\POS(13,-3)\ar@{->>} (5,4)
\POS(13,-3)\ar@{->>} (-1.5,0.5)
\POS(-3,-3)\ar@{->>} (9,-6.5)

\endxy
}}

\def \fivegonarrowtBBb#1#2#3#4#5{{
\xy
\POS(10,4) \ar@{=} +(-10,0)_{#5}
\ar@{-} +(3,-7)^{#4}
\POS(0,4) \ar@{-} +(-3,-7)_{#1}
\POS(0,4) *+{\bullet}
\POS(-3,-3) \ar@{-} +(8,-7)_{#2}
\POS(5,-10)\ar@{-} +(8,7)_{#3}

\POS(13,-3)\ar@{->>} (5,4)
\POS(13,-3)\ar@{->>} (-0.75,2.25)
\POS(5,-10)\ar@{->>} (-2.25,-1.25)

\endxy
}}

\def \fivegonarrowtBBc#1#2#3#4#5{{
\xy
\POS(10,4) \ar@{=} +(-10,0)_{#5}
\ar@{-} +(3,-7)^{#4}
\POS(0,4) \ar@{-} +(-3,-7)_{#1}
\POS(0,4) *+{\bullet}
\POS(-3,-3) \ar@{-} +(8,-7)_{#2}
\POS(5,-10)\ar@{-} +(8,7)_{#3}

\POS(13,-3)\ar@{->>} (1,-6.5)
\POS(13,-3)\ar@{->>} (-1.5,0.5)
\POS(0,4)\ar@{->>} (11.5,0.5)

\endxy
}}

\def \fivegonarrowtBBd#1#2#3#4#5{{
\xy
\POS(10,4) \ar@{=} +(-10,0)_{#5}
\ar@{-} +(3,-7)^{#4}
\POS(0,4) \ar@{-} +(-3,-7)_{#1}
\POS(0,4) *+{\bullet}
\POS(-3,-3) \ar@{-} +(8,-7)_{#2}
\POS(5,-10)\ar@{-} +(8,7)_{#3}

\POS(13,-3)\ar@{->>} (1,-6.5)
\POS(13,-3)\ar@{->>} (-2.25,-1.25)
\POS(10,4)\ar@{->>} (-0.75,2.25)

\endxy
}}

\def \fivegonarrowtBBe#1#2#3#4#5{{
\xy
\POS(10,4) \ar@{=} +(-10,0)_{#5}
\ar@{-} +(3,-7)^{#4}
\POS(0,4) \ar@{-} +(-3,-7)_{#1}
\POS(0,4) *+{\bullet}
\POS(-3,-3) \ar@{-} +(8,-7)_{#2}
\POS(5,-10)\ar@{-} +(8,7)_{#3}

\POS(13,-3)\ar@{->>} (7.5,4)
\POS(13,-3)\ar@{->>} (1,-6.5)
\POS(-3,-3)\ar@{->>} (2.5,4)

\endxy
}}

\def \fivegonarrowtBBf#1#2#3#4#5{{
\xy
\POS(10,4) \ar@{=} +(-10,0)_{#5}
\ar@{-} +(3,-7)^{#4}
\POS(0,4) \ar@{-} +(-3,-7)_{#1}
\POS(0,4) *+{\bullet}
\POS(-3,-3) \ar@{-} +(8,-7)_{#2}
\POS(5,-10)\ar@{-} +(8,7)_{#3}

\POS(13,-3)\ar@{->>} (5,4)
\POS(13,-3)\ar@{->>} (3,-8.25)
\POS(0,4)\ar@{->>} (-1,-4.75)

\endxy
}}

\def \fivegonarrowtBAa#1#2#3#4#5{{
\xy
\POS(10,4) \ar@{=} +(-10,0)_{#5}
\ar@{-} +(3,-7)^{#4}
\POS(0,4) \ar@{-} +(-3,-7)_{#1}
\POS(0,4) *+{\bullet}
\POS(-3,-3) \ar@{-} +(8,-7)_{#2}
\POS(5,-10)\ar@{-} +(8,7)_{#3}

\POS(0,4)\ar@{->>} (1,-6.5)
\POS(0,4)\ar@{->>} (7,-8.25)
\POS(10,4)\ar@{->>} (11,-4.75)

\endxy
}}

\def \fivegonarrowtBAb#1#2#3#4#5{{
\xy
\POS(10,4) \ar@{=} +(-10,0)_{#5}
\ar@{-} +(3,-7)^{#4}
\POS(0,4) \ar@{-} +(-3,-7)_{#1}
\POS(0,4) *+{\bullet}
\POS(-3,-3) \ar@{-} +(8,-7)_{#2}
\POS(5,-10)\ar@{-} +(8,7)_{#3}

\POS(0,4)\ar@{->>} (1,-6.5)
\POS(0,4)\ar@{->>} (9,-6.5)
\POS(13,-3)\ar@{->>} (5,4)

\endxy
}}

\def \fivegonarrowtBAc#1#2#3#4#5{{
\xy
\POS(10,4) \ar@{=} +(-10,0)_{#5}
\ar@{-} +(3,-7)^{#4}
\POS(0,4) \ar@{-} +(-3,-7)_{#1}
\POS(0,4) *+{\bullet}
\POS(-3,-3) \ar@{-} +(8,-7)_{#2}
\POS(5,-10)\ar@{-} +(8,7)_{#3}

\POS(0,4)\ar@{->>} (11.5,0.5)
\POS(0,4)\ar@{->>} (11,-4.75)
\POS(-3,-3)\ar@{->>} (7,-8.25)

\endxy
}}

\def \fivegonarrowtBAd#1#2#3#4#5{{
\xy
\POS(10,4) \ar@{=} +(-10,0)_{#5}
\ar@{-} +(3,-7)^{#4}
\POS(0,4) \ar@{-} +(-3,-7)_{#1}
\POS(0,4) *+{\bullet}
\POS(-3,-3) \ar@{-} +(8,-7)_{#2}
\POS(5,-10)\ar@{-} +(8,7)_{#3}

\POS(0,4)\ar@{->>} (11.5,0.5)
\POS(0,4)\ar@{->>} (9,-6.5)
\POS(5,-10)\ar@{->>} (-1.5,0.5)

\endxy
}}

\def \fivegonarrowtBAe#1#2#3#4#5{{
\xy
\POS(10,4) \ar@{=} +(-10,0)_{#5}
\ar@{-} +(3,-7)^{#4}
\POS(0,4) \ar@{-} +(-3,-7)_{#1}
\POS(0,4) *+{\bullet}
\POS(-3,-3) \ar@{-} +(8,-7)_{#2}
\POS(5,-10)\ar@{-} +(8,7)_{#3}

\POS(0,4)\ar@{->>} (10.75,2.25)
\POS(0,4)\ar@{->>} (1,-6.5)
\POS(5,-10)\ar@{->>} (12.25,-1.25)

\endxy
}}

\def \fivegonarrowtBAf#1#2#3#4#5{{
\xy
\POS(10,4) \ar@{=} +(-10,0)_{#5}
\ar@{-} +(3,-7)^{#4}
\POS(0,4) \ar@{-} +(-3,-7)_{#1}
\POS(0,4) *+{\bullet}
\POS(-3,-3) \ar@{-} +(8,-7)_{#2}
\POS(5,-10)\ar@{-} +(8,7)_{#3}

\POS(0,4)\ar@{->>} (11.5,0.5)
\POS(0,4)\ar@{->>} (-1,-4.75)
\POS(13,-3)\ar@{->>} (3,-8.25)

\endxy
}}

\def \fivegonarrowtAa#1#2#3#4#5{{
\xy
\POS(10,4) \ar@{=} +(-10,0)_{#5}
\ar@{-} +(3,-7)^{#4}
\POS(0,4) \ar@{-} +(-3,-7)_{#1}
\POS(0,4) *+{\bullet}
\POS(-3,-3) \ar@{-} +(8,-7)_{#2}
\POS(5,-10)\ar@{-} +(8,7)_{#3}

\POS(0,4)\ar@{->>} (11.5,0.5)
\POS(0,4)\ar@{->>} (9,-6.5)
\POS(0,4)\ar@{->>} (1,-6.5)

\endxy
}}

\def \fivegonarrowtAb#1#2#3#4#5{{
\xy
\POS(10,4) \ar@{=} +(-10,0)_{#5}
\ar@{-} +(3,-7)^{#4}
\POS(0,4) \ar@{-} +(-3,-7)_{#1}
\POS(0,4) *+{\bullet}
\POS(-3,-3) \ar@{-} +(8,-7)_{#2}
\POS(5,-10)\ar@{-} +(8,7)_{#3}

\POS(10,4)\ar@{->>} (-1.5,0.5)
\POS(10,4)\ar@{->>} (9,-6.5)
\POS(10,4)\ar@{->>} (1,-6.5)

\endxy
}}

\def \fivegonarrowtAc#1#2#3#4#5{{
\xy
\POS(10,4) \ar@{=} +(-10,0)_{#5}
\ar@{-} +(3,-7)^{#4}
\POS(0,4) \ar@{-} +(-3,-7)_{#1}
\POS(0,4) *+{\bullet}
\POS(-3,-3) \ar@{-} +(8,-7)_{#2}
\POS(5,-10)\ar@{-} +(8,7)_{#3}

\POS(13,-3)\ar@{->>} (-1.5,0.5)
\POS(13,-3)\ar@{->>} (1,-6.5)
\POS(13,-3)\ar@{->>} (5,4)

\endxy
}}

\def \fivegonarrowtAd#1#2#3#4#5{{
\xy
\POS(10,4) \ar@{=} +(-10,0)_{#5}
\ar@{-} +(3,-7)^{#4}
\POS(0,4) \ar@{-} +(-3,-7)_{#1}
\POS(0,4) *+{\bullet}
\POS(-3,-3) \ar@{-} +(8,-7)_{#2}
\POS(5,-10)\ar@{-} +(8,7)_{#3}

\POS(5,-10)\ar@{->>} (-1.5,0.5)
\POS(5,-10)\ar@{->>} (11.5,0.5)
\POS(5,-10)\ar@{->>} (5,4)

\endxy
}}

\def \fivegonarrowtAe#1#2#3#4#5{{
\xy
\POS(10,4) \ar@{=} +(-10,0)_{#5}
\ar@{-} +(3,-7)^{#4}
\POS(0,4) \ar@{-} +(-3,-7)_{#1}
\POS(0,4) *+{\bullet}
\POS(-3,-3) \ar@{-} +(8,-7)_{#2}
\POS(5,-10)\ar@{-} +(8,7)_{#3}

\POS(-3,-3)\ar@{->>} (9,-6.5)
\POS(-3,-3)\ar@{->>} (11.5,0.5)
\POS(-3,-3)\ar@{->>} (5,4)

\endxy
}}

\newcommand{\rd}{d}

\newcommand{\cS}{\begin{cal}S\end{cal}}
\newcommand{\cR}{\begin{cal}R\end{cal}}
\newcommand{\cT}{\begin{cal}T\end{cal}}
\renewcommand{\ln}{\log}

\def\bit#1\eit{\begin{itemize}#1\end{itemize}}
\def\ben#1\een{\begin{enumerate}#1\end{enumerate}}

\newenvironment{sloppyequation}[0]{\sloppy\begin{flushleft}\hspace*{0.75cm}\(}{\)\end{flushleft}\fussy}
\newenvironment{sloppytext}[0]{\sloppy\begin{flushleft}}{\end{flushleft}\fussy}

\newcommand{\beqsloppy}{\begin{sloppyequation}}
\newcommand{\eeqsloppy}{\end{sloppyequation}}
\newcommand{\btxtsloppy}{\begin{sloppytext}}
\newcommand{\etxtsloppy}{\end{sloppytext}}

\newtheorem{defi}{Definition}

\newtheorem{proposition}{Proposition}

\title{
From polygons and symbols to polylogarithmic functions}

\author{Claude~Duhr\\
Institute for Particle Physics Phenomenology,
University of Durham\\ Durham, DH1 3LE, U.K.\\
and\\
Institut f\"ur theoretische Physik, ETH Z\"urich,\\ Wolfgang-Paulistr. 27, CH-8093, Switzerland\\
E-mail:~\email{duhrc@itp.phys.ethz.ch}}

\author{Herbert~Gangl\\
Department of Mathematical Sciences,
University of Durham\\ Durham, DH1 3LE, U.K.\\
E-mail:~\email{herbert.gangl@durham.ac.uk}}

\author{John~R.~Rhodes\\
Department of Mathematical Sciences,
University of Durham\\ Durham, DH1 3LE, U.K.\\
E-mail:~\email{j.r.rhodes@durham.ac.uk}}

\abstract{We present a review of the symbol map, a mathematical tool that can be useful in simplifying expressions among multiple polylogarithms, and recall its main properties. 
A recipe is given for how to obtain the symbol of  a multiple polylogarithm in terms of the combinatorial properties of an associated  rooted decorated polygon. We also outline a systematic approach to constructing a function corresponding to 
a given symbol, and illustrate it in the particular case of harmonic polylogarithms up to weight~four. Furthermore, part of the ambiguity of this process is highlighted by exhibiting a family of non-trivial elements in the kernel of the symbol map for arbitrary weight.}

\preprint{IPPP/11/56\\
DCPT/11/112}

\begin{document}

\catcode`\@=11
\font\manfnt=manfnt
\def\Watchout{\@ifnextchar [{\W@tchout}{\W@tchout[1]}}
\def\W@tchout[#1]{{\manfnt\@tempcnta#1\relax%
  \@whilenum\@tempcnta>\z@\do{%
    \char"7F\hskip 0.3em\advance\@tempcnta\m@ne}}}
\let\foo\W@tchout
\def\dubious{\@ifnextchar[{\@dubious}{\@dubious[1]}}
\let\enddubious\endlist
\def\@dubious[#1]{%
  \setbox\@tempboxa\hbox{\@W@tchout#1}
  \@tempdima\wd\@tempboxa
  \list{}{\leftmargin\@tempdima}\item[\hbox to 0pt{\hss\@W@tchout#1}]}
\def\@W@tchout#1{\W@tchout[#1]}
\catcode`\@=12


\def \sha{{\,\amalg\hskip -3.6pt\amalg\,}}
\def \uplus{\sha}
\def \mmu#1{{\mu\Big(\hskip -8pt #1\hskip -5pt\Big)}}


\section{Introduction}
\label{sec:intro}
Polylogarithms and their multivariable generalizations~\cite{Goncharov:1998, Goncharov:2001} play an equally important role in modern mathematics and in physics. In mathematics they occur for instance in connection with algebraic $K$-theory and mixed Tate motives, e.g.~\cite{Bloch-Irvine, Beilinson-Cyclotomic, Zagier-Texel, Deligne-Beilinson, Goncharov-Advances, Goncharov-Galois}, with Hilbert's third problem (on scissors congruences), e.g.~\cite{Dupont-SahII, Cathelineau-Hilbert, Goncharov-Volumes}, as volume functions for hyperbolic spaces, e.g.~\cite{Boehm-Book, Neumann-Zagier, Zagier-Dilog, Zagier-Inv, Kellerhals-Volumes, Goncharov-Volumes, Neumann-Yang}, {and} are also related to characteristic classes, e.g.~\cite{Gelfand-Losik}, special values of $L$-functions in algebraic number theory, e.g.~\cite{Zagier-Dilog, Zagier-Texel, Goncharov-Advances}, algebraic cycles, e.g.~\cite{Bloch-Kriz, GSM, GGL:2009} or, in the form of iterated integrals, in algebraic topology, e.g.~\cite{ChenBullAMS, Hain, Wojtkowiak}. In physics, the computation of higher order corrections to physical observables requires the analytical evaluation of Feynman integrals that can generally be expressed in terms of (special classes of) multiple polylogarithms, e.g.~\cite{Remiddi:1999ew,Gehrmann:2000zt,Ablinger:2011te,Vermaseren:2005qc,Moch:2004xu,Vogt:2004mw,Moch:2004pa,Bonciani:2003hc,Bernreuther:2004ih,Bernreuther:2004th,Bernreuther:2005rw,Mastrolia:2003yz,Bonciani:2004qt,Bonciani:2004gi,Czakon:2004wm,Bern:2006ew,Heinrich:2004iq,Smirnov:2001cm,Bork:2010wf,Henn:2010ir,Aglietti:2006tp,Aglietti:2004ki,Aglietti:2004nj,Gehrmann:2001ck,Anastasiou:2006hc,Moch:2002hm,Moch:2002an,Aglietti:2008fe,DelDuca:2009ac, Gehrmann:2001pz,Maitre:2005uu,Maitre:2007kp,Vollinga:2004sn}.
While in all of these applications it would be desirable to have a minimal spanning set---``basis functions" in physics parlance---for
the polylogarithmic expressions involved in a given problem, it is well known that these latter functions satisfy various intricate functional equations among themselves, making the question of how to find a minimal spanning set very hard to answer in general. As a consequence, seemingly complicated results, say for a Feynman integral, may admit a much shorter analytic representation, the simplicity of the answer being hidden due to the existence of an abundance of functional equations among these functions. There is thus  a strong interest for a better understanding of the functional equations among multiple polylogarithms, both from a formal mathematical standpoint and in view of practical applications in physics.

A way to approach functional equations among (multiple) polylogarithms is provided by the so-called \emph{symbol map}, a linear map that associates to each multiple polylogarithm of weight $n$ an element in the $n$-fold tensor power of some vector space of one-forms. The virtue of the symbol map is that it captures to a good extent the main combinatorial and analytical properties of certain transcendental functions, and in particular it is expected that all functional equations among multiple polylogarithms are in the kernel of the symbol map. Loosely speaking, this means that a necessary condition for two expressions written in terms of multiple polylogarithms to be equal modulo functional equations is that they have the same symbol, a condition that is usually much easier to check than proving equality at the level of the functions. The inverse problem (sometimes called \emph{integration} of a symbol) of finding a function whose symbol  matches a given tensor satisfying a certain integrability condition is much harder and we know of no general algorithm  to construct such a function.

While special cases of the symbol map have been profitably used by mathematicians for over two decades (for example in connection with functional equations see e.g. refs.~\cite{Zagier-Texel, Goncharov-Advances, Zagier-Appendix, Gangl-Selecta}), it has only very recently been introduced into physics in the context of the $\begin{cal}N\end{cal}=4$ Super Yang-Mills (SYM) theory in ref.~\cite{Goncharov:2010jf}, where it was applied to greatly simplify the analytic expression for the two-loop six-point remainder function obtained in ref.~\cite{DelDuca:2009au,DelDuca:2010zg}. In the wake of that work, the symbol map has seen various applications, mostly in the context of $\begin{cal}N\end{cal}=4$ SYM. In particular, by now the symbols of all two-loop remainder functions are completely known~\cite{CaronHuot:2011ky}, while at three loops the symbols of the remainder functions for the hexagon in general kinematics~\cite{Dixon:2011pw} and for the octagon in special kinematics~\cite{Heslop:2011hv} are known up to some free parameters that could not be fixed from general considerations. However, only in the latter octagon case an integrated form of the symbol is {also} known. Other approaches, aiming at the determination of the symbol of loop amplitudes by exploiting the operator product expansion in the collinear limit~\cite{Alday:2010ku,Gaiotto:2011dt} or the relationship between Feynman integrals and the volumes of polyhedra in non-euclidean spaces~\cite{Davydychev:1997wa,Spradlin:2011wp}, have also been considered. Furthermore, the symbol map was recently used to obtain compact analytic expressions for certain one-loop hexagon integrals in $D=6$ dimensions~\cite{Dixon:2011ng,DelDuca:2011ne,DelDuca:2011jm,DelDuca:2011wh}. More phenomenological applications, as for example in ref.~\cite{Buehler:2011ev}, have also been considered.

The aim of this paper is twofold: While the symbol map has already been extensively used in the $\begin{cal}N\end{cal}=4$ SYM community in physics, it seems still rather little known in other areas of physics in which the computation of Feynman integrals plays an important role. On the one hand, we therefore present a concise review on this topic, putting special emphasis on how to apply the symbol map to obtain simpler or shorter analytic results for functions arising from {certain} Feynman integrals. On the other hand, we believe that our work goes beyond the existing literature on the subject in various aspects. While so far the symbol of a transcendental function was defined recursively by considering iterated differentials, we introduce a simple diagrammatic rule that allows to directly read off the symbol from all possible ``triangulations'' of a certain decorated polygon associated to a multiple polylogarithm~\cite{GGL:2009}. Furthermore, we also address the problem of how to integrate a symbol to a function by presenting an effective approach to construct a candidate spanning set of functions in terms of which the symbol might be integrated.

The structure of the paper is as follows: In section~\ref{sec:polylogs} we give a short review of multiple polylogarithms and of their properties. In section~\ref{sec:tensors} we review the main properties of the symbol  map and we show how to obtain the symbol of a multiple polylogarithm as the weighted sum of all possible maximal dissections of a certain decorated polygon associated to the polylogarithm. In section~\ref{sec:example} we give a short example of how to integrate a symbol following the approach introduced in refs.~\cite{Goncharov:2010jf,Alday:2010jz}, before generalizing this procedure to higher weights in section~\ref{sec:higher_weights}. In order to highlight remaining difficulties and ambiguities when trying to integrate to a function, we also give a family of non-trivial elements in the kernel of the symbol map.  We illustrate these concepts in section~\ref{sec:HPLbasis} where we apply them to derive a spanning set up to weight four for a special class of multiple polylogarithms, the so-called harmonic polylogarithms~\cite{Remiddi:1999ew}. The appendices contain a summary of the mathematical notions used throughout the paper, as well as some technical details and proofs left out in the main text. We also include an appendix with a collection of symbols for multiple polylogarithms up to weight four.

{\bf Remark:} The authors wish the reader to be aware that this paper contains the work of both physicists and mathematicians. As a consequence, it should be noted that the paper has been written to try to accommodate 
the language of both communities. We have tried to find a compromise in the level of details of the paper, and so while some arguments may go deeper than felt necessary by some readers, the text may be too sketchy at times for others.

\newpage
\section{Short review of multiple polylogarithms}
\label{sec:polylogs}

{\bf Definition.} Multiple polylogarithms can be defined recursively, for $n\geq 0$, via the iterated integral~\cite{Goncharov:1998, Goncharov:2001}
 \beq\label{eq:Mult_PolyLog_def}
 G(a_1,\ldots,a_n;x)=\,\int_0^x\,{\rd t\over t-a_1}\,G(a_2,\ldots,a_n;t)\,,\\
\eeq
with $G(x) = G(;x)=1$, an exception being when $x=0$ in which case we put $G(0)=0$ (clearly any expression $\int_0^0\dots$ should be zero), and with $a_i\in \mathbb{C}$ are chosen constants and $x$ is a complex variable. In the following, we will also consider $G(a_1,\ldots,a_n;x)$ to be functions of $a_1, \ldots, a_n$.  In the special case where all the $a_i$'s are zero, we define, using the obvious vector notation $\vec a_n=(\underbrace{a,\dots,a}_{n})$, $a\in \C$,
\beq
G(\vec 0_n;x) = {1\over n!}\,\ln^n x\,,
\eeq
consistent with the case $n=0$ above. Note that, while in the Mathematics literature these functions appear already in the early 20th century in the works of Poincar\'e and of Lappo-Danilevsky \cite{Lap35} as ``hyperlogarithms", as well as in the 1960's in  Chen's work on iterated integrals (e.g., \cite{ChenBullAMS})\footnote{In a sense, they already made an appearance in Kummer's pioneering work \cite{Kummer} in 1840.}, in the physics literature these functions
are often called {\em Goncharov polylogarithms}, due to the wealth of structure that the latter has established for them over the last 20 years. Throughout this paper, we follow the physics  convention for the definition of the iterated integrals, which differs slightly from the mathematical one; e.g., in ref.~\cite{Goncharov:2001}, the function corresponding to $G(a_1,\dots,a_n;x)$ would be denoted $I(0;a_n,\dots,a_1;x)$, i.e., with the reverse order of the $a_i$ but keeping the same variable $x$. 

We will refer to the vector $\vec a=(a_1, \ldots, a_n)$ as the {\em vector of singularities}  
attached to the multiple polylogarithm and the number of elements $n$, counted with multiplicities, in that vector is called the {\em weight} of the multiple polylogarithm.

{\bf Properties.} We collect here a number of useful and well-known properties (cf. e.g. ref.~\cite{Goncharov:2001,Goncharov-Galois}). Iterated integrals form a {\em  shuffle algebra}~\cite{Ree:1958} (see appendix~\ref{app:algebras} for a short review of shuffle algebras), which allows one to express the product of two multiple polylogarithms of weight $n_1$ and $n_2$ as a linear combination with integer coefficients of multiple polylogarithms of weight $n_1+n_2$, via
  \beq\bsp\label{eq:G_shuffle}
  G(a_1,\ldots,a_{n_1};x) \, G(a_{n_1+1},\ldots,a_{n_1+n_2};x) &\,=\sum_{\sigma\in\Sigma(n_1, n_2)}\,G(a_{\sigma(1)},\ldots,a_{\sigma(n_1+n_2)};x),\\
      \esp\eeq
where $\Sigma(n_1,n_2)$ denotes the set of all shuffles of $n_1+n_2$ elements, i.e., the subset of the symmetric group $S_{n_1+n_2}$ defined by (cf. ref.~\cite{ChenBullAMS}, eq.~(1.5.6))
\beq\label{eq:Sigma_def}
\Sigma(n_1,n_2) = \{\sigma\in S_{n_1+n_2} |\, \sigma^{-1}(1)<\ldots<\sigma^{-1}({n_1}) {\rm~~and~~} \sigma^{-1}(n_1+1)<\ldots<\sigma^{-1}(n_1+{n_2})\}\,.
\eeq
The algebraic properties of multiple polylogarithms imply that not all the $G(\vec a;x)$ for fixed $x$ are independent, but that there are (polynomial) {\em relations} among them.  In particular, we can reduce them, modulo products of lower weight functions, to functions whose rightmost index of all the vectors of singularities is non-zero (apart from objects of the form $G(\vec 0_n;x)$), e.g., 
\beq\bsp
G(a,0,0;x) &\,= G(0,0;x)\,G(a;x) - G(0,0,a;x) - G(0,a,0;x)\\
&\, = G(0,0;x)\,G(a;x) - G(0,0,a;x) - \left(G(0,a;x)\,G(0;x) - 2G(0,0,a;x)\right)\\
&\,= G(0,0;x)\,G(a;x) + G(0,0,a;x) -G(0,a;x)\,G(0;x)\,,
\esp\eeq
where the middle summand is of the desired form (and the remaining summands are products).

If the (rightmost) index $a_n$ of $\vec a$  is non-zero, then the function $G(\vec a;x)$ is {\em invariant under a rescaling} of all its arguments, i.e., for any $k\in \mathbb{C}^*$ we have
\beq\label{eq:Gscaling}
G(k\,\vec a;k\, x) = G(\vec a;x)\qquad (a_n\neq 0) \,.
\eeq
Furthermore, multiple polylogarithms satisfy the {\em H\"older convolution}~\cite{BBBL}, i.e., whenever $a_1\neq 1$ and $a_n\neq0$, we have, $\forall p\in\mathbb{C}^*$,
\beq\label{eq:Hoelder}
G(a_1,\ldots, a_n;1) = \sum_{k=0}^n(-1)^k\,G\left(1-a_k,\ldots, 1-a_1;1-{1\over p}\right)\,G\left(a_{k+1},\ldots, a_n;{1\over p}\right)\,.
\eeq
Below in section~\ref{sec:higher_weights} we will be particularly interested in the limiting case $p\to \infty$ of this identity,
\beq\label{eq:Hoelder_inf}
G(a_1,\ldots, a_n;1) = (-1)^n\,G\left(1-a_n,\ldots, 1-a_1;1\right)\,.
\eeq

Whenever they converge, multiple polylogarithms can equally well be represented \cite{Goncharov:1998} as {\em multiple nested sums} (e.g., for $|x_i|<1$)
\beq\label{eq:Lim_def}
\textrm{Li}_{m_1,\ldots,m_k}(x_1,\ldots,x_k) = \sum_{n_1<n_2<\dots <n_k} \frac{x_1^{n_1} x_2^{n_2} \cdots x_k^{n_k} }{n_1^{m_1} n_2^{m_2} \cdots n_k^{m_k} } =
 \sum_{n_k=1}^\infty{x_k^{n_k}\over n_k^{m_k}}\,\sum_{n_{k-1}=1}^{n_k-1}\ldots\sum_{n_1=1}^{n_{2}-1}{x_1^{n_1}\over n_1^{m_1}}\,.
\eeq
Note that we are using Goncharov's original summation convention \cite{Goncharov:1998}; other authors define $\textrm{Li}_{m_1,\ldots,m_k}(x_1,\ldots,x_k)$  using the reverse summation convention instead, i.e.~$n_1>\dots>n_k$.

The $G$ and Li functions define in fact the same class of functions and are related by
\beq\label{eq:Gm_def}
\textrm{Li}_{m_1,\ldots,m_k}(x_1,\ldots,x_k) = (-1)^k\,G_{m_k,\ldots,m_1}\left({1\over x_k}, \ldots, {1\over x_1\ldots x_k}\right)\,,
\eeq
(note the reverse order of the indices for $G$) where we used the notation
\beq
G_{m_1,\ldots,m_k}\left(t_1, \ldots, t_k\right)= G(\underbrace{0,\ldots,0}_{m_1-1},{t_1}, \ldots, \underbrace{0,\ldots,0}_{m_k-1}, {t_k};1)\,.
\eeq

It is possible to find {\em closed expressions} for (very few) special classes of multiple polylogarithms, for arbitrary weight, in terms of classical polylogarithm functions, e.g., for $a\neq 0$ we have
\beq\bsp
G(\vec 0_n;x) = {1\over n!}\,\ln^nx, \qquad &G(\vec a_n;x) = {1\over n!}\,\ln^n\left(1-{x\over a}\right),\\
G(\vec 0_{n-1},a;x) = -\textrm{Li}_n\left({x\over a}\right), \qquad & G(\vec 0_{n},\vec a_{p};x) = (-1)^p\,S_{n,p}\left({x\over a}\right)\,,
\esp\eeq
where $S_{n,p}$ denotes the Nielsen polylogarithm~\cite{Nielsen}.
Moreover, up to weight three, multiple polylogarithms are well-known to be 
expressible in terms of ordinary logarithms, dilogarithms and trilogarithms (cf. ref.~\cite{Lewin-Book}, \S8.4.3, implicitly, as well as refs.~\cite{Goncharov-Volumes, Kellerhals-Volumes}). In particular, if $a$ and $b$ are non-zero and different, we find
\beq\label{eq:GabzLi}
G(a,b;x) = \text{Li}_2\left(\frac{b-x}{b-a}\right)-\text{Li}_2\left(\frac{b}{b-a}\right)+\log \left(1-\frac{x}{b}\right) \log \left(\frac{x-a}{b-a}\right)\,.
\eeq

{\bf Aim.} The aim of this paper is to present an algorithmic approach how to deal with---in fact rather to circumvent---the complicated functional equations that relate multiple polylogarithms, and how to find, given a choice of certain singularities $a_i$, a (possibly minimal) spanning set for functions in which to express multiple polylogarithms with singularities only in these $a_i$, provided such a set exists. The approach we present is 
rather
generic and can be applied to any expression involving multiple polylogarithms. This is made possible by using results closely related to work of Goncharov, Spradlin, Vergu and Volovich~\cite{Goncharov:2010jf}, which in turn was inspired by the theory of (mixed Tate) motives\footnote{Let us point out that this is far from being the first exhibit of a direct connection between mixed Tate motives and mathematical physics, as such a relationship  has been explored, e.g., by Kreimer in work with Bloch and  Esnault~\cite{BlochEsnaultKreimer, BlochKreimer}, such a connection was clearly apparent from letter correspondence between Broadhurst and Deligne~\cite{BroadhurstDeligne} resulting e.g.~in ref.~\cite{Broadhurst}, work of Belkale--Brosnan~\cite{BelkaleBrosnan} or more recently by Brown~\cite{BrownMassless} and others. One should also mention work of Connes and Marcolli~\cite{ConnesMarcolli} in this direction.}, and in particular by using a certain tensor calculus associated to iterated integrals, which is called ``symbol calculus'' in the following (the name ``symbol" originating from \cite{Goncharov:2010jf}~and from ref.~\cite{Goncharov-simple-Grassmannian}), and which we will review in the next section.

An important remark is that the construction of a ``symbol" seems to be a rather special case of a very general construction by Chen~\cite{ChenBullAMS},  where it appears as the image of an iterated integral as a 0-cocycle in the so-called ``bar construction" attached to, say, $X$ equal to the projective line minus a number of points (more generally, the construction has been investigated for a hyperplane configuration~\cite{BrownThesis}, \S3), and it lands in the $n$-fold tensor product of the vector space of 1-forms on the underlying space $X$.
Moreover, Chen characterised the image as the formal words in these 1-forms satisfying a natural integrability condition. Therefore, it would seem appropriate to reflect this in the notation for this object, e.g.~as ``Chen symbol".
A polygon attached to an iterated integral enjoys the useful property that it gives a very concise way of explicitly producing integrable words, i.e. (Chen) symbols, of that kind.

As an application, we restrict ourselves in section~\ref{sec:HPLbasis} to a specific subclass of multiple polylogarithms that are of particular importance in applications in high-energy physics. These so-called \emph{harmonic} polylogarithms (HPL's)  $H(\vec a; x)$ were first singled out and thoroughly studied in ref.~\cite{Remiddi:1999ew}. HPL's correspond to a special case of the iterated integral defined in eq.~\eqref{eq:Mult_PolyLog_def} where $a_i\in\{-1,0,1\}$. More precisely, they are defined via
\beq
H(\vec a;x) = (-1)^k\,G(\vec a;x)\,, \quad a_i\in \{-1,0,1\}\,,
\eeq
where $k$ is the number of elements in $\vec a$ equal to $(+1)$. Many one and two-loop
Feynman integrals can be expressed in terms of HPL's up to weight four and generalizations thereof~\cite{Gehrmann:2000zt, Ablinger:2011te}. As harmonic polylogarithms are just a special case of the multiple polylogarithms introduced at the beginning of this section, all HPL's through weight three can be expressed through classical polylogarithms. By contrast, similar to the general case of multiple polylogarithms, it is expected that HPL's of weight $\geq 4$ are no longer expressible in terms of classical ones alone. In section~\ref{sec:HPLbasis}  we illustrate our technique by constructing a spanning set of harmonic polylogarithms in weight~4.

\section{Symbols and polygons}
\label{sec:tensors}

The differential structure of multiple polylogarithms can be captured very well combinatorially using a certain kind of decorated polygons with some additional structure, as developed in ref.~\cite{GGL:2009}, where they were called {\em $R$-deco polygons}. We note that there are related notions that had occurred previously in Goncharov's work, e.g. in refs.~\cite{Goncharov-Galois, Goncharov-Duke}. There is an algebraic object  
attached to such a polygon, and hence to the corresponding multiple polylogarithm. This object, which has been dubbed a {\em symbol}  in ref.~\cite{Goncharov:2010jf},  is an element in a tensor power of a certain vector space and  contains a lot of information about the original function.

\subsection{An example in a nutshell}\label{nutshell} In this subsection we give a quick idea of how, following ref.~\cite{GGL:2009}, one can associate 
to a multiple polylogarithm---or rather to an associated rooted decorated polygon---its symbol (we show in section~\ref{rules} that this definition is equivalent to the definition given in ref.~\cite{Goncharov:2010jf}). In the following subsection we then give a more detailed account of the construction.

A multiple polylogarithm of weight $n$ gives rise to a certain $(n+1)$-gon.
As a ``foreshadowing" example, we first give the 4-gon $P=P(c,b,a,x)$ attached to some weight 3 multiple polylogarithm 
$G(a,b,c;x) = - Li_{1,1,1}(b/c,a/b,x/a) = I(0;c,b,a;x)$:
\begin{center}
$P(c,b,a,x) = $  {\fourgon c b a x}  $ \quad \longleftrightarrow \quad G(a,b,c;x)$
\end{center}
\noindent
which comes equipped with decorations (in this order) $c$, $b$, $a$ and $x$, the latter decoration $x$ being for the distinguished {\em root side} (drawn by a double line in the picture), and also carries information on the orientation of the polygon in the form of a fat vertex (which should be thought of as the ``first" vertex, while the root side ---adjacent to the first vertex--- is the ``last'' side).

In a {\em first step}, one lists all possible ways to draw the maximal number of non-intersecting arrows (an arrow is a directed line from a vertex of $P$ to a non-adjacent side), which for an $(n+1)$-gon amounts to
$n-1$ such arrows, and one formally adds the resulting objects (the framing polygon being identical, but each equipped with a different maximal set of arrows).

In our example $n=3$, such a maximal set contains $n-1=2$ arrows, and there are precisely 12 different such sets, given by 
\vskip30pt
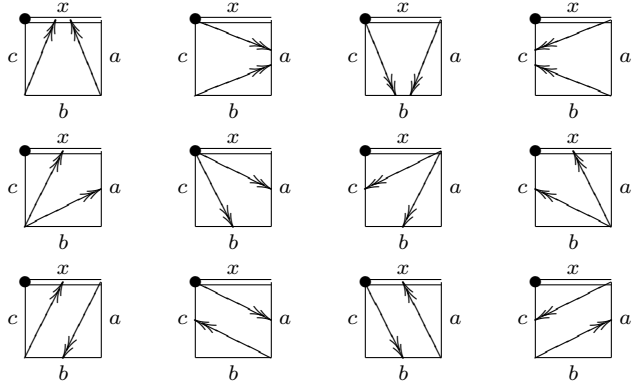
\xy 
\hskip 100pt 
\fourgona \POS(4,5) \ar@{<<-} +(-4,-10)  \POS(6,5) \ar@{<<-} +(4,-10) 
\hskip 60pt 
\fourgona \POS(10,1) \ar@{<<-} +(-10,4) \POS(10,-1) \ar@{<<-} +(-10,-4)
\hskip 60pt 
\fourgona \POS(6,-5) \ar@{<<-} +(4,10) \POS(4,-5) \ar@{<<-} +(-4,10)
\hskip 60pt 
\fourgona \POS(0,-1) \ar@{<<-} +(10,-4) \POS(0,1) \ar@{<<-} +(10,4)
\endxy
\vskip 20pt
\xy 
\hskip 100pt \fourgona \POS(5,5) \ar@{<<-} +(-5,-10) \POS(10,0) \ar@{<<-} +(-10,-5)
\hskip 60pt \fourgona \POS(10,0) \ar@{<<-} +(-10,5)\POS(5,-5) \ar@{<<-} +(-5,10)
\hskip 60pt \fourgona \POS(5,-5) \ar@{<<-} +(5,10) \POS(0,0) \ar@{<<-} +(10,5)
\hskip 60pt \fourgona \POS(0,0) \ar@{<<-} +(10,-5) \POS(5,5) \ar@{<<-} +(5,-10) 
\endxy

\vskip 20pt
\xy 
\hskip 100pt \fourgona \POS(5,5) \ar@{<<-} +(-5,-10) \POS(5,-5) \ar@{<<-} +(5,10) 
\hskip 60pt \fourgona \POS(10,0) \ar@{<<-} +(-10,5) \POS(0,0) \ar@{<<-} +(10,-5)
\hskip 60pt \fourgona \POS(5,5) \ar@{<<-} +(5,-10) \POS(5,-5) \ar@{<<-} +(-5,10)
\hskip 60pt \fourgona \POS(10,0) \ar@{<<-} +(-10,-5)\POS(0,0) \ar@{<<-} +(10,5)
\endxy
\vskip 20pt



In a {\em second step}, to each such maximal set $A$ of arrows in $P$, we associate a rooted tree (as the tree dual to the polygon dissection defined by the arrows) whose decorations are (decorated and rooted) 2-gons. As an example, to the 4-gons in the last column above we attach 
\vskip30pt
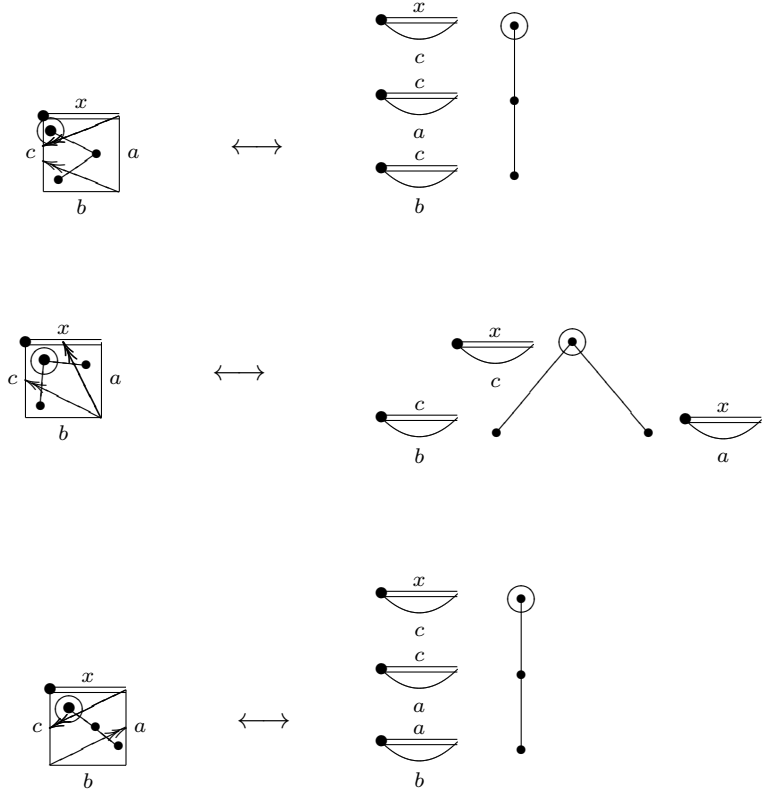
\xy 
\hskip 60pt \fourgona \POS(0,-1) \ar@{<<-} +(10,-4) \POS(0,1) \ar@{<<-} +(10,4)
\POS(1,3) *+{{\bullet}} *{\cir{}}
\POS(7,0) *+{\scriptstyle{\bullet}}
\POS(2,-3.5) *+{\scriptstyle{\bullet}}
\POS(1,3)  \ar@{-} +(6,-3)
\POS(7,0)  \ar@{-} +(-5,-3.5)

 \POS(0,1)  \ar@{<<-} +(10,4)
 
\hskip 70pt  
$\longleftrightarrow  $
\hskip 40pt


\POS(5,10) \ar@{-} +(0,0)^{{\twogon c x }}
\POS(5,0) \ar@{-} +(0,0)^{{\twogon a c }}
\POS(5,-10) \ar@{-} +(0,0)^{{\twogon b c }}

 \POS(5,-10) 

\POS(15,17) *+{\scriptstyle{\bullet}} *{\cir{}}
\POS(15,7) *+{\scriptstyle{\bullet}} 
\POS(15,-3) *+{\scriptstyle{\bullet}} 

\POS(15,17) \ar@{-} +(0,-10)
\POS(15,7) \ar@{-} +(0,-10)
\POS(15,-3) 

\POS(0,0)
\endxy

\vskip30pt
\xy 
\hskip 60pt 
\fourgona \POS(0,0) \ar@{<<-} +(10,-5) \POS(5,5) \ar@{<<-} +(5,-10) 
\POS(2.5,2.5) *+{{\bullet}} *{\cir{}}
\POS(8,2) *+{\scriptstyle{\bullet}}
\POS(2,-3.5) *+{\scriptstyle{\bullet}}
\POS(2.5,2.5)  \ar@{-} +(5.5,-.5)
\POS(2,-3.5)  \ar@{-} +(.5,6)

\POS(0,0)
 \POS(5,5) \ar@{<<-} +(5,-10) 
 
\hskip 70pt
$\longleftrightarrow  $
\hskip 40pt

\POS(-5,-10) \ar@{-} +(0,0)^{{\twogon b c }}
\POS(5,0) \ar@{-} +(0,0)^{{\twogon c x}}
\POS(35,-10) \ar@{-} +(0,0)^{{\twogon a x}}

\POS(25,5) *+{\scriptstyle{\bullet}} *{\cir{}}
\POS(15,-7) *+{\scriptstyle{\bullet}} 
\POS(35,-7) *+{\scriptstyle{\bullet}} 

\POS(25,5)  \ar@{-} +(-10,-12)
\POS(25,5)  \ar@{-} +(10,-12)

\endxy

\vskip30pt
\xy 
\hskip 60pt \fourgona \POS(10,0) \ar@{<<-} +(-10,-5)\POS(0,0) \ar@{<<-} +(10,5)
\POS(2.5,2.5) *+{{\bullet}} *{\cir{}}
\POS(6,0) *+{\scriptstyle{\bullet}}
\POS(9,-2.5) *+{\scriptstyle{\bullet}}
\POS(2.5,2.5)  \ar@{-} +(3.5,-2.5)
\POS(6,0) \ar@{-} +(3.,-2.5)

\POS(0,0) \ar@{<<-} +(10,5)

\hskip 70pt  
$\longleftrightarrow  $
\hskip 40pt

\POS(5,10) \ar@{-} +(0,0)^{{\twogon c x }}
\POS(5,0) \ar@{-} +(0,0)^{{\twogon a c }}
\POS(5,-10) \ar@{-} +(0,0)^{{\twogon b a }}

\POS(15,17) *+{\scriptstyle{\bullet}} *{\cir{}}
\POS(15,7) *+{\scriptstyle{\bullet}} 
\POS(15,-3) *+{\scriptstyle{\bullet}} 

\POS(15,17) \ar@{-} +(0,-10)
\POS(15,7) \ar@{-} +(0,-10)
\POS(15,-3) 

\POS(0,0)
\endxy

\vskip 30pt

\def \symb{{\cal S}}

Any linear order on the vertices of such a rooted 
tree which is {\em compatible} in the sense discussed below in section~\ref{rules} with the partial order on it (only the middle tree above is non-linear hence allows more than one such linear order) now gives a term in the {\em symbol} $\symb(P)$ attached to $P$. 
In practice, this means that every branching in a tree contributes to the symbol by the shuffle of (the vertices that appear on each of) its branches (see below for a more detailed description).

\emph{Third step}: Each of the 2-gons $B$ in one of the linear orders on the vertices now is mapped via a suitable map $\mu$ to a rational function in the original decorations of the polygons (in the example a natural target space would be the function field $\mathbb{Q}(a,b,c,x)$ of rational functions in the variables $a$, $b$, $c$, $x$). More precisely, if $B= \twogon x y$  
where $y$ denotes the root decoration, then we map $B$ to $\mu(B) = 1-y/x$, provided $x\neq 0$, and to $\mu(B)=y$ otherwise.

\smallskip
\emph{Last step}: Fixing the signs. 
We need to invoke
a sign for the individual elementary tensors, and this sign is determined by using the number of {\em backward arrows} in a dissection.
In order to see this quickly, it is convenient to ``break up" the polygon at its {\em first vertex} (in the pictures it is typically indicated by a bullet).
Then we ``roll out" the sequence of sides and arrange it as a line from left to right, starting with the first vertex and ending with the root side; dissecting arrows inside the polygon will be stretched out (in a way that they still do not intersect). We give it for the third example above:

\vskip30pt
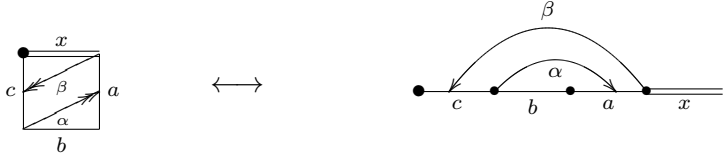
\xy 
\hskip 60pt \fourgona \POS(10,0) \ar@{<<-} +(-10,-5)^{\scriptscriptstyle\alpha\phantom{xx}} \POS(0,0) \ar@{<<-} +(10,5)_{{}^\beta\phantom{xx}}
\hskip 70pt  
$\longleftrightarrow  $
\hskip 40pt

\POS(5,0) *+{{\bullet}}
\POS(5,0) \ar@{-} +(10,0)_c
\POS(15,0) *+{\scriptstyle{\bullet}} 
\POS(15,0) \ar@{-} +(10,0)_b
\POS(25,0) *+{\scriptstyle{\bullet}} 
\POS(25,0) \ar@{-} +(10,0)_a
\POS(35,0) *+{\scriptstyle{\bullet}} 
\POS(35,0) \ar@{=} +(10,0)_x


\POS(15,0) \ar@/^1pc/ +(16,0)_\alpha
\POS(35,0) \ar@/_2pc/ +(-26,0)_\beta

\POS(0,0)
\endxy

\vskip30pt
Now a {\em backward arrow} is one which, in the rolled-out version of the polygon, has its end point to the left of its starting point (i.e. points from right to left, like $\beta$ above), while a {\em forward arrow} has it to its right (i.e. points from left to right, like $\alpha$ above).


Here is  a more formal definition: There is a natural linear order on the sides $e_1,\dots,e_n$ of an $n$-gon as above, starting with the non-root side $e_1$ incident with the first vertex and ending with the root side $e_n$ (in the example above it is the linear order given by the sides $e_1,\dots,e_4$ decorated by $c$, $b$, $a$ and $x$, and the vertices $v_1=e_n\cap e_1$ (the first vertex), $v_2=e_1\cap e_2$, $v_3=e_2\cap e_3$, $v_4=e_3\cap e_4$. This induces a linear order on the vertices $v_j$ which arise as the intersection of $e_j$ and $e_{j-1}$ (indices taken modulo $n$), where  the first vertex is the smallest element in that order.
Then a (non-trivial) arrow can be encoded by a pair $(v_j,e_k)$ with $k\notin\{j-1,j\}$, and it is {\em backward} if and only if $k<j-1$.
With these notions, the {\em sign } attached to a polygon $P$ with a maximal  arrow set $A$ is given by 
$(-1)^{\#\{\text{backward arrows of }A\}}$.
In the three examples discussed above in more detail we get two backward arrows for the first maximal dissection of the square and one backward arrow for the remaining two dissections.

Putting all of the above ingredients together and writing $\tau_A$ for the tree dual to the 
maximal set of arrows $A$, and $(\tau_A, \prec)$ for its partial order, the final formula for the {\em symbol} $\cS(P)$ of an $n$-gon $P$ is
\beq
\symb(P) = \sum_{\text{ max sets }A\atop \text{of arrows in }P}(-1)^{\#\{\text{backward arrows of }A\}}  \sum_{\text{linear orders }\lambda\atop \text{compatible with }(\tau_A,\prec)}
\ \mmu{\twogon{{b_1}}{{a_1}}} \otimes \dots \otimes \mmu{\twogon{{b_{n-1}}}{{a_{n-1}}} }\,.
\eeq
As an example, the first and third of the three maximal sets of arrows above give 

$$ + \mmu{\twogon c x} \otimes \mmu{\twogon a c} \otimes \mmu{\twogon b c } = \left(1-\frac x c\right)\otimes \left(1-\frac c a\right)\otimes \left(1-\frac c b\right)\,$$
and 
$$ - \mmu{\twogon c x} \otimes \mmu{\twogon a c} \otimes \mmu{\twogon b a} = - \left(1-\frac x c\right)\otimes \left(1-\frac c a\right)\otimes \left(1-\frac a b\right)\,,$$
respectively, while the middle term (corresponding to a {\em non-linear} dual tree{, i.e. a dual tree with branchings}) {contributes via the shuffle product of the two branches}
\begin{eqnarray*}
-\mmu{\twogon c x} \otimes \bigg( \mmu{\twogon b c} \uplus \mmu{\twogon a x}\bigg) &&  \\
=-\mmu{\twogon c x} \otimes \mmu{\twogon b c} \otimes \mmu{\twogon a x} &-&  \mmu{\twogon c x} \otimes \mmu{\twogon a x} \otimes \mmu{\twogon b c} \\
= -\left(1-\frac x c\right)\otimes  \left(1-\frac c b\right)\otimes \left(1-\frac x a\right) & - &\left(1-\frac x c\right)\otimes \left(1-\frac x a\right)\otimes  \left(1-\frac c b\right)\,,
\end{eqnarray*}
{where we introduced the symbol for the shuffle product}
\beq
{a\uplus b = a\otimes b+b\otimes a\,.}
\eeq

Motivation and justification of this assignment has been given to an extent in ref.~\cite{GGL:2009}, where it forms part of an expression arising from the well-known bar construction in algebraic topology applied to a differential graded algebra on the polygons above (which in turn is motivated by certain algebraic cycles originally studied by Bloch~\cite{Bloch} and Bloch--Kriz~\cite{Bloch-Kriz}). For an earlier appearance of a very similar structure (called the $\otimes^m$--invariant there), see ref.~\cite{Goncharov-Galois}, \S4.4.

To summarize: an important part of the differential structure of a weight $n$ multiple polylogarithm 
is captured
by a certain decorated $(n+1)$-gon. More precisely, if the arguments of the multiple polylogarithm are expressed in terms of variables/constants $x_1$, \dots, $x_m$ for some $m$, the polygon is an $(n+1)$-gon with decoration by simple expressions in $x_1,\dots,x_m$; now to this (rooted decorated oriented) polygon there is attached in a natural way an expression (its ``symbol") in $V^{\otimes n}$ where $V$ is a finite rank submodule (it might be convenient for the reader to think of $V$ as a finite dimensional vector space) of the space $\Q(x_1,\dots,x_m)^\times$ (of infinite rank), i.e., the invertible rational functions in the variables $x_1,\dots,x_m$.

\subsection{Rules of symbol calculus}\label{rules}
Roughly, a symbol is a formal sum of elementary $n$-fold tensors $a_1\otimes\dots \otimes a_n$, and one works 
in each tensor factor as with (a refined form of) $\rdh\log$ terms. {In other words, each factor $a_i$ in a tensor product is tacitly understood as }
\beq
{d\log a_i \equiv {da_i\over a_i}\,.}
\eeq  
{Furthermore}, we use shuffle products and the following {\em rules} (essentially boiling down to multilinearity, but in an unusual form, as we pass from multiplicative to additive notation):
\begin{itemize}
\item {\bf Distributivity.}\\
\beq\label{distr1}
C\otimes(a\cdot b)\otimes D = C \otimes a \otimes D +C \otimes b \otimes D
\eeq
and consequently
\beq\label{distr2}
C\otimes a^n \otimes D = n \big(C \otimes a \otimes D\big)\,,\qquad  n\in \Z,
\eeq
where $C$ and $D$ denote fixed elementary tensors. Note that $n$ here is a coefficient rather than part of the first tensor factor; in particular, putting $n=0$ we see that $C\otimes 1\otimes D = 0$.
\item  {\bf Neglecting torsion.}\\
We will ``work up to torsion", which means that we will put 
\beq
C\otimes \rho_n \otimes D = 0\,,\qquad  n\in \Z,
\eeq
for $\rho_n$ an $n$-th root of unity.

\item  {\bf Shuffle product.} \\
An important property of the symbol is that it preserves products: more precisely, it maps the product of two multiple polylogarithms to the {\em shuffle product} of their respective symbols, i.e.
\beq\label{producttoshuffle}
\cS\big(G(a_1,\dots,a_r;x)G(b_1,\dots,b_s;y)\big) = \cS\big(G(a_1,\dots,a_r;x)\big)\uplus\cS\big(G(b_1,\dots,b_s;y)\big)\,
\eeq
{where $\uplus$ is the symbol used for the shuffle product of two tensors, defined on elementary tensors by}
\beq
{(a_1\otimes\ldots\otimes a_{n_1})\uplus(a_{n_1+1}\otimes\ldots\otimes a_{n_1+n_2}) = \sum_{\sigma\in\Sigma(n_1, n_2)}\,a_{\sigma^{-1}(1)}\otimes\ldots\otimes a_{\sigma^{-1}(n_1+n_2)}\,,}
      \eeq
{where $\Sigma(n_1, n_2)$ was defined in eq.~\eqref{eq:Sigma_def}. For more details on shuffle algebras, we refer to appendix~\ref{app:algebras}. We note that,
o}n the left hand side of eq.~\eqref{producttoshuffle}, the shuffle permutations are applied to the arguments of the two functions (cf.~e.g.~eq.~\eqref{eq:G_shuffle}), while on right hand side one shuffles the tensor factors instead, in a completely analogous fashion. 

Note that eq.~\eqref{producttoshuffle} is a rather non-trivial fact, as one can already see in the first non-obvious case:
\begin{eqnarray*}
&&\cS\big(G(a;x)G(b;x)\big) = \cS\big(G(a,b;x)+G(b,a;x)\big)= \\
&=& \Big(\big(1-\frac x a)\otimes (1-\frac x b)\ - \ \big(1-\frac x a)\otimes (1-\frac a b)\ +\big(1-\frac x b)\otimes (1-\frac b a)\Big) \\ &+&  \Big(\big(1-\frac x b)\otimes (1-\frac x a)\ - \ \big(1-\frac x b)\otimes (1-\frac b a)\Big) +  \big(1-\frac x a)\otimes (1-\frac a b)\Big)\,,
\end{eqnarray*}
which agrees with $\cS\big(G(a;x)\big)\uplus \cS\big(G(b;x)\big) = \big(1-\frac x a)\otimes (1-\frac x b) + \big(1-\frac x b)\otimes (1-\frac x a)$, due to cancellations of terms.

We will encounter expressions which involve both tensor and shuffle products---in order to avoid writing many parentheses, our convention is that a shuffle takes precedence over a tensor, i.e.
\beq\label{eq:shuffle_binding}
a\otimes b \uplus c \equiv  a\otimes (b \uplus c)\,.
\eeq
Furthermore, we abbreviate elementary tensors with the same factors as follows:
\beq
 a^{\otimes n} = \underbrace{a\otimes \dots \otimes a}_{n \text{ times}}\,.
 \eeq

\item {\bf Refined} $\rdh\log$ {\bf terms.}\\
We emphasise here already, though, that we will not treat $\rdh\log c$ for a rational constant $c$ as zero (as opposed to ref.~\cite{Goncharov:2010jf}) since we would lose a lot of important information this way. Instead we extend the above calculus to rational numbers in complete analogy with the above; so we have, e.g.,
\beq
 C\otimes 2^m\cdot 3^n\cdot x^{-5} \otimes D = m( C\otimes 2 \otimes D) \ +\  n(C\otimes 3 \otimes D) \  -5\ (C\otimes x\otimes D )\,.
 \eeq

\item {\bf Root decoration 0 annihilates:}\\
Since $G(\dots;0)=0$, we also need to put $\cS(G(\dots;0))=0$, and this indicates that we can (and will) ignore polygons whose root side is decorated by 0.

\end{itemize}


\medskip
{\bf Linear orders of a tree.} For a rooted tree $T$, which we view without a fixed embedding into the plane, 
hence e.g.~we {consider as equal} the two trees 

\vskip10pt
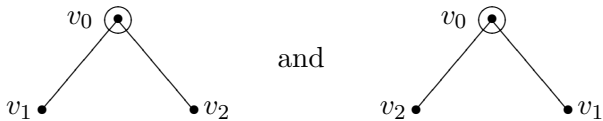
\xy \hskip 50pt\POS(25,5) *+{\scriptstyle{\bullet}} *{\cir{}}
\POS(15,-7) *+{\scriptstyle{\bullet}} 
\POS(35,-7) *+{\scriptstyle{\bullet}} 

\POS(25,5)  \ar@{-} +(-10,-12)
\POS(25,5)  \ar@{-} +(10,-12)

\POS(12,-7) *+{v_1}
\POS(38,-7) *+{v_2}
\POS(20,5)  *+{v_0}
\hskip 40pt  \POS(35,0)  *+{\rm and} \hskip100pt 
\POS(25,5) *+{\scriptstyle{\bullet}} *{\cir{}}
\POS(15,-7) *+{\scriptstyle{\bullet}} 
\POS(35,-7) *+{\scriptstyle{\bullet}} 

\POS(25,5)  \ar@{-} +(-10,-12)
\POS(25,5)  \ar@{-} +(10,-12)

\POS(12,-7) *+{v_2}
\POS(38,-7) *+{v_1}
\POS(20,5)  *+{v_0}

\endxy

\vskip 10pt
\noindent There is a natural partial order $\prec$ on its vertices $v_j$ ($j\in J$), given as follows: the root vertex $v_0\prec v_j$ for any $j\in J$, and $v_j\prec v_k $ for $v_j\neq v_0$ if and only if there is a direct path from root to a leaf passing first through $v_j$ and then through $v_k$.

A {\em linear order on the vertices} of $T$ which is {\em compatible with the order $\prec$} is a sequence $(v_0,v_{j_1},\dots,v_{j_r})$ of all the vertices $v_j$ ($j\in J$) such that $v_{j_i} \prec v_{j_k}$ implies $j_i\leq j_k$. 
(This means that if two vertices are in a relation with respect to the partial order, then they should be related in any compatible linear order in the same way, while if they are not related in the partial order, there is no condition for how they should be related in that linear order.)
In the example, there are precisely two linear orders which are compatible with the partial order, as the root vertex always comes first: $(v_0,v_1,v_2)$ and $(v_0,v_2,v_1)$.

\medskip
{\bf Definition of a symbol.} Now we are ready to give a complete definition of the symbol attached to a (rooted decorated oriented) $(n+1)$-gon $P$ with decorations $(t_1,\dots,t_n,x)$, and then extend it by linearity and shuffle product to any sum of (products of) polygons, hence also for multiple polylogarithms:

\beq\label{eq:symb_def}
\symb(P) = \sum_{\text{ max sets }A\atop \text{of arrows in }P} (-1)^{\#\{\text{backward arrows of }A\}} \sum_{\text{linear orders }\lambda\atop \text{compatible with }(\tau_A,\prec)}
\ \mu\big(\pi_1^{A,\lambda}\big) \otimes \dots \otimes \mu\big({\pi_n^{A,\lambda} }\big)\,,
\eeq
where the 2-gons $\pi_\nu^{A,\lambda}$ are determined by the maximal dissection $A$ together with the linear order $\lambda$ which is compatible with the partial order $\prec$ on $\tau_A$, the dual tree of the dissection $A$,
 in the manner given above in the second step of section~\ref{nutshell} (i.e. for each 2-gon arising from the dissection of $A$ there is a vertex of $\tau_A$ decorated by that 2-gon, and for any two 2-gons  that are adjacent there is an edge in $\tau_A$ connecting the corresponding vertices.

\medskip
{\bf Integrability condition.} 
A very useful property of the rooted decorated polygons, found by the second author in collaboration with F.~Brown and A.~Levin, is that each polygon (or rather its symbol) satisfies a certain integrability condition.
Indeed, an arbitrary sum of elementary tensors does not necessarily lie in the image of the symbol map.
Instead, it was pointed out in ref.~\cite{BrownThesis}, making explicit in a special case the very general approach of Chen~\cite{ChenBullAMS}, that a necessary and sufficient condition  for a symbol
\beq
 S=\sum_{I=(i_1,\dots,i_m)} c_I \ \omega_{i_1} \otimes \cdots \otimes \ \omega_{i_m} \quad (c_I\in \Q)\,,
 \eeq
to be integrable to a function is that 
\beq
\sum_{I=(i_1,\dots,i_m)} c_I \ \omega_{i_1} \otimes \cdots \otimes
 ( \omega_{i_j} \wedge  \omega_{i_{j+1}})\otimes
\cdots \otimes \ \omega_{i_m} \ = \ 0 \quad  \text{for all } 1\leq j\leq m\,,
\eeq
{where $\omega_{i_j} \wedge  \omega_{i_{j+1}}$ denotes the usual exterior product of two differential forms.}
We rewrite this for our purposes as
\beq\label{eq:integrability}
\sum_{I=(i_1,\dots,i_m)} c_I \  \big[\rdh\log \omega_{i_j}\wedge \rdh\log \omega_{i_{j+1}}  \big]\, \omega_{i_1} \otimes \cdots \otimes
 \widehat \omega_{i_j} \otimes  \widehat\omega_{i_{j+1}} \otimes
\cdots \otimes \ \omega_{i_m} \ = \ 0 \,,
\eeq
{where the hats indicate that we omit the corresponding factors in the tensor product.}
As an example, we indicate the statement for $G(a,b;x)$, whose symbol
\beq
 \cS\big(G(a,b;x)\big) =   \Big(1-\frac x b\Big)\otimes \Big(1-\frac x a\Big)\ - \ \Big(1-\frac x b\Big)\otimes \Big(1-\frac b a\Big) +  \Big(1-\frac x a\Big)\otimes \Big(1-\frac a b\Big)
\eeq
satisfies
\beq\label{eq:int_example}
 \rdh\log\Big(1-\frac x b\Big)\wedge \rdh\log \Big(1-\frac x a\Big)\ - \ \rdh\log \Big(1-\frac x b\Big)\wedge \rdh\log \Big(1-\frac b a\Big) +  \rdh\log \Big(1-\frac x a\Big)\wedge \rdh\log \Big(1-\frac a b\Big) =0\,,
 \eeq
 {where we recall that $\rdh\log x ={\rd x/x}$.}
Indeed, writing 
\beq
\rdh\log\Big(1-{\alpha\over\beta}\Big) = \frac{\rd y - \rd \beta}{y-\beta} \Big\vert_{y=0}^\alpha
\eeq
the left-hand side of eq.~\eqref{eq:int_example} becomes
\beq
\frac{\rd y - \rd b}{y-b} \Big\vert_{y=0}^x \wedge \frac{\rd y - \rd a}{y-a} \Big\vert_{y=0}^x 
\,-\, \frac{\rd y - \rd b}{y-b} \Big\vert_{y=0}^x \wedge \frac{\rd y - \rd a}{y-a} \Big\vert_{y=0}^b 
\,+\, \frac{\rd y - \rd a}{y-a} \Big\vert_{y=0}^x \wedge \frac{\rd y - \rd b}{y-b} \Big\vert_{y=0}^a
\eeq
and we find, e.g., that the coefficient of $\rd x\wedge \rd a$ is given by
\beq 
\frac{1}{x-b} \cdot\frac{-1}{x-a}\ -\ \frac{1}{x-b}\cdot \frac{-1}{b-a}\ + \ \frac{1}{x-a}\cdot \frac{1}{a-b} = 0\,.
\eeq
The coefficients of  $\rd x\wedge \rd b$ and  $\rd b\wedge \rd a$ vanish in a similar way.

\medskip
{\bf Relationship to the symbol of ref.~\cite{Goncharov:2010jf}.} In ref.~\cite{Goncharov:2010jf}, the Goncharov, Spradlin, Vergu and Volovich use the differential equation for multiple polylogarithms recursively to arrive at the definition of a symbol. More precisely, if $F:\C^n\to\C$ denotes a complex valued function depending on $n$ complex variables $x_k$, $1\le k\le n$,  the authors of ref.~\cite{Goncharov:2010jf} define the ``symbol of the transcendental function'' $F$ in the following recursive way: if the total differential of $F$ can be expressed in the form
\beq
\rd F = \sum_i\,F_i\,\rdh\log R_i\,,
\eeq
where $F_i$ and $R_i$ are functions of the variables $x_k$, and $R_i$ are moreover rational functions, then the symbol of $F$ is defined recursively via
\beq
\cS(F) = \sum_i\,\cS(F_i)\otimes R_i\,.
\eeq
In the case where $F$ is a multiple polylogarithm, we can write down the differential of $F$ in an explicit form. For example, in the special case where all the arguments of the multiple polylogarithm are generic (i.e., they are mutually different and do not take particular values), we obtain~\cite{Goncharov:2001}
\beq\label{eq:MPL_diff_eq}
\rd G(a_{n-1},\ldots,a_{1};a_n) = \sum_{i=1}^{n-1}G(a_{n-1},\ldots,\hat{a}_i,\ldots,a_{1};a_n)\,\rdh \log\left({a_i-a_{i+1}\over a_i-a_{i-1}}\right)\,.
\eeq
The symbol of $G(a_1,\ldots,a_{n-1};a_n)$ is then defined in the form
\beq\label{eq:GSVV_def}
\cS(G(a_{n-1},\ldots,a_{1};a_n)) = \sum_{i=1}^{n-1}\cS(G(a_{n-1},\ldots,\hat{a}_i,\ldots,a_{1};a_n))\otimes\left({a_i-a_{i+1}\over a_i-a_{i-1}}\right)\,.
\eeq
The symbol we just obtained looks seemingly different from the definition we gave in eq.~\eqref{eq:symb_def}, which consists in summing over all possible maximal sets of arrows of the polygon $P(a_{1},\ldots,a_{n-1},a_n)$ associated to $G(a_{n-1},\ldots,a_{1};a_n)$. In the following we show that the two definitions are equivalent up to a rearrangement of the terms in the sum, and hence give rise to the same tensor. 

Let us consider the $n$-gon $P=P(a_1,\ldots,a_n)$ (i.e. with sides decorated by $a_i$, the last one $a_n$ decorating the root side). We will show that the symbol of $P$ satisfies the recursion~\eqref{eq:GSVV_def}. For simplicity, we will concentrate here on the case of generic decorations. 
Let $\Lambda_P$ be the set of all linear orders on the dual tree attached to any of the maximal sets of arrows of $P$.
Then there is an obvious bijection between the terms in the double sum in eq.~\eqref{eq:symb_def} and the elements in $\Lambda_P$.
We can partition $\Lambda_P$ by collecting all those linear orders into a subset which share the same ÒlastÓ 2-gon that decorates the last vertex of this linear order. This partitions $\Lambda_P$ into a priori $2n$ subsets, as those last vertices correspond precisely to the 2-gons that we can cut off from $P$.

Note that cutting off the last 2-gon in a linear order on a maximal dissection corresponds to contracting the associated edge in the dual tree. Note also that, clearly, the last vertex must be a \emph{leaf} of the (rooted) dual tree, and hence each last 2-gon necessarily cuts off two successive sides of $P$.

\emph{Remark}: For each $n$-gon $P$ with $n > 2$, the three 2-gons involving the root side of $P$ can never become the last one in any linear order in $\Lambda_P$ . More explicitly, these are the two 2-gons $\twogon{a_{n-1}}{a_{n}}$ and $\twogon{a_{1}}{a_{n}}$. The former can arise only by cutting off the root side, while the latter can arise both by cutting off the root side and by cutting off the first side and the first vertex.
As a consequence, $\Lambda_P$ partitions into only $2n-3$ \emph{non-empty} subsets (of same cardinality) of the above type.

In view of the above, it is clear that any such subset indexes exactly the linear orders $\Lambda_{\tilde P}$ on the (dual trees of the maximal dissections of the) subpolygon $\tilde P$ of $P$ which is obtained from $P$ by cutting off a fixed 2-gon, followed by contracting the dissecting arrow to a point.

There are typically two ways of cutting off a 2-gon in which such a subpolygon $\tilde P$ can occur: cutting off the 2-gons $\twogon{a_i}{a_{i\pm1}}$ ($i=2,\ldots , n - 1$) leaves complementary subpolygons $\tilde P_i^\pm$ which are identified with $\tilde P_i$ upon contraction of the dissecting arrow. Note that the two 2-gons will give terms of opposite parity, as precisely one of them will be a forward arrow.
The only exception arises from cutting off $\twogon{a_1}{a_2}$, which corresponds to a forward arrow, for which only one complementary subpolygon $\tilde P_1^+$ can occur.

In summary, we get:\\
{\bf Claim:} \emph{There is a bijection of sets}
\beq\label{lamb_bij}
\Lambda_P \stackrel{1:1}{\longleftrightarrow}\bigcup_{i=1}^{n-1}\left(\Lambda_{\tilde P_i^+}\times \Big\{\hskip -8pt\twogon{a_i}{a_{i+1}}\hskip -5pt\Big\}\right)
\cup
\bigcup_{i=2}^{n-1}\left(\Lambda_{\tilde P_i^-}\times \Big\{\hskip -8pt\twogon{a_i}{a_{i-1}}\hskip -5pt\Big\}\right)\,.
\eeq
\emph{Moreover, the sign of a maximal set of arrows of $\Lambda_P$ agrees with the sign of the corresponding maximal 
set of arrows in $\Lambda_{\tilde P_i^+}$ and is opposite to the sign of the corresponding maximal set of arrows of $\Lambda_{\tilde P_i^-}$.}

All we need to note here is that the above remark ensures that both $\Lambda_{\tilde P_n^\pm}$ and $\Lambda_{\tilde P_1^-}$ are empty, so are left out at the right hand side, and that cutting off a 2-gon of the form $\twogon{a_i}{a_{i+1}}$ corresponds to a forward arrow, hence contributes
a sign $+1$ to the maximal dissection, while $\twogon{a_i}{a_{i-1}}$ corresponds to a backward arrow, hence contributes a sign $-1$.

Due to the bijection~\eqref{lamb_bij}, we can rewrite eq.~\eqref{eq:symb_def} by first summing over all subpolygons $\tilde P_i^\pm$, followed by a sum over all possible elements in $\Lambda_{\tilde P_i^\pm}$. The inner sum then evaluates to the symbol of the subpolygon $\tilde P_i^\pm$, and we are left with
\beq\label{polygon_rec}
\cS(P) = \sum_{i=1}^{n-1}\cS(\tilde P_i^+)\otimes \mmu{\twogon{a_i}{a_{i+1}}} - \sum_{i=2}^{n-1}\cS(\tilde P_i^-)\otimes \mmu{\twogon{a_i}{a_{i-1}}}\,,
\eeq
where the relative minus sign arises because, as discussed above, $\Lambda_{\tilde P_i^\pm}$ contribute with opposite signs. After identification of $\tilde P_i^+$ and $\tilde P_i^-$, eq.~\eqref{polygon_rec} agrees with the recursion~\eqref{eq:GSVV_def} modulo the additivity of the symbol. In order to finish the proof, we need to show that also the bases of the recursions are the same. It is indeed easy to check by explicit computation that, e.g., the symbol of $G(a_2,a_1;a_3)$ obtained from the recursive definition~\eqref{eq:GSVV_def} agrees with the symbol obtained from our polygon construction, eq.~\eqref{eq:symb_def}.

We note that in ref.~\cite{Goncharov:2010jf} the $\rdh\log$ of a constant is put to 0. Although this seems rather natural (and turns out to be sufficient in several cases), we advocate to use a refined version of this (which is what is typically used when working with symbols in a number field):
for each element of a set of multiplicatively independent elements in a given number field one can choose a logarithm independently but then the logarithm of any product formed from those elements is determined.
For example, we will see in section~\ref{sec:HPLbasis} that in the context of harmonic polylogarithms the only constants that need to be treated in this fashion are powers of 2, and hence it is sufficient to think of ``2" as an irreducible element.
The reason for considering this refined version is that it is very helpful for recognising functions from which a given symbol might originate. In particular, it has proved to be very useful{, e.g.}, when ``recognizing" HPL's for keeping track of terms which come from a (shuffle) product of polylogarithms, see section~\ref{sec:HPLbasis}.

{While the symbol of a multiple polylogarithms obtained by considering the maximal set of arrows of the associated decorated polygon 
 is equivalent to the symbol obtained from the recursive definition~\eqref{eq:GSVV_def}, we believe that both approaches have their virtues. While the latter 
might be easier to implement into a computer program in general, it is strictly speaking only valid in the case of generic arguments of the polylogarithms. Indeed, if the arguments are non-generic, we obtain divergences in the right-hand side of eq.~\eqref{eq:MPL_diff_eq}, e.g. when $a_i=a_{i+1}$ for some $i$. It is then in principle necessary to resort to a careful regularization to deal with the degenerate cases~\cite{Goncharov:2001}. The definition of the symbol based on the decorated polygons, being combinatorial in nature, avoids this problem by construction and allows to identify very easily the degenerate cases (they correspond, e.g., to arrows ending on a side whose decoration is zero), and to discard them from the start, avoiding in this way the need of regularization. Furthermore, as we will see in the next section, the polygon approach has naturally built in the refined `$d\log$-prescription', because {the} combinatorial nature of the construction does not make a distinction between constants as e.g. ``2'' for which one might be typically tempted to define $d\log 2=0$.}

\medskip
{\bf Symbols for classical polylogarithms.} The polygons attached to classical polylogarithms $Li_m(x) = -G(\underbrace{0,\dots,0}_{m-1 \text{ terms}},1;x)$, are given by decorations $x$ (for the first side) and 0 (for the remaining non-root sides) as well as 1 (for the root side). Their attached symbol consists of (the negative of) a {\em single} elementary tensor, in fact we have
 \beq
 \label{eq:symb_Li_m}
  \cS\big(Li_m(x)\big) = - \ \big( (1-x)\otimes \underbrace{x\otimes\cdots\otimes x}_{m-1\ {\mathrm factors}}\big)\,,
  \eeq
where we have $m$ factors (``weight" $m$) on the right hand side. (Note the parentheses which separate the coefficient, here $-1$, from the actual tensor, to avoid a misinterpretation as $(x-1)\otimes x\otimes\cdots\otimes x$.)

Such tensors have long been considered in connection with functional equations of polylogarithms---in fact, Zagier~\cite{Zagier-Texel, Zagier-Appendix} has given a criterion for such equations built on those tensors, which has been used (cf.~ref.~\cite{Gangl-Selecta}) to find the first non-trivial equations for $Li_6$ and $Li_7$ (beyond weight~7 none are known), and the corresponding expressions for multiple polylogarithms are important already in Goncharov's early work (e.g.~\cite{Goncharov-Advances}) where he generalises the underlying tensor algebra considerably.

\section{A simple example}
\label{sec:example}
\paragraph{The symbol attached to $G(-1,1;x)$.} In this section we illustrate the fact that the symbol calculus provides a convenient tool to detect functional equations among (multiple) polylogarithms,
on the example of $G(-1,1;x)$ (which happens to coincide with the HPL $- H(-1,1;x)$). Even though we could of course immediately apply eq.~(\ref{eq:GabzLi}) to express $G(-1,1;x)$ in terms of classical polylogarithms, we will derive a similar functional equation using the tensor calculus introduced in the previous section. The multiple polylogarithm $G(-1,1;x)$ is associated to a trigon,
\beq
G(-1,1;x)\quad \longleftrightarrow\quad P(1,-1,x) = \threegon{1}{-1}{x}\,.
\eeq
The dissection of the trigon can easily be translated into the tensor associated to the polylogarithm,
\begin{eqnarray*}
& \threegon{1}{-1}{x} &\\
&\downarrow&\\
&\threegonarrowb{1}{-1}{x} + \threegonarrowc{1}{-1}{x} - \threegonarrowa{1}{-1}{x}&\\
&\downarrow&\\
&\mmu{\twogon{-1}{x}}\otimes\mmu{\twogon{1}{-1}} + \mmu{\twogon{1}{x}}\otimes\mmu{\twogon{-1}{x}} - \mmu{\twogon{1}{x}}\otimes\mmu{\twogon{-1}{1}}&
\end{eqnarray*}
The last line allows to read off the symbol of $G(-1,1;x)$,
\beq\bsp\label{eq:G2_tensor}
\cS&(G(-1,1;x)) \\
&=\left(1-{x\over -1}\right)\otimes\left(1-{-1\over1}\right) + 
\left(1-{x\over 1}\right)\otimes\left(1-{x\over-1}\right)-
\left(1-{x\over 1}\right)\otimes\left(1-{1\over-1}\right)\\
&=(1+x)\otimes2  + (1-x)\otimes (1+x)- (1-x)\otimes 2\,.
\esp\eeq

{Before turning to the question of how to attach a function to this symbol, let us briefly comment on how this symbol could have been obtained by using the recursive definition~\eqref{eq:GSVV_def}. Using eq.~\eqref{eq:MPL_diff_eq}, we obtain
\beq\bsp
\rd G(-1,1;x) &\,= G(-1;x)\,\rdh\log\Big(1-(-1)\Big)+G(1;x)\,\rdh\log\left({(-1)-x\over(-1)-1}\right)\\
&\,=G(-1;x)\,\rdh\log2+G(1;x)\,\rdh\log(1+x)-G(1;x)\,\rdh\log2\,.
\esp\eeq
The three terms in the last line of this equation are in one-to-one correspondence with the three terms in the symbol in eq.~\eqref{eq:G2_tensor}. Note, however, that we had to treat all the arguments of the (three-variable) function $G({\scriptstyle\bullet},{\scriptstyle\bullet};{\scriptstyle\bullet})$ as generic, and to use the refined `$\rdh\log$-prescription', i.e. $\rdh\log2\neq0$. Putting to zero all the $\rdh\log2$ terms is equivalent to putting to zero all elementary tensors in the symbol where a factor inside a tensor product is constant~\cite{Goncharov:2010jf}. As we will see below, we prefer to keep these terms as they provide us with valuable information about the function that should be associated to the symbol.
}

\paragraph{Attaching a function to a symbol.} Since every multiple polylogarithm of weight two can be expressed as a combination of classical polylogarithms, we make the ansatz that $G(-1,1;x)$ can be written, up to an additive constant, in the form
\beq\label{eq:ansatz}
\sum_ic_i\,\textrm{Li}_2(f_i(x)) + \sum_{j,k}c_{jk}\,\ln(g_j(x))\,\ln(h_k(x))\,,
\eeq
such that the tensor associated to this expression corresponds to the tensor in eq.~(\ref{eq:G2_tensor}), where $c_i$ and $c_{jk}$ are rational numbers and $f_i$, $g_j$ and $h_k \in \Q(x)$ are rational functions. We subdivide this problem into smaller ones by postulating
that we can distinguish between the three different contributions in eq.~(\ref{eq:G2_tensor}). 
By using a procedure suggested in ref.~\cite{Goncharov:2010jf, Alday:2010jz} we can distinguish the
first sum from the second by projecting the respective symbols onto their symmetric or alternating parts: each term in the second sum will give zero contribution for the latter one, while each summand in the first sum will give a non-zero contribution. 
Indeed the tensor associated to a product of logarithms is totally symmetric, and hence its contribution to the tensor vanishes if we project onto the antisymmetric component of the tensor in eq.~(\ref{eq:G2_tensor}). 

\paragraph{Preparatory steps: decomposing tensors into symmetric and antisymmetric parts.} We recall that, for a vector space $V$ (over $\C$, say, or more generally over a field of characteristic $\neq 2$) there is a direct sum decomposition $V\otimes V =
(V\odot V) \oplus (V\wedge V)$ (other notations, as used e.g. in refs.~\cite{Zagier-Texel} or \cite{Goncharov-Advances}, are $V\odot V = {\textrm{ Sym}}^2(V)$, and $V\wedge V
=\bigwedge{}^{\hskip -2pt 2\hskip 2pt} V$), and $V\odot V$ is generated by $a\odot b$ (for some $a$, $b\in V$), while
$V\wedge V$ is generated by $a\wedge b$ where we introduce the following rather standard notations for symmetric and antisymmetric tensors,
\beq\bsp
a\odot b &\,\equiv a\otimes b + b\otimes a\,,\\
a\wedge b &\,\equiv a\otimes b - b \otimes a\,.
\esp\eeq
\paragraph{Back to the example.} With this notation, the decomposition of a generic elementary rank two tensor (i.e., $a\otimes b$
for some $a$, $b\in V$)
into its symmetric and antisymmetric components can be expressed as
\beq
a\otimes b = {1\over 2}\,(a\odot b) + {1\over2}\,(a\wedge b)\,.
\eeq
Concentrating solely on the antisymmetric component of eq.~(\ref{eq:G2_tensor}), and using the antisymmetry of the wedge product, which, e.g., induces $2\wedge2=0$, we obtain
\beq\bsp\label{eq:wedge_G2}
\,&(1+x)\wedge2\  -\  \,(1-x)\wedge 2 \ +\  \,(1-x)\wedge (1+x)\\
 &\,= 
\,{1-x\over2}\wedge{1+x\over2}= \,\left(1-{1+x\over2}\right)\wedge{1+x\over2}\,.
\esp\eeq
As the tensor associated to a product of logarithms does not have an antisymmetric component, eq.~(\ref{eq:wedge_G2}) suggests that it is  the antisymmetric part of the tensor associated to some sum of dilogarithms, and from eq.~(\ref{eq:symb_Li_m}), it is easily identified as the antisymmetric part of $\cS\left(-\textrm{Li}_2\left({1+x\over2}\right)\right)$. 

Having identified the dilogarithm contributions to $G(-1,1;x)$, we can proceed via a bootstrap procedure and subtract off this contribution, leaving only a totally symmetric tensor
\beq\bsp
\cS\left(G(-1,1;x) + \textrm{Li}_2\left({1+x\over2}\right)\right) &\,=  2\odot(1+x)  - {1\over2}\,(2\odot2)\\
&\, = \cS\left(  \ln2\,\ln(1+x) -{1\over2}\,\ln^22\right) \,.
\esp\eeq
\paragraph{Fixing the constant.} We have shown that the tensor associated to $G(-1,1;x)$ equals the tensor associated to the combination $-\textrm{Li}_2\left({1+x\over2}\right) + \ln2\,\ln(1+x) -{1\over2}\,\ln^22$. It would be premature, however, to conclude that both expressions are equal, but they are equal only up to an additive constant independent of $x$. Indeed, specializing to $x=0$, and using the fact that $G(-1,1;x=0)=0$ and $\textrm{Li}_2(1/2) = {\pi^2\over12}-{1\over2}\log^22$, we see that
\beq
\Bigg\{G(-1,1;x) -\left[-\textrm{Li}_2\left({1+x\over2}\right) + \ln2\,\ln(1+x) -{1\over2}\,\ln^22\right]\Bigg\}_{\big|x=0} = {\pi^2\over 12}\,.
\eeq
Thus, we obtain
\beq
G(-1,1;x) = -\textrm{Li}_2\left({1+x\over2}\right) + \ln2\,\ln(1+x) -{1\over2}\,\ln^22 + {\pi^2\over 12}\,.
\eeq
Note that this additive constant is not detected by the symbol, because 
\beq
\cS(\pi^2) = -\cS(\log^2(-1))=0\,.
\eeq

\section{Integrating symbols: an algorithmic approach}
\label{sec:higher_weights}
In the previous section we illustrated how the symbol calculus can be used to derive a functional equation among polylogarithms. While in that weight two example all the steps were easily carried out by hand, an algorithmic approach is desirable in more complicated cases.
In this section, we present our approach that often allows to integrate a symbol of a (transcendental) function, i.e., to find a function $F$, written as a linear combination of (products of) multiple polylogarithms, whose symbol matches a given tensor $S$, which in the following we always assume to satisfy the integrability condition~\eqref{eq:integrability}.

From the example of the previous section it should be clear that the main challenges to address when going to higher weight $w$ are 
\begin{enumerate}
\item choosing appropriate arguments of the desired function types (as a few examples of function types of weight four, we list $\textrm{Li}_4$, $\textrm{Li}_{2,2}$ or $\textrm{Li}_2 \log \log$) such that their symbols span the vector space in which the tensor $S$ lies;
\item finding the generalization of the decomposition of weight two tensors into symmetric and antisymmetric parts (indicated in the simple example of weight two in section~\ref{sec:example}) to higher weights. This problem was addressed up to weight four in refs.~\cite{Goncharov:2010jf, Alday:2010jz}. 
\end{enumerate}

Let us assume that we have a linear combination $S$ (with rational coefficients) of elementary tensors
where the factors in each elementary tensor 
 are rational functions of some variables $x_1,\ldots, x_n$. In the following we assume the tensor to be of ``pure weight'' $w$, i.e., each elementary tensor is assumed to have the same number of {factors.
Without loss of generality we can then assume that $S$ takes the form (all sums are assumed to be finite)
\beq\label{eq:symb_1}
S = \sum_{i_1,\ldots,i_w}c_{i_1,\ldots,i_w}\,\Big(R_{i_1}(x_1,\ldots,x_n)\otimes\ldots\otimes R_{i_w}(x_1,\ldots,x_n)\Big)\,,
\eeq
where $R_{i_\ell}(x_1,\ldots,x_n)$ ($1\leq i_\ell\leq m(S)$ for some $m(S)$ determined by the initial tensor $S$) are rational functions in the  variables $x_i$ and the $c_{i_1,\dots,i_w}$ are rational numbers. 
Distributivity (cf.~eqs.~(\ref{distr1}, \ref{distr2})) then implies that, without loss of generality, $S$ can be rewritten (with new constants $\tilde c_{j_1,\dots,j_w}\in \mathbb{Q}$) as
\beq
S = \sum_{j_1,\ldots,j_w}\tilde c_{j_1,\ldots,j_w}\,\pi_{j_1}\otimes\ldots\otimes \pi_{j_w}\,,
\eeq
where $\pi_{j}=\pi_{j}(x_1,\ldots,x_n)$ ($1\leq j\leq K$ for some $K$) are suitably chosen rational functions which are multiplicatively independent (i.e., there is no non-trivial relation of the form $\prod_{j=1}^K \pi_j^{r_j} = \pm1$, for $r_j\in\mathbb{Z}$). 
In practice, we achieve this by simply factorizing the rational functions $R_{i}(x_1,\ldots,x_n)$ in eq.~\eqref{eq:symb_1} into irreducible polynomials over $\mathbb{Q}$, say (i.e., polynomials in $\mathbb{Q}[x_1,\ldots,x_n]$ that cannot be written as the product of two non-constant polynomials in $\mathbb{Q}[x_1,\ldots,x_n]$). Let us denote the set of these polynomials by $P_S = \{\pi_1,\ldots,\pi_K\}$, which will constitute our building blocks in the following. The symbol can then be seen as an element of the tensor algebra over the $\Q$-vector space generated by the formal basis vectors in the set $P_S$ (more precisely, we should consider it as an element of the $w$th grading of the tensor algebra over the $\Z$-module 
$\langle \pm \prod_{j=1}^K \pi_j^{r_j}\mid r_j\in \Z\rangle$). Our goal is now to find a function, say $F$, that is a rational linear combination of (multiple) polylogarithms (and products thereof) whose arguments are rational functions in the $x_i$ such that $\cS(F) = S$. The procedure to achieve this proceeds in two steps: first we have to decide on the types of functions that should appear in $F$, and then we have to concoct suitable rational functions in the $x_i$ as arguments of these functions such that for some linear combination of these functions the resulting expression fulfills the condition $\cS(F) = S$.
Note that this latter step is not algorithmic in general, as it may involve finding arguments for the functions that have singularities \emph{outside} $P_S$. 

\subsection{Choosing the types of functions}
Our first goal is to construct a set of function types (our `basic types') out of which we can construct our candidate function $F$. 
This involves two steps, and we want both the functions and their arguments to be `as simple as possible', but we need to take into account that all the possible function types one can write down for a given weight are related by an abundance of functional equations. The main criterion we will use in the following is that a function type that can be written as a product of lower weight function types is `simpler' than a function type of \emph{pure weight/transcendentality} (i.e., a function that cannot be written as a sum of terms, each of which being a product of lower weights). Furthermore, we are guided by the conjecture (which the second author learned many years ago from Goncharov) that
a multiple polylogarithm $\textrm{Li}_{m_1,\ldots,m_k}$ with $m_j=1$ for some $j$ can be expressed in terms of multiple polylogarithms where no index is equal to 1.
This conjecture suggests to put a restriction on the function types of pure transcendentality that can appear for a given weight. Furthermore, the shuffle and stuffle relations provide us with further constraints. As an example, we can deduce from
\beq
\textrm{Li}_{m_1,m_2}(x,y) + \textrm{Li}_{m_2,m_1}(y,x) = \textrm{Li}_{m_1}(x) \,\textrm{Li}_{m_2}(y) - \textrm{Li}_{m_1+m_2}(x\,y)\,.
\eeq
that we can hence ignore $\textrm{Li}_{m_2,m_1}$  with $m_2<m_1$.
For low weights, the corresponding sets of (presumably independent) functions which are indecomposable, i.e.~cannot be written in terms of products of lower order functions, are given in Table~\ref{tab:pure_trans_basis}.

\begin{table}[!t]
\begin{center}
\begin{tabular}{|c|c|}
\hline
Weight & Basic function types of pure weight \\  
\hline
1 & $\ln x$\\
2 & $\textrm{Li}_2(x)$\\
3 & $\textrm{Li}_3(x)$\\
4 & $\textrm{Li}_4(x)$, $\textrm{Li}_{2,2}(x,y)$\\
5 & $\textrm{Li}_5(x)$, $\textrm{Li}_{2,3}(x,y)$\\
6 & $\textrm{Li}_6(x)$, $\textrm{Li}_{2,4}(x,y)$, $\textrm{Li}_{3,3}(x,y)$, $\textrm{Li}_{2,2,2}(x,y,z)$\\
\hline
\end{tabular}
\caption{\label{tab:pure_trans_basis}Indecomposable multiple polylogarithms of pure weight $\leq 6$.}
\end{center}
\end{table}

\subsection{Finding the arguments}
\label{sec:arguments}
Having at hand a suitable set to construct the basic function types from, we still need to find the {\em arguments} of these function types. In the context of particle physics it has proved helpful to use guidance from educated guesses, motivated by physical constraints (cf. refs.~\cite{Dixon:2011pw, Heslop:2011hv}), to construct the symbol and/or the functions expressing the desired physical quantities.
To see how 
one might proceed even without any such guidance, 
let us concentrate first on classical polylogarithms only. 

We start by defining, for $P_S=\{\pi_j\}_j$ as above, 
\beq
\overline{P}_S = P_S \cup P'_S\,,
\eeq
where $P'_S$ is the set of all prime factors that appear in $\pi_i\pm\pi_j$ and $1\pm\pi_i$, $\forall\pi_i, \pi_j\in P_S$. Let us denote the elements of $\overline{P}_S$ by $\overline{\pi}_i$.
Since $S$ is constructed out of the irreducible polynomials $\pi_i\in P_S\,\subset \overline{P}_S$, it is perhaps natural to hope that all arguments appearing in the polylogarithmic expressions needed for $S$ can be written in the form
\beq
R^\pm_{n_1,\ldots,n_k}(x_1,\ldots,x_n) = \pm\overline{\pi}_1^{n_1}(x_1,\ldots,x_n)\,\ldots\,\overline{\pi}_k^{n_k}(x_1,\ldots,x_n) \,,
\eeq
where $n_1,\ldots,n_k$ are integers. Let us denote the set of these functions $R^\pm_{n_1,\ldots,n_k}(x_1,\ldots,x_n)$ by $\cR_S$, i.e.~this is, up to sign, the multiplicative span of the $\bar\pi_j\in \overline{P}_S$. 
Note that in practice it is often enough to consider $\cR_S$ to be the span of only a subset of the polynomials in $\overline{P}$. 
Finally, note that the set $\cR_S$ carries a group structure, given by the multiplication of rational functions.

\paragraph{Choosing arguments for classical polylogarithms.} However, not all of these functions are good candidates for arguments of polylogarithms. Indeed, if for example such a function appears as an argument of a classical polylogarithm, then by eq.~\eqref{eq:symb_Li_m} for $R=R^\pm_{n_1,\ldots,n_k}(x_1,\ldots,x_n)$ we can write
\beq\bsp
\cS(\textrm{Li}_n(R)) 
= -\,(1-R)\otimes\underbrace{R\otimes\cdots\otimes R}_{(n-1)\textrm{ times}}\,.
\esp\eeq
It is now easy to see that if we want this tensor to be an element of the tensor algebra of the vector space generated by the set $\overline{P}_S$, then we need to impose the constraint
\beq\label{eq:1_minus_R}
1-R\in \cR_S\,.
\eeq
Let us introduce the notation
\beq\label{eq:1_minus_R_upper_1}
\cR^{(1)}_S = \{R\in \cR_S \mid 1-R\in \cR_S\} \subset \cR_S\,.
\eeq
It follows that, for $R\in \cR^{(1)}_S$, the symbol of $\textrm{Li}_n(R)$ ($n\ge1$) is a linear combination of tensors of the form $\overline{\pi}_{\ell_1}\otimes\ldots\otimes\overline{\pi}_{\ell_n}$. Hence all the rational functions in the set $\cR^{(1)}_S$ are good candidates for arguments of the classical polylogarithms that can appear in our function $F$. Note that $\cR^{(1)}_S$ is no longer a group in general. It can however be given some more structure by considering the following action of the symmetric group $S_3$ on rational functions, defined for a rational function $R$ and rational functions $\sigma_i$ of one variable by
\begin{eqnarray}\label{eq:S3_action}
\sigma_1(R) = R\,, & \sigma_2(R) = 1-R\,, & \sigma_3(R) = 1/R\,,\nonumber\\
\sigma_4(R) = 1/(1-R)\,, & \sigma_5(R) = 1-1/R\,, & \sigma_6(R) = R/(R-1)\,.
\end{eqnarray}
Note that the $\sigma_j$ form a group (under composition of functions) isomorphic to the permutation group $S_3$ on three letters.
It is easy to check that $\cR_S^{(1)}$ is closed under this action of $S_3$. As $S_3$ is generated by the two elements $\sigma_2$ and $\sigma_3$, it is enough to check that $\cR_S^{(1)}$ is closed under these two maps. Closure under $\sigma_2$ is trivial by definition of $\cR_S^{(1)}$. To see that it is also closed under $\sigma_3$ we have to check that $\forall R\in \cR_S^{(1)}$, $1-\sigma_3(R)\in \cR_S$. Indeed, we have
\beq
1-\sigma_3(R) = 1-1/R = -(1-R)\,R^{-1} \in \cR_S\,,
\eeq
because of the group structure of $\cR_S$.

\paragraph{Choosing arguments for polylogarithms of depth $>1$.} So far we have only considered classical polylogarithms, but in general we should also be able to make a sensible ansatz for the arguments of multiple polylogarithms of depths greater than one. In the following, we find it more convenient to work with the functions $G_{m_1,\ldots,m_k}$ defined in eq.~\eqref{eq:Gm_def} rather than with the functions $\textrm{Li}_{m_1,\ldots,m_k}$. As the two function types are related via eq.~\eqref{eq:Gm_def}, one can easily convert from one representation to the other.

 Let us consider a multiple polylogarithm of depth two, say $G_{2,2}$. We are hence looking for a pair of rational functions $(R_1, R_2)\in\cR_S\times \cR_S$ that are good candidates for the arguments of $G_{2,2}$. The symbol of $G_{2,2}(R_1,R_2)$ is given by
\beq\bsp
\cS&(G_{2,2}(R_1,R_2)) = -\left(1-{1\over R_1}\right)\otimes R_1\uplus\left[\left(1-{R_1\over R_2}\right)\otimes {R_1\over R_2}\right]\\
&\,-\left(1-{1\over R_2}\right)\otimes R_2\uplus\left[\left(1-{R_2\over R_1}\right)\otimes {R_2\over R_1}\right]
-\left(1-{1\over R_2}\right)\otimes \left(1-{R_2\over R_1}\right)\otimes R_1\uplus R_2\\
&\,+2\left(1-{1\over R_2}\right)\otimes \left(1-{R_2\over R_1}\right)\otimes R_1\otimes R_1
-2\left(1-{1\over R_2}\right)\otimes \left(1-{1\over R_1}\right)\otimes R_1\otimes R_1\\
&\,-\left(1-{1\over R_2}\right)\otimes R_2\uplus\left[\left(1-{1\over R_1}\right)\otimes R_1\right]
\,,
\esp\eeq
recalling our notation for the shuffle products (see eq.~\eqref{eq:shuffle_binding}),
\beq\bsp
A\otimes B\otimes C\uplus D \,&= A\otimes B\otimes C\otimes D + A\otimes B\otimes D\otimes C\,,\\
A\otimes B\uplus( C\otimes D) \,&= A\otimes B\otimes C\otimes D + A\otimes C\otimes B\otimes D + A\otimes C\otimes D\otimes B\,.
\esp\eeq
Using the same reasoning as for classical polylogarithms, we see that the candidate arguments for multiple polylogarithms of depth two are pairs of rational functions from the set
\beq
\cR_S^{(2)} = \{(R_1,R_2)\in\cR_S^{(1)}\times\cR_S^{(1)}|R_1-R_2\in\cR_S\}\,.
\eeq
An important consequence is that no new rational functions are needed to construct the set $\cR_S^{(2)}$, but all the information is already contained in $\cR_S^{(1)}$. The new set $\cR_S^{(2)}$ then consists of pairs of elements of $\cR_S^{(1)}$, subject to the additional constraint that their difference must again be factorizable in terms of the same prime elements. Moreover, we saw that $\cR_S^{(1)}$ is endowed with a natural action of the group $S_3$, defined in eq.~\eqref{eq:S3_action}. It is hence natural to ask for non-trivial symmetry groups that leave the set $\cR_S^{(2)}$ invariant. First, it is easy to see that the defining condition for $\cR_S^{(2)}$ is invariant (up to an overall sign) under swapping the two entries of any given pair.
Second, the action of $S_3$ defined in eq.~\eqref{eq:S3_action} induces a (simultaneous on both factors) action on $\cR_S^{(1)}\times\cR_S^{(1)}$, defined for $\sigma_i\in S_3$ by
\beq
(R_1, R_2) \stackrel{\sigma_i}{\longrightarrow} (\sigma_i(R_1),\sigma_i(R_2))\,.
\eeq
It is now easy to check that $\cR_S^{(2)}$ is closed under this action. To see this, it is enough to check that $\sigma_i(R_1) - \sigma_i(R_2)\in \cR_S$ for $i=1,2$ and whenever $(R_1,R_2)\in \cR_S^{(2)}$. Indeed, we have
\beq\bsp
&\sigma_2(R_1) - \sigma_2(R_2) = (1-R_1) - (1-R_2) = -(R_1-R_2)\in\cR_S\,,\\
&\sigma_3(R_1) - \sigma_3(R_2) = 1/R_1 - 1/R_2 = -(R_1-R_2)R_1^{-1}R_2^{-1}\in\cR_S\,,
\esp\eeq
where we used the fact that $R_1-R_2\in \cR_S$ and that $\cR_S$ is a multiplicative group. Combining this $S_3$ symmetry with the invariance under an exchange of arguments, here $R_1$ and $R_2$, we see that $\cR_S^{(2)}$ is closed under the action of the group $S_3\times S_2$, defined for $(\sigma,\rho)\in S_3\times S_2$ by
\beq
(R_1,R_2) \stackrel{(\sigma,\rho)}{\longrightarrow}\left(\sigma(R_{\rho(1)}), \sigma(R_{\rho(2)})\right)\,,
\eeq
i.e. the factor $S_2$ simply acts as a permutation of the entries.

The previous discussions for depths one and two readily generalize to higher depth. Our candidate arguments for the multiple polylogarithms of depth $k$ are $k$-tuples of rational functions from the set
\beq\label{eq:RTk}
\cR_S^{(k)} = \{(R_1,\ldots,R_k)\in\cR_S^{(1)}\times\ldots\times\cR_S^{(1)} \mid R_i-R_j\in\cR_S, 1\le i<j\le k\}\,,
\eeq
and using exactly the same argument as in the depth two case, we see that $\cR_S^{(k)}$ can be equipped with an action of the group $S_3\times S_k$, acting on $(R_1,\ldots,R_k)\in\cR_S^{(k)}$ via
\beq
(R_1,\ldots,R_k) \stackrel{(\sigma,\rho)}{\longrightarrow}\left(\sigma(R_{\rho(1)}), \ldots,\sigma(R_{\rho(k)})\right)\,,
\eeq
i.e. the factor $S_k$ simply acts as a permutation of the entries.

\subsection{Integrating the symbol (1)}
\label{sec:integ_symb_1}
\def \symb{ S}
We now turn to the problem of integrating the tensor $S$ that satisfies the integrability condition~\eqref{eq:integrability}.
In practice, such tensors could come from computing a Feynman integral in terms of multiple polylogarithms, or by computing its symbol by other means~\cite{CaronHuot:2011ky,Dixon:2011pw, Heslop:2011hv,Gaiotto:2011dt,Dixon:2011ng,DelDuca:2011wh}.
Our goal is to find a function $F$, more precisely a linear combination of (multiple) polylogarithms, such that $\cS(F) = S$.
The considerations of the previous section 
suggest writing down an ansatz for the set $\Phi$ of {\em functions} (i.e., function types, together with rational functions as arguments) to express $F$ in.  
We assume the ``folklore conjecture" that any functional equation of MPL's involving different weights should actually split into functional equations of pure weight, so we can assume $S$ to be of pure weight. It is convenient to 
partition $\Phi$ into subsets $\Phi^{(w)}$ according to the weight $w$, say
$\Phi^{(w)} = \{b_i^{(w)}\}_i$. 
Then we want to find rational coefficients $c_{i_1\ldots i_\ell}$ such that
\beq\bsp\label{eq:S_Tensor}
S &\,= \sum_{i}c_i\,\cS\left(b_i^{(w)}\right) + \sum_{\stackrel{i_1,i_2}{w_1+w_2=w}}c_{i_1i_2}\,\cS\left(b_{i_1}^{(w_1)}\,b_{i_2}^{(w_2)}\right) + \ldots + \sum_{i_1,\ldots,i_w}c_{i_1\ldots i_w}\,\cS\left(b_{i_1}^{(1)}\ldots b_{i_w}^{(1)}\right)
\esp\eeq
where we can assume that $w_i\geq w_{i+1}$ for all $i$.
In view of property \eqref{producttoshuffle}, each term on the right hand side can be written as a shuffle of terms $\cS\left(b_i^{(w)}\right)$.
All the terms on both sides of this equation are elements of the grade $w$ part of the tensor algebra over the vector space spanned by basis vectors labeled by the elements in $\overline{P}_S$. We know from linear algebra that a straightforward basis for the tensor space of weight $w$ tensors over the vector space spanned by $\overline{P}_S$  is given by $\{\bar\pi_{i_1}\otimes\ldots\otimes \bar\pi_{i_w} | \bar\pi_{i_\ell}\in \overline{P}_S\}$. 
At this stage we have hence mapped the problem of finding a function $F$ satisfying $\cS(F)=S$ into a problem of linear algebra, more precisely the problem of finding the coefficients $c_{i_1\ldots i_\ell}$ such that eq.~\eqref{eq:S_Tensor} is true. As we know a basis of the tensor algebra, we can just compute the coefficients by extracting and comparing the coefficient of each basis vector on either side of the equation and solve the ensuing linear system. Note that the function $F$ obtained in this way is not unique. Indeed, the map $\cS$ is non-injective and hence it would be possible to find a different function $F'$ such that $F-F'$ in the kernel of $\cS$. This issue will be addressed in section~\ref{sec:kernel}.

\subsection{A set of projectors}
Even though we have solved the problem of integrating the symbol in principle and we have reduced it to a linear algebra problem, the linear system one has to solve can be quite large. 
It is therefore preferable to break it down into smaller problems, for example by introducing a suitable filtration on the target space\footnote{A filtration of a vector space $V$ is a sequence $\{V_i\}_{1\le i\le n}$ of subspaces of $V$ such that $V_1\subset V_2\subset\ldots\subset V_n=V$\,.} which allows to successively solve the problem for the filtration pieces.
Such a filtration would allow to separate the different contributions in eq.~\eqref{eq:S_Tensor}: to get the ball rolling, we would like to solve for the coefficients of the functions $b_i^{(w)}$ without having to care about the product terms. This can be achieved by introducing a projector that sends to zero exactly all the product terms.
\begin{defi}\label{defi:Pi}
Let $V$ be a vector space. 
We define linear operators $\Pi_w$ acting on elementary tensors of lengths $w\ge1$ by  $\Pi_1 = \textrm{id}$ and for $w\ge 2$ by
\beq\bsp
\Pi_w&(a_1\otimes\ldots\otimes a_w)\\ &\, = {w-1\over w}\,\left[\Pi_{w-1}(a_1\otimes\ldots\otimes a_{w-1})\otimes a_w - \Pi_{w-1}(a_2\otimes\ldots\otimes a_{w})\otimes a_1\right]\,.
\esp\eeq
\end{defi}
Similar sets of operators acting on shuffle algebras have already been defined in refs.~\cite{Ree:1958,Griffing:1995}, differing from our definition only by an overall normalization (F.~Brown, in a text  in preparation on the representation theory of polylogarithms, uses essentially the same operators, referring to ref.~\cite{Bourbaki}, II \S2, Prop.~7a). More precisely, the operators in refs.~\cite{Ree:1958,Griffing:1995} are defined by $\rho_1=\textrm{id}$, and for $w\ge2$ by
\beq
\rho_w(a_1\otimes\ldots\otimes a_w) = \rho_{w-1}(a_1\otimes\ldots\otimes a_{w-1})\otimes a_w - \rho_{w-1}(a_2\otimes\ldots\otimes a_{w})\otimes a_1\,.
\eeq
The exact correspondence between $\Pi_w$ and $\rho_w$ is simply, for all $w\ge1$,
\beq
\Pi_w = {1\over w}\rho_w\,.
\eeq
It follows from refs.~\cite{Ree:1958,Griffing:1995} that the kernel of $\rho_w$ corresponds exactly to the ideal\footnote{An ideal in a commutative algebra $\cA$ is an additive subgroup $\cI$ such that $\forall a\in \cA$ and $\forall b\in \cI$, one has $ab\in \cI$.} generated by all shuffle products. Since $\Pi_w$ and $\rho_w$ only differ by an overall normalization, we immediately arrive at the following
\begin{proposition}
The kernel of $\Pi_w$ equals the ideal generated by all shuffle products, i.e., for every element $\xi$ in a shuffle algebra, $\Pi_w(\xi)=0$ if and only if $\xi$ can be written as a linear combination of shuffle products.
\end{proposition}
{In other words, the projectors $\Pi_w$ are by construction such that they annihilate precisely all shuffles. Conversely, if $\Pi_w$ applied to some tensor $\xi$ does not vanish, then it is not possible to express $\xi$ as a linear combination of shuffle products.}

The reason for the normalization factor in Definition~\ref{defi:Pi} is that it makes $\Pi_w$ idempotent, in other words, $\Pi_w$ is a projector.
\begin{proposition} For any  $w\in\mathbb{N}\,,$
$\Pi_w$ is idempotent, and hence a projector, i.e.
\beq
\Pi_w^2 = \Pi_w\,.
\eeq
\end{proposition}
\begin{proof}
In ref.~\cite{Ree:1958}, Lemma~1.2, and also in eq.~(2) in ref.~\cite{Griffing:1995},  it was shown that $\rho_w$ satisfies the identity
\beq
w\,(a_1\otimes\ldots\otimes a_w) = \sum_{k=0}^{w-1}(a_1\otimes\ldots\otimes a_k)\uplus\rho_{w-k}(a_{k+1}\otimes\ldots\otimes a_w)\,.
\eeq
We now act on both sides of this equation with $\rho_w$. Since $\rho_w$ annihilates all shuffles, all the terms on the right-hand side vanish, except for $k=0$. Hence we obtain
\beq
w\,\rho_w(a_1\otimes\ldots\otimes a_w) = \rho^2_w(a_1\otimes\ldots\otimes a_w)\,,
\eeq
and so $\Pi_w^2=\Pi_w$, after dividing both sides by $w^2$.
\end{proof}

\subsection{Integrating the symbol (2)}
The projectors defined in the previous section can be used to improve the chances of successful integration of the symbol described in section~\ref{sec:integ_symb_1}.
As the kernel of $\Pi_w$ corresponds to all possible linear combinations of shuffles, its effect on eq.~\eqref{eq:S_Tensor} is to remove all product terms, i.e.,
\beq
\Pi_wS = \sum_{i}c_i\,\Pi_w\cS\left(b_i^{(w)}\right)\,,
\eeq
and so we can solve for {\em fewer} coefficients $c_i$ without having to worry about the product terms. 

\def \la{\lambda}
\def \cF{{\mathcal F}}
Having found the coefficients $c_i$, we are left with the determination of the coefficients of the product terms.
In order to proceed via induction, we determine the behaviour of shuffles under tensor products of projectors (which are themselves projectors). We first study the instructive case of a projector $\Pi_\ell\otimes \Pi_{\ell'}$ which we apply to a shuffle 
\beq \label{simpleshuffle}
(a_1\otimes \dots\otimes a_{k})\sha (b_{1}\otimes \dots\otimes b_{k'})\,,
\eeq
where $k+k'=\ell+\ell' = w$ (still $w$ denoting the weight), with $k\geq k'\geq 1$, $\ell\geq \ell'\geq 1$ and $k\neq \ell$. We can rewrite the shuffle in a form more suitable to applying $\Pi_\ell\otimes \Pi_{\ell'}$ by regrouping the
sum \eqref{simpleshuffle} using deconcatenation (i.e.~splitting into a left hand part and a right hand part)
of the sequence $(1,\dots,k)$ into two blocks $(1,\dots, k_1)$ and $(k_1+1,\dots,k_1+k_2)$, giving
\beq 
\sum_{k_1=1}^{\min(k,\ell)} \big((a_1\otimes\dots\otimes a_{k_1})\sha (b_1\otimes \dots\otimes b_{\ell-k_1})\big) \,\otimes \, \big((a_{k_1+1}\otimes \dots \otimes a_k)\sha (b_{\ell-k_1+1}\otimes \dots \otimes b_{k'})\big)\,.
\eeq
Applying $\Pi_\ell\otimes \Pi_{\ell'}$ to this sum will annihilate each of the $\min(k,\ell)$ summands individually,
in fact it will annihilate the left hand factor except (possibly) for  $k_1=\ell$, since in the other cases there is a proper shuffle on the left which will be mapped to zero by $\Pi_\ell$ already, by Proposition 1.
If the remaining case occurs, the corresponding right hand factor is mapped to zero by $\Pi_{\ell'}$ instead.
For example, if $(k_1,k_2)=(3,3)$ and $(\ell_1,\ell_2)=(4,2)$, we rewrite 
\beq\bsp
(a_1\otimes  a_2 \otimes a_3) \sha (b_1\otimes b_2\otimes b_3) =\,& 
\big( a_1\sha (b_1\otimes b_2\otimes b_3) \big) \otimes  (a_2 \otimes a_3)\\
+\,& \big((a_1\otimes  a_2) \sha (b_1\otimes b_2)\big) \ \otimes\  (a_3\sha b_3) \\
+\,&\big((a_1\otimes  a_2 \otimes a_3) \sha (b_1)\big) \ \otimes\ \big(b_2\otimes b_3\big) \,,
\esp\eeq
where  already the factor $\Pi_4$ of $\Pi_4\otimes \Pi_2$ annihilates the left part of each of the terms.

More generally, for a partition $\la=(\la_1,\dots,\la_r)$ of $w$, we call the {\em $\la$-shuffle} of a $w$-fold tensor $a_1\otimes\dots\otimes a_w$ the product
\beq
(a_1\otimes\dots\otimes a_{\la_1})\sha (a_{\la_1+1}\otimes\dots\otimes a_{\la_1+\la_2})\sha \dots \sha(a_{\la_1+\dots+\la_{r-1}+1}\otimes\dots\otimes a_w)
\eeq
and given any partition $\la'=(\la'_1,\dots,\la'_s)$ with $s<r$, the projector  $\Pi_{\la'}$ will vanish on a $\lambda$-shuffle, as we can again regroup its terms according to $\la'$ and find for each so combined summand at least one projector factor
$\Pi_{\la'_j}$ which vanishes on the corresponding part. This sketches a proof of


\begin{proposition} For two different non-increasing partitions $\la$ and $\la'$ of $w$, of length $\ell(\la)$ and $\ell(\la')$, respectively, we have that a $\la$-shuffle vanishes under $\Pi_{\la'}$  whenever $\ell(\la)\geq \ell(\la')$.
\end{proposition}

This proposition suggests to define a sequence of subspaces, each of which contains the next one (hence forming a descending filtration).
For this, we consider the standard lexicographic order $\succ$ on non-increasing partitions of $w$ given by (denoting $\lfloor x\rfloor$ the largest integer $\leq x$)
 \beq\bsp
 (w)\succ(w-1,1)\succ
(w-2,2)\succ\dots&\,\succ(w-\lfloor\frac w 2\rfloor ,\lfloor\frac w 2\rfloor)\succ (w-2,1,1)\succ\\
&\,\succ\dots\succ (2,\underbrace{1,\dots,1}_{w-2 \text{ slots}}\succ (\underbrace{1,\dots,1}_{w \text{ slots}})\,.
\esp\eeq
As an example, in weight five we have the following non-decreasing partitions, ordered as
\beq
(5) \succ (4,1) \succ (3,2) \succ (3,1,1) \succ (2,2,1) \succ (2,1,1,1) \succ (1,1,1,1,1)\,.
\eeq

From these, we form the descending filtration alluded to above by taking the span of all shuffles of partition types greater than a given type $\la$, say, as
\beq
\cF_\la = \langle \la'-\text{shuffles in $V$ for all }\la\succ \la', \ \la\neq \la' \rangle\,.
\eeq
Hence $\cF_{(1,\dots,1)}$ is the zero space, $\cF_{(2,1,\dots,1)}$ is generated by all $(1,\dots,1)$-shuffles (i.e. fully symmetric tensors), $\cF_{(3,1,\dots,1)}$ is generated by all $(1,\dots,1)$-shuffles and $(2,1,\dots,1)$-shuffles etc. The proposition now guarantees that $\cF_\la\subset \ker \Pi_\la $.
Similarly, we put, for convenience, the ``shifted" sequence
\beq
\cF_\la^{\succeq} = \langle \la'-\text{shuffles in $V$ for all }\la\succeq \la' \rangle\,,
\eeq
so that $\cF_{\la^+}^{\succeq} = \cF_\la$ where $\la^+$ denotes the immediate successor of $\la$ in the lexicographic order above.

We proceed by induction on $\lambda$ with respect to the order $\succ$, starting with the shuffle $\la=(w)$. 
By the analysis above, we can put $S_{(w)}\equiv  \sum_i c_i \cS\big(b_i^{(w)}\big)$, provided we have been able to solve for the $c_i$. This gives us the basis step for the induction.
Now assume we have found an (integrable) tensor $S_\la$ ``approximating" $S$ in the sense that $S-S_\la\in 
\cF_{\la}$. Then we try to construct a ``better approximation" $S_{\la^+}$ for the successor partition $\la^+$ of $\la$ by finding integrable tensors $T_{\la^+} \in \cF_{\la^+}^\succeq = \cF_\la$ such that $S_{\la^+}\equiv S_\la+T_{\la^+}$ satisfies $ S-S_{\la^+} \in \cF_{\la^+}$.  If we are successful in finding such a $T_\la^+$, this finishes the induction step.

 We expect to be able to find such a $T_{\la^+}$ using (certain sums of products of) multiple polylogarithms, by selecting for each part $\la_r^+$ of $\la^+$ an ``indecomposable function type" of weight $\la_r^+$ (see table~\ref{tab:pure_trans_basis}) and taking their product. 
{In other words, we assume that the tensor $T_{\la^+}$ can be written, as a linear combination of $\la^+$-shuffles, in the form
\beq
T_{\la^+} = \sum_{i_1,\ldots,i_{l}}c_{i_1,\ldots,i_l}\,\cS\Big(b_{i_1}^{(\la^+_1)}\Big)\uplus\ldots\uplus\cS\Big(b_{i_l}^{(\la^+_l)}\Big)\,,
\eeq
 with $l=\ell(\la^+)$ and $\la^+=(\la^+_1,\ldots,\la^+_l)$.}
In the weight~5 example {of table~\ref{tab:pure_trans_basis}}, the  respective polylogarithmic functions we consider are
\beq\bsp
(5) &\,\leftrightarrow \textrm{Li}_5(R_1),\,\textrm{Li}_{2,3}(R_1,R_2)\,,\\
(4,1) &\,\leftrightarrow \textrm{Li}_4(R_1)\,\ln R_2,\,\textrm{Li}_{2,2}(R_1,R_2)\,\ln R_3\,,\\
(3,2) &\,\leftrightarrow \textrm{Li}_3(R_1)\,\textrm{Li}_2(R_2)\,,\\
(3,1,1) &\,\leftrightarrow \textrm{Li}_3(R_1)\,\ln R_2\,\ln R_3\,,\\
(2,2,1) &\,\leftrightarrow \textrm{Li}_2(R_1)\,\textrm{Li}_2(R_2)\,\ln R_3\,,\\
(2,1,1,1) &\,\leftrightarrow \textrm{Li}_2(R_1)\,\ln R_2\,\ln R_3\,\ln R_4\,,\\
(1,1,1,1,1) &\,\leftrightarrow \ln R_1\,\ln R_2\,\ln R_3\,\ln R_4\,\ln R_5\,,
\esp\eeq
where the $R_i$ correspond to the rational functions as indicated in section~\ref{sec:arguments}.

Since the proposition implies $\Pi_\la \cF_\la=0$, we get $\Pi_\la (S-S_\la) = 0$. Therefore we are essentially working at each step in the quotient space $\cF_\la^{\succeq}/\cF_\la$ and have to solve for considerably fewer coefficients than if we worked in $\cF_\la^{\succeq}$.
If at any stage we cannot find a $T_{\la^+}$ with the desired property, i.e. cannot solve for the corresponding coefficients using our limited spanning set of input functions $b_i^{(s)}$, the algorithm stops (we can of course try to rerun it with a larger set of input functions).
Otherwise, it ends with producing an integrable tensor $S_{(1,\dots,1)}$ with $S-S_{(1,\dots,1)}=0$, solving the main part of our integration problem.

\subsection{Elements in the kernel of the symbol map}
\label{sec:kernel}
The algorithmic approach we described in the previous section often allows us, given a tensor $S$, to construct a function $F$ such that $\cS(F) = S$. Let us now assume that the tensor $S$ was obtained in some way from an analytic expression $F_0$ (representing, say, a Feynman integral). It would be premature to conclude that the function $F$ we constructed is equal to the original expression $F_0$, because they are only equal up to terms that are in the kernel of $\cS$.

In the following we describe a way that, at least in most of the cases we studied so far, allows to fix this remaining ambiguity by parametrizing $F-F_0$ in a suitable way. The parametrization we are proposing takes the form
\beq
F-F_0 = \sum_{i}\,\tilde c_i\,\tilde k_i +\sum_\ell\sum_{\stackrel{i_1,\ldots,i_\ell}{w_1+\ldots+w_\ell=w}}c_{i_1\ldots i_\ell}\,k_{i_1\ldots i_\ell}\,b_{i_1}^{(w_1)}\ldots \,b_{i_\ell}^{(w_\ell)}\,,
\eeq
where $b_{i_k}^{(w_k)}$ are defined in the previous section, and $c_{i_1\ldots i_\ell}$ and $\tilde c_i$ are rational coefficients, and $k_{i_1\ldots i_\ell}$ and $\tilde k_i$ are generators of the kernel of the symbol map. Below we give a (non-exhaustive) list of such generators, which cover a wide range of applications. The free coefficients can then be fixed by considering special values for the variables $x_i$, e.g., $x_i=0$ or $x_i=1$, yielding a linear system for the coefficients.

In order for the above procedure to work we need to know the generators of the kernel of $\cS$. Even though this is a very difficult question to answer in general, we can compile a list of (presumably transcendental) numbers whose symbol should be defined as zero.
\begin{itemize}
\item $\cS(\log(-1))=0$. An associated polygon would be $\twogon{0}{-1}$, and this gives the symbol $-1$ which is zero (modulo 2-torsion which we ignore). In particular, this shows that the symbol does not detect the multivaluedness of the logarithm (let alone of the polylogarithms), and in particular does not suggest how to fix the branch cuts of the polylogarithms.
\item All multiple zeta values (MZV's) are in the kernel of $\cS$. Indeed, MZV's can be defined as the values in $x_i=1$ of the multiple polylogarithms $\textrm{Li}_{m_1,\ldots,m_k}(x_1,\ldots,x_k)$. Then it is easy to see that the associated decorated polygon will have all decorations equal to 0 or 1, and hence the symbol vanishes.
\item In addition, there are combinations of transcendental numbers that individually have a non-vanishing symbol\footnote{We recall that the ``refined" symbol of a constant is not necessarily zero.}, but there is a linear combination with zero symbol, e.g.,
\beq\label{eq:Li4(1/2}
\cS\left(\textrm{Li}_4\left({1\over 2}\right) +{1\over24}\,\ln^42\right) = -\left(1-{1\over2}\right)\otimes {1\over 2}\otimes {1\over 2}\otimes {1\over 2} + {1\over 24} (2\odot 2\odot2\odot2) =0\,.
\eeq
\end{itemize}
The previous example is a special case of a more general result for so-called colored multiple zeta values, defined by the alternating sums
\beq\label{eq:CMZV_def}
\zeta(m_1,\ldots,m_k;s_1,\ldots,s_k) = \sum_{0<n_1<n_2<\dots <n_k} \frac{s_1^{n_1} s_2^{n_2} \cdots s_k^{n_k} }{n_1^{m_1} n_2^{m_2} \cdots n_k^{m_k} }\,,
\eeq
with $m_i\in\mathbb{N}$ and $s_i\in\{\pm1\}$.
It is easy to check that $\zeta(1,1,1,1;-1,-1,1,1) = -\textrm{Li}_4\left({1\over 2}\right)$, and so eq.~\eqref{eq:Li4(1/2} can be written as
\beq
\cS(\zeta(1,1,1,1;-1,-1,1,1) - {1\over24}\,\ln^42) = 0\,.
\eeq
More generically, we have,
\begin{proposition}\label{proposition:CMZV}$\phantom{aaaaaaaaaaaaaaaaaaaaaaaaaaaaaaa}$
\begin{enumerate}
\item If at least one of the $m_i$ is different from $\pm1$ and $(m_1, s_1)\neq (1,1)$, then 
\beq
\cS(\zeta(m_1,\ldots,m_k;s_1,\ldots,s_k))=0\,.
\eeq
\item $\forall s_i \in \{-1,1\}$, $\forall m\ge1$, 
\beq
\cS\left(\zeta(\underbrace{1,\ldots,1}_{m \textrm{ times}};-1,s_2,\ldots,s_m) - \,{1\over m!}\ln^m{1\over2}\right) = 0.
\eeq
\end{enumerate}
\end{proposition}
\noindent The proof of this proposition will be given in appendix~\ref{app:CMZV_proof}. 

\section{Application: a spanning set for harmonic polylogarithms}
\label{sec:HPLbasis}
In this section we illustrate our approach by expressing all harmonic polylogarithms (HPL's) up to weight four in terms of a spanning set of functions. We start by determining the arguments our spanning set should depend upon. Indeed, in the case of harmonic polylogarithms we can classify all the arguments of the spanning set of functions under the mild assumption that all the arguments should be rational functions. Then, it is easy to see that the polygons associated to HPL's correspond to polygons where the root edge is decorated by the variable $x$, and all other sides are decorated only by $0$ or $\pm1$. This implies that the tensor associated to an HPL is an element of the tensor algebra of the vector space generated by the formal basis vectors $[x]$, $[1-x]$, $[1+x]$ and $[2]$. Indeed, the decorations of the polygon associated to an HPL are all $\pm1$ or $0$, except for the root which is decorated by the variable $x$. It is then easy to see that any dissection of this polygon will only involve the following five bigons
 \beq\bsp
 \mmu{\twogon{0}{x}} = x\,,&\quad  \mmu{\twogon{1}{x}} = 1-x\,,\\
   \mmu{\twogon{-1}{x}} = 1+x\,,&\quad
    \mmu{\twogon{1}{-1}} =  \mmu{\twogon{-1}{1}}=2\,.
    \esp\eeq 
Hence, the sets $P_{HPL}$ and $\overline{P}_{HPL}$ of irreducible polynomials defined in section~\ref{sec:higher_weights} are
 \beq
 P_{HPL}=\{2,x,1-x,1+x\} {\rm~~and~~} \overline{P}_{HPL}=P_{HPL}\cup\{2\pm x, 3\pm x,1\pm2x\}\,.
 \eeq
The most general rational function we can construct out of the irreducible polynomials in $\overline{P}_{HPL}$ then reads
\beq
\pm x^{\alpha_1}\,(1-x)^{\alpha_2}\,(1+x)^{\alpha_3}\,2^{\alpha_4}\,(2-x)^{\alpha_5}\,(2+x)^{\alpha_6}\,(3-x)^{\alpha_7}\,(3+x)^{\alpha_8}\,(1-2x)^{\alpha_9}\,(1+2x)^{\alpha_{10}}\,,
\eeq
with $\alpha_i\in\mathbb{Z}$. In the case of HPL's, however, it turns out (a posteriori) that we can restrict ourselves to the following set of rational functions
\beq
\cR_{HPL} = \{ \pm2^\delta\,x^\alpha\,(1-x)^\beta\,(1+x)^\gamma\, |\, \alpha, \beta, \gamma, \delta\in\mathbb{Z}\}\,.
\eeq
In the following we denote the elements of $\cR_{HPL}$ by
\beq\label{eq:Rabcd}
R_{\alpha\beta\gamma\delta}^{\pm}(x) = \pm2^\delta\,x^\alpha\,(1-x)^\beta\,(1+x)^\gamma\,,
\eeq
with $\alpha, \beta,  \gamma, \delta\in \mathbb{Z}$.
It follows then from the previous section that we should consider the subset of these rational functions (those contained in the set $\cR^{(1)}_{HPL}$; see eq.~\eqref{eq:1_minus_R} and eq.~\eqref{eq:1_minus_R_upper_1}) 
that satisfy the constraint
\beq\label{eq:arg_constraint}
1-R_{\alpha\beta\gamma\delta}^{\pm}(x) = R_{\alpha'\beta'\gamma'\delta'}^{s}(x)\,,
\eeq
for some integers $\alpha',\beta', \gamma'$ and $\delta'$, and $s=\pm 1$. A little algebra shows that quadruples $(\alpha, \beta, \gamma, \delta)$ are confined to the values given in Tab.~\ref{tab:arguments}. The first line corresponds to constant arguments, and will not be discussed any further. 
Note that we could also include the inverses of the arguments in Tab.~\ref{tab:arguments}.
Using the inversion formula for the classical polylogarithms, we can however always express $\textrm{Li}_n$ functions with inverted arguments in terms of polylogarithms taken at the arguments in Tab.~\ref{tab:arguments},
\beq
\textrm{Li}_n\left({1\over x}\right) = (-1)^{n-1} \textrm{Li}_n\left({x}\right) + \textrm{ products of lower weight terms},\quad n\ge2\,.
\eeq
Since furthermore the arguments in Tab.~\ref{tab:arguments} are less than 1 for $x\in[0,1]$, we will in the following only consider these functions as arguments of the $\textrm{Li}_n$ functions. Note however that even the functions in Tab.~\ref{tab:arguments} are not completely independent, because we can use the distribution formula for the $\textrm{Li}_n$ function to relate three of them, e.g.,
\beq\bsp
\textrm{Li}_n\left(x^2\right) & = 2^{n-1}\,\left(\textrm{Li}_n(x) + \textrm{Li}_n(-x)\right)\,,
\esp\eeq
and three others using the same equation with $x$ replaced by ${1-x\over1+x}$.
We finally arrive at the conclusion that, if we want to reduce all HPL's to some small (possibly minimal) set of (multiple) polylogarithms, the $\textrm{Li}_n$ sector of that set contains classical polylogarithms with 16 different arguments. As we will see below, for lower weights we can find more relations among the spanning set of functions, reducing its size even further.
In the following we will also have to deal with multiple polylogarithms of depth greater than one. From the previous section we know that the pairs of arguments
$(R_1,R_2)$ of these functions come from Tab.~\ref{tab:arguments}, subject to the additional constraint $R_1-R_2\in \cR_{HPL}$.

\begin{table}
\begin{center}
\begin{tabular}{|c|c|c|c|c|c||c|c|c|c|c|c|}
\hline
$s$ & $\alpha$ & $\beta$ & $\gamma$ & $\delta$ & $R_{\alpha\beta\gamma\delta}^{s}(x)$ & $s$ & $\alpha$ & $\beta$ & $\gamma$ & $\delta$ & $R_{\alpha\beta\gamma\delta}^{s}(x)$\\
\hline\hline
- & 0 & 0 & 0&0 & -1 & + & 0 & 0 & 0& -1 &1/2\\
+ & 1 & 0 & 0& 0 & $x$ & - & 1 & 0 & 0& 0 &$-x$\\
+ & 0 & 1 & 0& 0& $1-x$ & + & 0 & 0 & -1& 0& $1/(1+x)$\\
+ & 2 & 0 & 0& 0 & $x^2$& + & 0 & 1 & 1& 0 & $1-x^2$\\
- & 1 & -1 & 0& 0& $x/(x-1)$ & + & 1 & 0 & -1& 0 & $x/(x+1)$\\
+ & 0 & 1 & -1& 0 &  $(1-x)/(1+x)$ &- & 0 & 1 & -1& 0 & $(x-1)/(x+1)$\\
- & 2 & -2 & -2& 0& $x^2/(x^2-1)$ &+ & 0 & 2 & -2& 0& $(1-x)^2/(1+x)^2$\\
+ & 0 & 1 & 0& -1 & $(1-x)/2$ & + & 0 & 0 & 1& -1& $(1+x)/2$\\
- & 1 & -1 & 0& 1 & $2x/(x-1)$&+ & 1 & 0 & -1& 1 & $2x/(x+1)$\\
+ & 2 & 0 & -2& 2& $4x/(1-x)^2$ & - & 2 & -2 & 0& 2 & $-4x/(x+1)^2$\\
\hline
\end{tabular}
\caption{\label{tab:arguments}Solutions to the constraint~(\ref{eq:arg_constraint}). Only half of the solutions are shown, all other solutions being related by $R^s_{(-\alpha)(-\beta)(-\gamma)(-\delta)}(x) = 1/R_{\alpha\beta\gamma\delta}^{s}(x)$.}
\end{center}
\end{table}

Finally, using the results from the previous section, as well as some elementary identities among classical polylogarithms, we find the following spanning set of indecomposable functions up to weight four,
\begin{itemize}
\item for weight one,
\beq
\cB_1^{(1)}(x)= \ln x, \quad \cB_1^{(2)}(x)= \ln (1-x),\quad \cB_1^{(3)}(x)= \ln (1+x)\,, \quad \cB^{(4)}_1(x) = \ln 2\,,
\eeq
\item for weight two,
\beq
\cB_2^{(1)}(x)= \textrm{Li}_2(x), \quad \cB_2^{(2)}(x)= \textrm{Li}_2(-x),\quad \cB_2^{(3)}(x)= \textrm{Li}_2\left({1-x\over 2}\right)\,,
\eeq
\item for weight three,
\beq\bsp
&\cB_3^{(1)}(x)= \textrm{Li}_3(x), \quad \cB_3^{(2)}(x)= \textrm{Li}_3(-x),\quad \cB_3^{(3)}(x)= \textrm{Li}_3(1-x),\\
& \cB_3^{(4)}(x)= \textrm{Li}_3\left({1\over1+x}\right)\,,\cB_3^{(5)}(x)= \textrm{Li}_3\left({1+x\over2}\right), \quad  \cB_3^{(6)}(x)= \textrm{Li}_3\left({1-x\over2}\right),\\
&  \cB_3^{(7)}(x)= \textrm{Li}_3\left({1-x\over1+x}\right), \quad  \cB_3^{(8)}(x)= \textrm{Li}_3\left({2x\over x-1}\right)\,,
\esp\eeq
\item for weight four,
\beq\bsp
&\cB_4^{(1)}(x)= \textrm{Li}_4(x), \quad \cB_4^{(2)}(x)= \textrm{Li}_4(-x),\\
& \cB_4^{(3)}(x)= \textrm{Li}_4(1-x), \quad\cB_4^{(4)}(x)= \textrm{Li}_4\left({1\over1+x}\right),\\
&\cB_4^{(5)}(x)= \textrm{Li}_4\left({x\over x-1}\right), \quad  \cB_4^{(6)}(x)= \textrm{Li}_4\left({x\over x+1}\right),\\
&\cB_4^{(7)}(x)= \textrm{Li}_4\left({1+x\over2}\right), \quad  \cB_4^{(8)}(x)= \textrm{Li}_4\left({1-x\over2}\right),\\
&  \cB_4^{(9)}(x)= \textrm{Li}_4\left({1-x\over1+x}\right), \quad  \cB_4^{(10)}(x)= \textrm{Li}_4\left({x-1\over x+1}\right),\\
&\cB_4^{(11)}(x)= \textrm{Li}_4\left({2x\over x+1}\right),\quad \cB_4^{(12)}(x)= \textrm{Li}_4\left({2x\over x-1}\right),
\esp\eeq
\beq\bsp
&\cB_4^{(13)}(x)= \textrm{Li}_4\left(1-x^2\right),\quad \cB_4^{(14)}(x)= \textrm{Li}_4\left({x^2\over x^2-1}\right),\\
&\cB_4^{(15)}(x)= \textrm{Li}_4\left({4x\over (x+1)^2}\right)\,.
\esp\eeq
\end{itemize}
These functions are sufficient to express all harmonic polylogarithms up to weight three. Starting from weight four, we need to extend the set of functions by adjoining {\em multiple} polylogarithms. We find that it is enough to add the following three supplementary functions in order to express all harmonic polylogarithms up to weight four,
\beq
\cB_4^{(16)}(x) = \textrm{Li}_{2,2}(-1,x),\quad \cB_4^{(17)}(x) = \textrm{Li}_{2,2}\left({1\over 2},{2x\over x+1}\right),\quad \cB_4^{(18)}(x) = \textrm{Li}_{2,2}\left({1\over 2},{2x\over x-1}\right)\,.
\eeq
Note that, if the vector of singularities only takes values in the set $\{0,1\}$, we can restrict ourselves to the smaller spanning set,
\begin{itemize}
\item for weight one: $\cB_1^{(1)}(x)$, $\cB_1^{(2)}(x)$,
\item for weight two: $\cB_2^{(1)}(x)$, 
\item for weight three: $\cB_3^{(1)}(x)$, $\cB_3^{(3)}(x)$,
\item for weight four: $\cB_4^{(1)}(x)$, $\cB_4^{(3)}(x)$, $\cB_4^{(5)}(x)$.
\end{itemize}
Our choice for the spanning set is of course not unique, and we might have chosen a different set of functions, related to the $\cB_i^{(j)}$ functions via functional equations. Our choice was motivated by the fact that  $\cB_i^{(j)}(x)$ is manifestly real for $x\in[0,1]$. Note that outside this interval, the branching structure of these functions can be more complicated. This issue is addressed in appendix~\ref{app:anal_cont}.

\subsection{Example}
Let us conclude this section by giving an example of how we can apply the procedure of section~\ref{sec:higher_weights} to express a generic HPL of weight four in terms of the functions $\cB^{(i)}_j$. We discuss in detail the example of $H(0,0,1,1;x) = S_{2,2}(x)$, all other cases being similar. For the list of all results, we refer to appendix~\ref{app:results}. We start by deriving the tensor associated to $H(0,0,1,1;x)$. The polygon associated to $H(0,0,1,1;x) = G(0,0,1,1;x)$ is
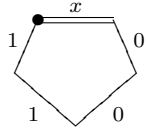
\beq\nonumber
\fivegon{1}{1}{0}{0}{x}
\eeq
There is only one relevant maximal dissection of this pentagon (all other dissections give rise to twogons with decorations 0 and / or 1 which vanish by definition),
\beq\nonumber
\fivegonarrowtCDa{1}{1}{0}{0}{x}  \quad\leftrightarrow\quad \mmu{\twogon{1}{x}}\otimes\mmu{\twogon{1}{x}}\otimes\mmu{\twogon{0}{x}}\otimes\mmu{\twogon{0}{x}}\,.
\eeq
From this dissection we can immediately read of the symbol of $H(0,0,1,1;x)$ as
\beq\label{eq:S_22_tensor}
\cS(H(0,0,1,1;x)) = (1-x)\otimes(1-x)\otimes x\otimes x\,.
\eeq
Note that in general the symbol of a harmonic polylogarithm $H(a_1,\ldots,a_w;x)$, with $a_i\in\{0,1\}$ is simply given by $(-1)^k\,(a_w-x)\otimes\ldots\otimes(a_1-x)$, where $k$ is the number of $a_i$'s equal to 1.
{Before turning to the question of how to express $H(0,0,1,1;x)$ in terms of the spanning set, let us review how we could have obtained the symbol~\eqref{eq:S_22_tensor} using the recursive definition of the symbol~\eqref{eq:GSVV_def}.
Note however, that in this case we cannot apply eq.~\eqref{eq:MPL_diff_eq} immediately, as the arguments of $G(0,0,1,1;x) = H(0,0,1,1;x)$ are not generic and we would hence obtain divergences in the right-hand side of eq.~\eqref{eq:MPL_diff_eq}. We therefore need to use a regularized version of the differential equation~\eqref{eq:MPL_diff_eq}~\cite{Goncharov:2001},
\beq
\rd H(0,0,1,1;x)=H(0,1,1;x)\,\rdh\log x\,,
\eeq
and so
\beq
\cS(H(0,0,1,1;x)=\cS(H(0,1,1;x))\otimes x\,.
\eeq
We can iterate this procedure to compute the symbol of $H(0,1,1;x)$. The differential equation for $H(0,1,1;x)$ is
\beq
\rd H(0,1,1;x) = H(1,1;x)\,\rdh \log x\,,
\eeq
and so we get
\beq
\cS(H(0,1,1;x)) = \cS(H(1,1;x))\otimes x\,.
\eeq
The symbol of $H(1,1;x)={1\over2}\log^2(1-x)$ is easy to obtain from eq.~\eqref{producttoshuffle},
\beq
\cS(H(1,1;x))={1\over2}\cS(\log^2(1-x)) = {1\over2}\cS(\log(1-x))\uplus\cS(\log(1-x)) = (1-x)\otimes(1-x)\,.
\eeq
Putting everything together, we immediately arrive at the symbol given in eq.~\eqref{eq:S_22_tensor}.
}

{Let us now turn to the actual problem we want to study, namely how to express $H(0,0,1,1;x)$ in terms of the spanning set for HPL's defined in the previous section.}
The goal is to find rational numbers $c_{i_1\ldots i_\ell}^{(k)}$ such that the tensor associated to 
\beq\bsp\label{eq:S22_ansatz}
\sum_{i=1}^{18}&c_{i}^{(1)}\,\cB_4^{(i)}(x) + 
\sum_{i=1}^8\sum_{j=1}^4c_{ij}^{(2)}\,\cB_3^{(i)}(x)\,\cB_1^{(j)}(x) +
\sum_{i,j=1}^3c_{ij}^{(3)}\,\cB_2^{(i)}(x)\,\cB_2^{(j)}(x) \\
&\,+ 
\sum_{i=1}^3\sum_{j,k=1}^4c_{ijk}^{(4)}\,\cB_2^{(i)}(x)\,\cB_1^{(j)}(x)\,\cB_1^{(k)}(x) 
+ 
\sum_{i,j,k,l=1}^4c_{ijkl}^{(5)}\,\cB_1^{(i)}(x)\,\cB_1^{(j)}(x)\,\cB_1^{(k)}(x)\,\cB_1^{(l)}(x)\,, 
\esp\eeq
equals the tensor given in eq.~(\ref{eq:S_22_tensor}). 
The symbol of  eq.~\eqref{eq:S22_ansatz} is easily obtained from the fact that $\cS$ is linear and maps products of polylogarithms to shuffles, i.e.,
\beq\bsp\label{eq:S22_ansatz_2}
\sum_{i=1}^{18}&c_{i}^{(1)}\,\cS(\cB_4^{(i)}(x)) + 
\sum_{i=1}^8\sum_{j=1}^4c_{ij}^{(2)}\,\cS(\cB_3^{(i)}(x))\sha\cS(\cB_1^{(j)}(x))\\
+&\,
\sum_{i,j=1}^3c_{ij}^{(3)}\,\cS(\cB_2^{(i)}(x))\sha\cS(\cB_2^{(j)}(x))  \\
+&\,\sum_{i=1}^3\sum_{j,k=1}^4c_{ijk}^{(4)}\,\cS(\cB_2^{(i)}(x))\sha\cS(\cB_1^{(j)}(x))\sha\cS(\cB_1^{(k)}(x)) \\
+&\,
\sum_{i,j,k,l=1}^4c_{ijkl}^{(5)}\,\cS(\cB_1^{(i)}(x))\sha\cS(\cB_1^{(j)}(x))\sha\cS(\cB_1^{(k)}(x))\sha\cS(\cB_1^{(l)}(x))\,.
\esp\eeq
The symbol of each function $\cB_j^{(i)}$ can be easily obtained, e.g.,
\beq\bsp
\cS(\cB_4^{(6)}(x)) =&\, \cS\left(\textrm{Li}_4\left({x\over x+1}\right)\right)  =-\left(1-{x\over x+1}\right)\otimes \left({x\over x+1}\right)\otimes \left({x\over x+1}\right)\otimes \left({x\over x+1}\right)\\
=&\, -\left({1\over x+1}\right)\otimes\left({x\over x+1}\right)\otimes\left({x\over x+1}\right)\otimes \left({x\over x+1}\right)\\
=&\, (x+1)\otimes x\otimes x\otimes x-(x+1)\otimes x\otimes x\otimes (x+1)-(x+1)\otimes x\otimes (x+1)\otimes x\\
&\,+(x+1)\otimes x\otimes (x+1)\otimes (x+1)-(x+1)\otimes (x+1)\otimes x\otimes x\\
&\,+(x+1)\otimes (x+1)\otimes x\otimes (x+1)+(x+1)\otimes (x+1)\otimes (x+1)\otimes x\\
&\,-(x+1)\otimes (x+1)\otimes (x+1)\otimes (x+1)\,.
\esp\eeq
The different shuffles in eq.~\eqref{eq:S22_ansatz_2} can be distinguished further by acting with the projectors defined in the previous section. In particular, we obtain
\beq
\Pi_4\,\cS(H(0,0,1,1;x)) = \sum_{i=1}^{18}c_{i}^{(1)}\,\Pi_4\,\cS(\cB_4^{(i)}(x))\,.
\eeq
Equating the coefficients of the different elementary tensors on both sides of this equation, we obtain a linear system that allows us to solve for the for the coefficients $c_i^{(1)}$, $1\le i\le18$. The solution is easily obtained, and is given by
\beq
c_1^{(1)} = -c_3^{(1)} = c_5^{(1)} = 1\,,
\eeq
and all other coefficients are zero. We now proceed recursively, and subtract the (integrable) tensor arising from the solution we have found. By construction, the symbol of this difference must vanish under the action of the projector $\Pi_4$,
\beq
\Pi_4\,\cS\left(H(0,0,1,1;x) - L(x) \right) = 0\,,
\eeq
where we defined
\beq
L(x) = \textrm{Li}_4(x)-\textrm{Li}_4(1-x) + \textrm{Li}_4\left({x\over x-1}\right)\,.
\eeq
We next turn to the determination of the coefficients $c_{ij}^{(2)}$, i.e., coefficients of terms of the form $\cB_3^{(i)}(x)\cB_1^{(j)}(x)$. We can isolate these terms by applying the projector $\Pi_3\otimes\Pi_1$,
\beq\bsp
(\Pi_3&\,\otimes \Pi_1)\,\cS\left(H(0,0,1,1;x) -L(x)\right)\\
=&\,
\sum_{i,j=1}^3c_{ij}^{(2)}\,(\Pi_3\otimes \Pi_1)[\cS(\cB_3^{(i)}(x))\sha\cS(\cB_1^{(j)}(x))]\\
=&\,
\sum_{i,j=1}^3c_{ij}^{(2)}\,[\Pi_3\cS(\cB_3^{(i)}(x))]\otimes[\Pi_1\cS(\cB_1^{(j)}(x))]\,,
\esp\eeq
and equating the coefficients of the elementary tensors on both sides of the equation we can solve for the coefficients $c_{ij}^{(2)}$. We find
 $c_{12}^{(2)} = -1$, and $c_{ij}^{(2)} = 0$ if $(i,j) \neq(1,2)$. We again subtract this contribution to find an expression that vanishes under the actions of both $\Pi_4$ and $\Pi_3\otimes\Pi_1$. Next we act with the projector $\Pi_2\otimes\Pi_2$ on this difference
 \beq
 (\Pi_2\otimes\Pi_2)\cS\left(H(0,0,1,1;x) -L(x) + \textrm{Li}_3(x)\ln(1-x)\right) = 0\,,
 \eeq
 and we immediately conclude that $c_{ij}^{(3)} = 0$, $\forall 1\le i,j\le 3$\,. Similarly, by acting with the projector $\Pi_2\otimes\Pi_1\otimes\Pi_1$ we conclude that all the coefficients $c_{ijk}^{(4)}$ must vanish. The remaining terms are thus all associated to products of logarithms, which can immediately be read off from the tensor
 \beq\bsp
\cS&\left(H(0,0,1,1;x) -L(x)+ \textrm{Li}_3(x)\ln(1-x)\right)\\
&\, = {1\over 24}\,(1-x)\odot(1-x) \odot (1-x) \odot (1-x) -{1\over6}\,x\odot(1-x) \odot (1-x) \odot (1-x)\\
&\, = \cS\left({1\over24}\,\ln^4(1-x) - {1\over6}\,\ln x\,\ln^3(1-x)\right)\,.
\esp\eeq
At this stage, we have found a combination of (a product of)  functions in our spanning set that has the same symbol as $H(0,0,1,1;x)$, and so the quantities are equal up to terms that are mapped to zero by $\cS$.
We make an ansatz assuming we have found sufficiently many elements in the kernel of $\cS$, as
\beq\bsp\label{eq:kernel_W4}
a_1&\,\left(\textrm{Li}_4\left({1\over 2}\right) +{1\over24}\,\ln^42\right) + a_2\,\pi^4 + \zeta_3\,\sum_{i=1}^4b_i\,\cB_1^{(i)}(x) \\
&\,+ 
\pi^2\,\left(\sum_{i=1}^3c_i\,\cB_2^{(i)}(x) + \sum_{i,j=1}^4c_{ij}\,\cB_1^{(i)}(x)\,\cB_1^{(j)}(x)\right)\,,
\esp\eeq
where $a_i$, $b_i$, $c_i$ and $c_{ij}$ are rational numbers, and $\zeta_3=\zeta(3)$ denotes the value in $s=3$ of the Riemann $\zeta$ function
\beq
\zeta(s) = \sum_{n=1}^\infty{1\over n^s}\,. 
\eeq
The coefficients can now be fixed by looking at particular values of $x$. In particular, harmonic polylogarithms are known analytically up to weight four for $x=0$, $x=\pm1$ and $x=\pm{1\over2}$. It turns out that in all cases these values are enough to fix all the free coefficients. In the case of $H(0,0,1,1;x)$, we finally arrive at
\beq\bsp
H(0,0,1,1;x) &\, = S_{2,2}(x)\\
&\,= -\text{Li}_4(1-x)+\text{Li}_4(x)+\text{Li}_4\left(\frac{x}{x-1}\right)\\
&\,-\text{Li}_3(x) \ln (1-x)+\frac{1}{24} \ln^4(1-x)-\frac{1}{6} \ln x \ln^3(1-x)\\
&\,+\zeta_3 \ln (1-x)+\frac{1}{12} \pi^2 \ln^2(1-x)+\frac{\pi ^4}{90}\,.
\esp\eeq

\section{Conclusion}
In this paper, we have provided a review of the symbol map, a linear map that associates to a multiple polylogarithm of weight $n$ an $n$-fold tensor in a way that captures many of the combinatorial properties of polylogarithms, and also respects the functional equations they satisfy. 
While so far the symbol map was defined recursively via iterated differentials, we have given a diagrammatic rule where the symbol of a multiple polylogarithm is obtained directly via a weighted sum over the maximal dissections of the decorated polygon associated to the polylogarithm introduced in ref.~\cite{GGL:2009}. 

Furthermore, we have addressed the problem of integrating a symbol, i.e., the problem of 
finding a function whose symbol matches a given tensor that satisfies the integrability condition~\eqref{eq:integrability}. We have presented a systematic approach of how to find a candidate spanning set for such a function. Once this candidate spanning set has been constructed, and working under the assumptions that its elements 
suffice to express the integrated symbol in, we showed how a set of projectors can be defined which help to find a function whose symbol matches the initial tensor. While our approach falls short of a complete algorithmic proof and 
is surely not adequate in all possible scenarios, we nevertheless believe that it can be applied in many situations, as was for example demonstrated in ref.~\cite{DelDuca:2011ne,DelDuca:2011jm,DelDuca:2011wh} where our method was successfully applied to obtain new compact analytic results for certain one-loop hexagon integrals in $D=6$ dimensions. Finally, we have used our approach to derive a spanning set for harmonic polylogarithms up to weight 4, and this spanning set was recently used to obtain an efficient numerical implementation of harmonic polylogarithms up to weight four~\cite{Buehler:2011ev}.

\section*{Acknowledgements}
CD is grateful to B.~Anastasiou, S.~Buehler, P.~Heslop, M.~Spradlin, C.~Vergu and A.~Volovich for useful discussions. HG expresses his thanks to F.~Brown and D.~Kreimer for helpful comments. The authors thank E.~W.~N.~Glover for helpful remarks on a preliminary version of the text. JR would like to thank EPSRC for their support. This work was supported by the Research Executive Agency (REA) of the European Union under the Grant Agreement number PITN-GA-2010-264564 (LHCPhenoNet).




\appendix

\section{Review on shuffle algebras}
\label{app:algebras}
As shuffle algebras are a recurrent theme when working with multiple polylogarithms and their symbols, we review in this appendix the most important notions. Before turning to the special case of shuffle algebras, we first review some basic notions about algebras in general.

\paragraph{Algebras over a field.} An algebra $\cA$ over a field $F$ ($F=\mathbb{R}$ or $F=\mathbb{C}$, say) is a vector space\footnote{More generally, we could consider $\cA$ to be a module over a ring $R$.}  over $F$ together with an associative  and distributive multiplication $\cA\otimes\cA\to\cA$. In the case the multiplication has a unit element, the algebra is called unital. Furthermore, an algebra is said to be \emph{graded} if $\cA$ can be written as a direct sum as a vector space,
\beq
\cA=\bigoplus_{n\in I}\cA_n\,,
\eeq
and if $\forall a\in \cA_m$ and $\forall b\in\cA_n$, we have 
\beq
a\cdot b\in \cA_{m+n}\,.
\eeq
One of the most prominent representatives of a graded algebra is the tensor algebra associated to an $F$-vector space $V$, defined by
\beq
\cT(V) = \bigoplus_{n=0}^\infty\cT_n(V)\,,
\eeq
where $\cT_0(V) = F$ and $T_1(V) = V$, and for $n\ge 2$ we define
\beq
\cT_n(V) = \underbrace{V\otimes\ldots\otimes V}_{n \textrm{ times}}\,.
\eeq
The multiplication on $\cT(V)$ is defined on elementary tensors $a_1\otimes\ldots\otimes a_n$ by
\beq\bsp\label{eq:tensor_prod}
\cT_m(V)\otimes\cT_n(V) &\,\to \cT_{m+n}(V)\\
(a_1\otimes\ldots\otimes a_m)\otimes (b_1\otimes\ldots\otimes b_n) &\,\mapsto a_1\otimes\ldots\otimes a_m\otimes b_1\otimes\ldots\otimes b_n\,,
\esp\eeq
making the tensor algebra into a \emph{graded} algebra in an obvious way. Furthermore, the tensor algebra is also unital, the unit being the unit $1\in F$.

A \emph{homomorphism} between two algebras $\cA$ and $\cB$ is a linear map $\phi$ that preserves the multiplication, i.e., a linear map $\phi$ such that $\forall a,b\in \cA$, $\phi(a\cdot b) = \phi(a)\cdot\phi(b)$.

\paragraph{Ideals in algebras.} An \emph{ideal} in an algebra $\cA$ (or more generally in a ring) is an additive subgroup $\cI$ of $\cA$ such that
\beq
\forall a\in\cA, \forall b\in\cI, \quad a\cdot b\in\cI {\rm~~and~~} b\cdot a\in \cI\,.
\eeq
An easy example of an ideal is given by considering the ring of integer numbers $\mathbb{Z}$ and its subset $n\mathbb{Z}$, $n\in \mathbb{Z}$, i.e., the set of all integer multiples of $n$. The set $n\mathbb{Z}$ is obviously an additive subgroup of $\mathbb{Z}$, and every time we multiply an element of $n\mathbb{Z}$ by an integer number, we obtain another multiple of $n$. Hence $n\mathbb{Z}$ is an ideal of $\mathbb{Z}$. Another example of an ideal is the kernel of an algebra homomorphism $\phi$. Indeed, $\textrm{Ker}\,\phi$ is a sub-vector space of $\cA$, and hence an additive subgroup. Furthermore, $\forall a\in\cA$ and $\forall b\in\textrm{Ker}\,\phi$, we have
\beq
\phi(a\cdot b)=\phi(a)\cdot\phi(b) = \phi(a)\cdot 0 = 0\,,
\eeq
and so $a\cdot b\in \textrm{Ker}\,\phi$, making $\textrm{Ker}\,\phi$ into an ideal in $\cA$.

\paragraph{Shuffle algebras.} After this rather general discussion on algebras, let us from now on focus exclusively on the example of shuffle algebras. As a starting point, let us consider a set $\cL$, whose elements we will refer to as \emph{letters}, and consider the set $\cW$ of all words constructed from elements in $\cL$, i.e., the set of all possible concatenations of letters in $\cL$, together with the \emph{empty word} $\varepsilon$, consisting of no letters (more precisely, $\cW$ is the free monoid generated by the elements in $\cL$). We can define a multiplication on $\cW$ given by concatenation of words, the empty word being the unit element.

Let us now consider the vector space $\cA$ over some field $F$ given by all formal linear combinations of words in $\cW$. $\cA$ is then in fact an algebra over $F$, the multiplication given simply by the concatenation of words. Furthermore, it is easy to see that $\cA$ is also graded by the length of the words (the concatenation of two words with length $m$ and $n$ gives a word of length $m+n$).

We can define another multiplication on $\cA$, the so-called \emph{shuffle product}, defined on words by
\beq\label{eq:shuffle_def}
(a_1\ldots a_m)\uplus(a_{m+1}\ldots a_{m+n}) = \sum_{\sigma\in\Sigma(n_1, n_2)}\,a_{\sigma^{-1}(1)}\ldots a_{\sigma^{-1}(n_1+n_2)},\\
      \eeq
where $\Sigma(n_1,n_2)$ denotes the set of all shuffles of $n_1+n_2$ elements, i.e., the subset of the symmetric group $S_{n_1+n_2}$ defined by
\beq
\Sigma(n_1,n_2) = \{\sigma\in S_{n_1+n_2} |\, \sigma^{-1}(1)<\ldots<\sigma^{-1}({n_1}) {\rm~~and~~} \sigma^{-1}(n_1+1)<\ldots<\sigma^{-1}(n_1+{n_2})\}\,.
\eeq
The vector space $\cA$ together with the shuffle product is called a \emph{shuffle algebra}. A shuffle algebra is again graded by the length of the words, and the empty word $\varepsilon$ is the unit element of the shuffle algebra. The definition~\eqref{eq:shuffle_def} of the shuffle product is equivalent to the following recursive definition, $\forall x, y\in\cL$, $\forall u,v\in \cW$, 
\beq\bsp
\varepsilon\uplus u &\,= u\uplus \varepsilon = u\,,\\
(xu)\uplus(yv) &\,= x(u\uplus(yv)) + y((xu)\uplus v)\,.
\esp\eeq
Note that the tensor algebra of a vector space $V$ can be equipped with a shuffle product in a natural way: the set of letters $\cL$ is simply a basis of $V$, the set of words corresponds to the elementary tensors $a_1\otimes\ldots\otimes a_n$, the empty word is simply the scalar 1, and the concatenation of words corresponds to the multiplication~\eqref{eq:tensor_prod} of two tensors.

In section~\ref{sec:polylogs} we have encountered another example of a shuffle algebra, the shuffle algebra of multiple polylogarithms, eq.~\eqref{eq:G_shuffle}. In that case letters are the elements $a_i$ of the vector of singularities $(a_1,\ldots,a_n)$, the latter being the words. Concatenation is simply defined by the concatenation of the vectors of singularities. The length of a word corresponds to the weight of the polylogarithm, i.e., the number of components of the vector of singularities. In other words, the shuffle algebra of polylogarithms is graded by the weight.

\section{Selected examples of symbols}
In this appendix we compile a list of the symbols of the most commonly used (multiple) polylogarithms. First of all, the symbol of an ordinary logarithm is simply the argument of the logarithm,
\beq\label{eq:log_symbol}
\cS(\ln x) = x\,.
\eeq
From eq.~\eqref{producttoshuffle} it follows then, for every non-negative integer $n$,
\beq\label{eq:log_power_symbol}
\cS\left({1\over n!}\ln^n x\right) = \underbrace{x\otimes\ldots\otimes x}_{n\textrm{ times}} \equiv x^{\otimes n}\,.
\eeq

The symbol of the classical polylogarithms have a similarly simple form, i.e., for every non-negative integer $n$ we obtain
\beq\label{eq:Lin_symbol}
\cS\left(\textrm{Li}_n(x)\right) = -(1-x)\otimes x^{\otimes (n-1)} = -(1-x)\otimes\underbrace{x\otimes\ldots\otimes x}_{(n-1)\textrm{ times}}\,,
\eeq
where we used the notation of eq.~\eqref{eq:log_power_symbol}. Note that for $n=1$, the classical polylogarithm can be expressed as an ordinary logarithm, $\textrm{Li}_1(x) = -\ln(1-x)$, which is consistent with the symbols given in Eqs.~\eqref{eq:log_symbol} and~\eqref{eq:Lin_symbol},
\beq
\cS\left(\textrm{Li}_1(x)\right) = -(1-x) = \cS(-\ln(1-x))\,.
\eeq
Finally, the symbol of a Nielson polylogarithm reads,
\beq\label{eq:Snp_symbol}
\cS(S_{n,p}(x)) = (-1)^p\,(1-x)^{\otimes p}\otimes x^{\otimes n} = (-1)^p\,\underbrace{(1-x)\otimes\ldots\otimes (1-x)}_{p\textrm{ times}}\otimes\underbrace{x\otimes\ldots\otimes x}_{n\textrm{ times}}\,.
\eeq
Again we note that the Nielson polylogarithms contain the classical polylogarithms as a special case, $S_{n-1,1}(x) = \textrm{Li}_n(x)$, an identity which is easily verified at the level of the symbols~\eqref{eq:Lin_symbol} and~\eqref{eq:Snp_symbol}.

\def \HPLgon#1{{
\xy
\POS(0,0) \ar@{=} (25,0)
\POS(35,0) \ar@{=} (60,0)
\POS(25,0) \ar@{.} (35,0)^{#1}
\POS(0,0) *+{\bullet}
\POS(0,0) \ar@{-} (-7,-7)_{a_n}
\POS(-7,-7) \ar@{-} (-7,-17)_{a_{n-1}}
\POS(-7,-17) \ar@{-} (0,-24)_{a_{n-2}}
\POS(60,0) \ar@{-} (67,-7)^{a_1}
\POS(67,-7) \ar@{-} (67,-17)^{a_2}
\POS(67,-17) \ar@{-} (60,-24)^{a_3}
\POS(60,-24) \ar@{.} (0,-24)

\POS(-7,-7)\ar@{->>} (10,0)
\POS(-7,-17)\ar@{->>} (20,0)
\POS(67,-17)\ar@{->>} (40,0)
\POS(67,-7)\ar@{->>} (50,0)

\endxy
}}

The previous examples of the classical and Nielson polylogarithms are both just special cases of harmonic polylogarithms where the components of the vector of singularities take values in $\{0,1\}$,
\beq
\textrm{Li}_n(x) = H(\vec 0_{n-1},1;x) {\rm~~and~~} S_{n,p}(x) = H(\vec 0_n,\vec 1_p;x)\,.
\eeq
The symbol of a harmonic polylogarithm $H(a_1,\ldots,a_n;x)$, with $a_i\in\{0,1\}$, can be written in the compact form
\beq\label{eq:HPL_symbol}
\cS(H(a_1,\ldots,a_n;x)) = (-1)^k\,(a_n-x)\otimes\ldots\otimes(a_1-x)\,,
\eeq
where $k$ is the number of components in the vector of singularities $(a_1,\ldots,a_n)$ equal to $1$. Indeed, the polygon $P(a_n,\ldots,a_1,x)$ associated to $G(a_1,\ldots,a_n;x) = (-1)^k\,H(a_1,\ldots,a_n;x)$ has the root side decorated by $x$, and all other sides decorated by $0$ or $1$. It is easy to see that the only relevant maximal dissection of such a polygon is
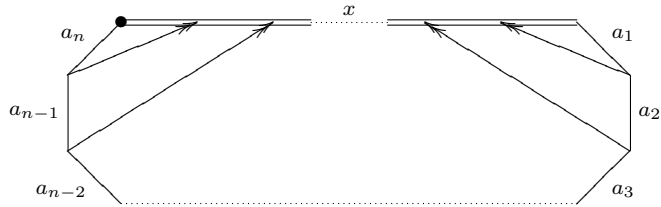
\begin{center}
\hskip 1.25cm
\HPLgon{x}
\end{center}
All other maximal dissection give rise to bigons decorated only by 0 and / or 1, which give a zero contribution to the symbol. The non-vanishing dissection produces a term in the symbol given by
\beq
\mmu{\twogon{a_n}{x}}\otimes\mmu{\twogon{a_{n-1}}{x}}\otimes\ldots\otimes\mmu{\twogon{a_2}{x}}\otimes\mmu{\twogon{a_{1}}{x}}\,,
\eeq
which is equal to the tensor in the right-hand side of eq.~\eqref{eq:HPL_symbol} (apart from the sign).

Generic harmonic polylogarithms where the components of the vector of singularities take values in $\{-1,0,1\}$ do not admit a compact expression for the symbol. In this case, the symbol is however easily obtained from the symbols of generic multiple polylogarithms, which are reviewed up to weight four in the next subsections.

\subsection{The symbol of a generic multiple polylogarithm of weight one}
A generic multiple polylogarithm of weight one can be written as $G(a;x)$, with $a,x\in \mathbb{C}$. We can associate a bigon to this function via
\beq
G(a;x) \leftrightarrow P(a,x) = \twogon{a}{x}\,.
\eeq
The symbol of $G(a;x)$ is then
\beq
\cS(G(a;x)) = \mmu{\twogon{a}{x}}\,.
\eeq

\subsection{The symbol of a generic multiple polylogarithm of weight two}
To a generic multiple polylogarithm $G(a,b;x)$ of weight two we associate a trigon
\beq\label{eq:Gabz}
G(a,b;x) \leftrightarrow P(b,a,x) = \threegon{b}{a}{x}\,.
\eeq
The symbol of $G(a,b;x)$ is then obtained by looking at all the maximal dissections of the trigon, obtained by inserting a single arrow. In the following we give the three maximal dissections, together with the term in the symbol they correspond to. We use the shorthand
\beq
ab|cd \equiv \mmu{\twogon{a}{b}}\otimes\mmu{\twogon{c}{d}}\,.
\eeq
The three maximal dissections of the trigon in eq.~\eqref{eq:Gabz}, together with the term in the symbol $\cS(G(a,b;x))$ they correspond to, are
\begin{center}
\begin{tabular}{ccccc}
\threegonarrowb{b}{a}{x} &$\phantom{aaaaaa}$& \threegonarrowc{b}{a}{x} & $\phantom{aaaaaa}$&\threegonarrowa{b}{a}{x}\\
\\
$+ax|ba$ & &$+bx|ax$ && $-bx|ab$
\end{tabular}
\end{center}

\subsection{The symbol of a generic multiple polylogarithm of weight three}
To a generic multiple polylogarithm $G(a,b,c;x)$ of weight three we associate a tetragon
\beq\label{eq:Gabcz}
G(a,b,c;x) \leftrightarrow P(c,b,a,x) = \fourgon{c}{b}{a}{x}\,.
\eeq
The symbol of $G(a,b,c;x)$ is then obtained by looking at all the maximal dissections of the tetragon, obtained by inserting two non-intersecting arrows. In the following we give the twelve maximal dissections, together with the term in the symbol they correspond to. We use the shorthand
\beq
ab|cd|ef \equiv \mmu{\twogon{a}{b}}\otimes\mmu{\twogon{c}{d}}\otimes\mmu{\twogon{e}{f}}\,,
\eeq
as well as the notation for shuffles,
\beq
A|B\uplus C \equiv A|B|C + A|C|B\,.
\eeq
The twelve maximal dissections of the tetragon in eq.~\eqref{eq:Gabcz}, together with the term in the symbol $\cS(G(a,b,c;x))$ they correspond to, are
\begin{center}
\begin{tabular}{cccc}
\xy 
\fourgona \POS(4,5) \ar@{<<-} +(-4,-10)  \POS(6,5) \ar@{<<-} +(4,-10) 
\endxy &
\xy
\fourgona \POS(10,1) \ar@{<<-} +(-10,4) \POS(10,-1) \ar@{<<-} +(-10,-4)
\endxy
&\xy
\fourgona \POS(6,-5) \ar@{<<-} +(4,10) \POS(4,-5) \ar@{<<-} +(-4,10)
\endxy
&\xy
\fourgona \POS(0,-1) \ar@{<<-} +(10,-4) \POS(0,1) \ar@{<<-} +(10,4)
\endxy
\\
\\
$+cx|bx|ax$&$+ax|ca|ba$&$-bx|ab\uplus cb$&$+cx|ac|bc$\\
\\
\xy 
\fourgona \POS(5,5) \ar@{<<-} +(-5,-10) \POS(10,0) \ar@{<<-} +(-10,-5)
\endxy
&\xy
\fourgona \POS(10,0) \ar@{<<-} +(-10,5)\POS(5,-5) \ar@{<<-} +(-5,10)
\endxy
&\xy\fourgona \POS(5,-5) \ar@{<<-} +(5,10) \POS(0,0) \ar@{<<-} +(10,5)
\endxy
&\xy\fourgona \POS(0,0) \ar@{<<-} +(10,-5) \POS(5,5) \ar@{<<-} +(5,-10) 
\endxy
\\
\\
$+cx|ax|ba$&$+ax|ba|cb$&$+cx|bc|ab$&$-cx|ax\uplus bc$\\
\\
\xy 
\fourgona \POS(5,5) \ar@{<<-} +(-5,-10) \POS(5,-5) \ar@{<<-} +(5,10) 
\endxy
&\xy\fourgona \POS(10,0) \ar@{<<-} +(-10,5) \POS(0,0) \ar@{<<-} +(10,-5)
\endxy
&\xy\fourgona \POS(5,5) \ar@{<<-} +(5,-10) \POS(5,-5) \ar@{<<-} +(-5,10)
\endxy
&\xy\fourgona \POS(10,0) \ar@{<<-} +(-10,-5)\POS(0,0) \ar@{<<-} +(10,5)
\endxy
\\
\\
$-cx|bx|ab$&$-ax|ca|bc$&$bx|ax\uplus cb$&$-cx|ac|ba$
\end{tabular}
\end{center}

\subsection{The symbol of a generic multiple polylogarithm of weight four}
To a generic multiple polylogarithm $G(a,b,c,d;x)$ of weight two we associate a pentagon
\beq\label{eq:Gabcdz}
G(a,b,c,d;x) \leftrightarrow P(d,c,b,a,x) = \fivegon{d}{c}{b}{a}{x}\,.
\eeq
The symbol of $G(a,b,c,d;x)$ is then obtained by looking at all the maximal dissections of the pentagon, obtained by inserting three non-intersecting arrows. In the following we give the 55 maximal dissections, together with the term in the symbol they correspond to. We use the shorthand
\beq
ab|cd|ef|gh \equiv \mmu{\twogon{a}{b}}\otimes\mmu{\twogon{c}{d}}\otimes\mmu{\twogon{e}{f}}\otimes\mmu{\twogon{g}{h}}\,,
\eeq
as well as the notation for shuffles
\beq\bsp
A|B|C\uplus D &\,\equiv A|B|C|D + A|B|D|C\,,\\
A|B\uplus(C|D) &\, \equiv A|B|C|D + A|C|B|D + A|C|D|B\,,\\
A|B\uplus C\uplus D &\, \equiv \sum_{\sigma\in S_3} A|\sigma(B)|\sigma(C)|\sigma(D)\,,
\esp\eeq
where in the last equation the sum runs over all permutations of the set $\{B, C, D\}$.
The 55 maximal dissections of the pentagon in eq.~\eqref{eq:Gabcdz}, together with the term in the symbol $\cS(G(a,b,c,d;x))$ they correspond to, are
\begin{center}\tabcolsep 2.5pt
\begin{tabular}{ccccc}
\fivegonarrowtCDa{d}{c}{b}{a}{x} &
\fivegonarrowtCEa{d}{c}{b}{a}{x}&
\fivegonarrowtCAa{d}{c}{b}{a}{x}&
\fivegonarrowtCBa{d}{c}{b}{a}{x}&
\fivegonarrowtCCa{d}{c}{b}{a}{x} \\
\\
$+dx|cx|bx|ax$&
$+ax|da|ca|ba$&
$-bx|ab\uplus (db|cb)$&
$+cx|dc\uplus(ac|bc)$&
$-dx|ad|bd|cd$\\
\\
\fivegonarrowtCDb{d}{c}{b}{a}{x} &
\fivegonarrowtCEb{d}{c}{b}{a}{x}&
\fivegonarrowtCAb{d}{c}{b}{a}{x}&
\fivegonarrowtCBc{d}{c}{b}{a}{x}&
\fivegonarrowtCCc{d}{c}{b}{a}{x}\\
\\
$+cx|dc\uplus(bx|ax)$&
$-dx|ad|ca|ba$&
$+bx|ax\uplus(db|cb)$&
$-cx|dc\uplus(ac|ba)$&
$+dx|ad|bd|cb$\\
\\
\fivegonarrowtCCb{d}{c}{b}{a}{x}  &
\fivegonarrowtCDc{d}{c}{b}{a}{x}&
\fivegonarrowtCEc{d}{c}{b}{a}{x}&
\fivegonarrowtCAc{d}{c}{b}{a}{x}&
\fivegonarrowtCBb{d}{c}{b}{a}{x}\\
\\
$+ax|da|bd|cd$&
$-dx|cx|bx|ab$&
$-ax|da|ca|bc$&
$+bx|ab\uplus(db|cd)$&
$+dx|cx|ac|bc$\\
\\
\fivegonarrowtCDd{d}{c}{b}{a}{x} &
\fivegonarrowtCEd{d}{c}{b}{a}{x}&
\fivegonarrowtCAd{d}{c}{b}{a}{x}&
\fivegonarrowtCBd{d}{c}{b}{a}{x}&
\fivegonarrowtCCd{d}{c}{b}{a}{x}\\
\\
$-cx|dc\uplus(bx|ab)$&
$+dx|ad|ca|bc$&
$-bx|ax\uplus(db|cd)$&
$-dx|cx|ac|ba$&
$-ax|da|bd|cb$\\
%
\end{tabular}
\begin{tabular}{ccccc}
\fivegonarrowtBAb{d}{c}{b}{a}{x}&
\fivegonarrowtBDb{d}{c}{b}{a}{x}&
\fivegonarrowtBBa{d}{c}{b}{a}{x}&
\fivegonarrowtBCb{d}{c}{b}{a}{x}&
\fivegonarrowtBEd{d}{c}{b}{a}{x}\\
\\
$+bx|ax\uplus(cb|dc)$&
$+dx|cd|ac|ba$&
$-dx|ax\uplus(bd|cb)$&
$+cx|dc\uplus(ax|ba)$&
$-dx|ad|ba|cb$\\
\\
\fivegonarrowtBAc{d}{c}{b}{a}{x}&
\fivegonarrowtBDf{d}{c}{b}{a}{x}&
\fivegonarrowtBBd{d}{c}{b}{a}{x}&
\fivegonarrowtBCd{d}{c}{b}{a}{x}&
\fivegonarrowtBEa{d}{c}{b}{a}{x}\\
\\
$+ax|ba|db|cb$&
$+cx|dc\uplus(bc|ab)$&
$-dx|ad|cd|bc$&
$-dx|cd\uplus(bx|ax)$&
$+dx|ax|ca|ba$\\
\\
\fivegonarrowtBAd{d}{c}{b}{a}{x}&
\fivegonarrowtBDe{d}{c}{b}{a}{x}&
\fivegonarrowtBBc{d}{c}{b}{a}{x}&
\fivegonarrowtBCc{d}{c}{b}{a}{x}&
\fivegonarrowtBEb{d}{c}{b}{a}{x}\\
\\
$-ax|ba|db|cd$&
$+dx|cx|bc|ab$&
$+ax|da|cd|bc$&
$+dx|cd\uplus(bx|ab)$&
$-dx|ax|ca|bc$\\
\\
\fivegonarrowtBAa{d}{c}{b}{a}{x}&
\fivegonarrowtBDa{d}{c}{b}{a}{x}&
\fivegonarrowtBBb{d}{c}{b}{a}{x}&
\fivegonarrowtBCa{d}{c}{b}{a}{x}&
\fivegonarrowtBEc{d}{c}{b}{a}{x}\\
\\
$-bx|ab\uplus(cb|dc)$&
$-dx|cd|ac|bc$&
$+dx|ax\uplus(bd|cd)$&
$+dx|cx|ax|ba$&
$+ax|da|ba|cb$\\
\\
\fivegonarrowtBAe{d}{c}{b}{a}{x}&
\fivegonarrowtBDc{d}{c}{b}{a}{x}&
\fivegonarrowtBBf{d}{c}{b}{a}{x}&
\fivegonarrowtBCf{d}{c}{b}{a}{x}&
\fivegonarrowtBEf{d}{c}{b}{a}{x}\\
\\
$+ax|ca|dc\uplus ba$&
$+dx|bd| ab\uplus cb$&
$-cx|dc\uplus bc\uplus ax$&
$+dx|ad|cd\uplus ba$&
$+dx|bx|ax\uplus cb$\\
\\
 \fivegonarrowtBAf{d}{c}{b}{a}{x}&
 \fivegonarrowtBDd{d}{c}{b}{a}{x} &
 \fivegonarrowtBBe{d}{c}{b}{a}{x}&
 \fivegonarrowtBCe{d}{c}{b}{a}{x}&
\fivegonarrowtBEe{d}{c}{b}{a}{x} \\
\\
$-ax|ca|dc\uplus bc$&
$-dx|bd|cd\uplus ab$&
$-dx|cx|ax\uplus bc$&
$-ax|da|cd\uplus ba$&
$-dx|bx|cb\uplus ab$\\
\\
\fivegonarrowtAa{d}{c}{b}{a}{x} &
\fivegonarrowtAb{d}{c}{b}{a}{x} &
\fivegonarrowtAc{d}{c}{b}{a}{x} &
\fivegonarrowtAd{d}{c}{b}{a}{x} &
\fivegonarrowtAe{d}{c}{b}{a}{x} \\
\\
$+ax|ba|cb|dc$&
$-dx|cd|bc|ab$&
$+dx|ax\uplus(cd|bc)$&
$-dx|cd\uplus(ax|ba)$&
$+dx|ax|ba|cb$\\
\end{tabular}
\end{center}

\section{Proof of Proposition~\ref{proposition:CMZV}}
\label{app:CMZV_proof}

In this section we present the proof of Proposition~\ref{proposition:CMZV}. The proof uses the combinatorics of the decorated polygons introduced in section~\ref{sec:tensors}. In order to be able to map the symbol of colored multiple zeta values (CMZV's) to polygons, we first have to relate the CMZV's to multiple polylogarithms. From the series representations~\eqref{eq:Lim_def} and~\eqref{eq:CMZV_def}, it is easy to see that one has the relation (provided that the CMZV's converge)
\beq\label{eq:CMZV_to_MPL}
\zeta(m_1,\ldots,m_k;s_1,\ldots,s_k) = (-1)^w\,G_{m_1,\ldots,m_k}(\hat s_1,\ldots,\hat s_k)\,,
\eeq
where we defined $w= m_1+\ldots+m_k$ and
\beq
\hat s_j = \prod_{i=1}^{j}s_i\,.
\eeq
Hence, using the correspondence between multiple polylogarithms and decorated polygons, we can associate to the CMZV $\zeta(m_1,\ldots,m_k;s_1,\ldots,s_k)$ the polygon\footnote{There is of course the factor $(-1)^w$ of eq.~\eqref{eq:CMZV_to_MPL} to be kept in mind.} 
\beq
P(\underbrace{0,\ldots,0}_{m_k-1\textrm{ times}},\hat s_k,\ldots \underbrace{0,\ldots,0}_{m_1-1\textrm{ times}},\hat s_1,1)\,.
\eeq

We start by introducing concepts needed to prove, and then prove Proposition~\ref{proposition:CMZV}. We will in fact proof results at the level of the polygons that are slightly more general than the results given in Proposition~\ref{proposition:CMZV} but not imposing the restriction $(m_1,s_1) \neq (1,1)$ (which corresponds to divergent CMZV's). The first proposition we give is a generalization of statement 2 in the proposition. Statement 1 in Proposition~\ref{proposition:CMZV} is equivalent to Proposition~\ref{secondprop}. 

\begin{Proposition}
\label{mainprop}
The symbol of $P(\varepsilon_{1},...,\varepsilon_{n},1)$ for some $\varepsilon_{i} = \pm 1$ is equal to $\lambda_{a,n}(2^{\ot n})$ for
\beq
\lambda_{a,n} = (-1)^{a}\binom{n-1}{a} {\rm~~and~~}a = n- \max\{i \: | \: \varepsilon_{i} = -1\}\,.
\eeq
\end{Proposition}

We start by noting that
from the H\"{o}lder convolution~\eqref{eq:Hoelder_inf} with $p = \infty$ it follows that $P(x_{1},...,x_{m},1)$ has the same symbol as the polygon $P(1-x_{m},...,1-x_{1},1)$.
So, without loss of generality, we consider the polygon
$P(\underbrace{0,...,0}_{t_{0}},2,\underbrace{0,...,0}_{t_{1}},2,0, ...,0 ,2,\underbrace{0,...,0}_{t_{m}},1)$
and will find its symbol, remembering to take into account the factor of $(-1)^{n}$. The move from sides labelled $1$ and $-1$ to sides labelled $0$ and $2$ increases the number of dissections that have coefficient $0$. The combinatorics of the dissections of polygons of this type is best captured by a certain type of planar trees, the so-called \emph{Hook-arrow trees}, which is a change of view on the maximal dissections, and hence the symbol, of a polygon.

\paragraph{Hook-arrow trees.}
Every full dissection of a polygon uniquely defines a certain spanning tree $\tau$ on the vertices which are the midpoints of the polygon sides, and vice-versa. These vertices, $v_{1},...,v_{n}$, inherit the label of the side they sit on and they form the vertices of $\tau$. We induce the edges of $\tau$ as all possible lines, between the $v_{i}$, that do not cross arrows from the dissection. Here is an example of a full dissection of a 4-gon with the spanning tree induced:

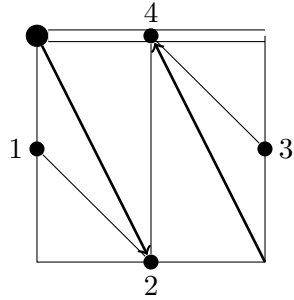
\begin{center}
\begin{tikzpicture}[inner sep=2pt,scale = 1.5, dot/.style={fill=black,circle,minimum size=1pt}]
\path(2,1) node[circle,draw,fill=black,minimum size = 8pt] (Q1){};
\path(2,-1) coordinate (Q2);
\path(4,-1) coordinate (Q3);
\path(4,1) coordinate(Q4);
\begin{scope}
\draw [double distance = 4pt] (Q4)--(Q1)node[pos=0.33](thirdq41){};
\end{scope}
\draw(Q1)-- (Q2)--(Q3)node[pos=0.33](thirdq23){}--(Q4);
\draw node[dot,label = left:$1$] (ql1) at (2,0) {};
\draw node[dot,label=below:$2$] (ql2) at (3,-1) {};
\draw node[dot,label=right:$3$] (ql3) at (4,0) {};
\draw node[dot, label = above:$4$] (ql4) at (3,1) {};
\draw[->,line width = 1pt] (Q3) to (ql4);
\draw[->,line width = 1pt] (Q1) to (ql2);
\draw[-] (ql2) to (ql4);
\draw[-] (ql1) to (ql2);
\draw[-] (ql4) to (ql3);
\end{tikzpicture} 
\end{center}

We also induce a root on $\tau$ as the vertex lying on the final side of the polygon. The edges are then oriented towards the root. For the above example of a dissected 4-gon the rooted spanning tree is:

 
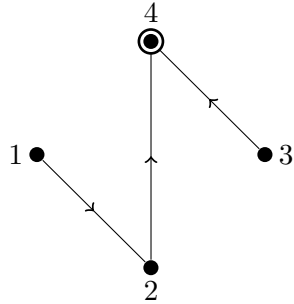
\begin{center}
\begin{tikzpicture}[inner sep=2pt,scale = 1.5, dot/.style={fill=black,circle,minimum size=1pt}]

\draw node[dot,label = left:$1$] (rl1) at (5,0) {};
\draw node[dot,label=below:$2$] (rl2) at (6,-1) {};
\draw node[dot,label=right:$3$] (rl3) at (7,0) {};
\draw node[fill=black,circle, minimum size = 10pt, label = above:$4$] (rl4) at (6,1) {};
\draw node[fill=white,circle, minimum size = 8pt] (rl4) at (6,1) {};
\draw node[dot] (rl4a) at (6,1) {};

\midarrow{rl2}{rl4};
\midarrow{rl1}{rl2};
\midarrow{rl3}{rl4};

\end{tikzpicture} 
\end{center}

The edges of the tree will not cross by construction; we define interlacing to reflect this for use in the definition of a hook-arrow tree.

\begin{Definition}
A tree with  a linear order on its vertices $w_j$ is said to be {\bf  interlaced } if there exists a choice of four vertices $ w_{1}<...<w_{4}$ such that both edges $\{w_{1},w_{3}\}$ and $\{w_{2},w_{4}\}$ are contained in the tree.
\end{Definition}

We now give a formal definition of a hook-arrow tree and illustrate the definition with an example.
\begin{Definition}
\label{hatdef}
A {\bf hook-arrow tree} is a rooted spanning tree on a set of vertices in a linear order, $v_{1}<...<v_{n}$, which is not interlaced and has root $v_{n}$.  
\end{Definition}

We can think of a hook-arrow tree as being a tree embedded in the plane on the vertices arranged in a circle.

\newpage\noindent 
{\bf Example.}
For the polygon P(2,0,2,0,0,2,0,1) (attached to the multiple polylogarithm $G(0,2,0,0,2,0,2;1)$) we have the following possible maximal dissection
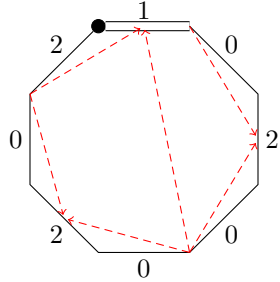
\small
\begin{center}
\octexampletwo{0.5}
\end{center}
\normalsize 
We note that, for clarity, in this example we give the arrows in the dissection of a polygon dashed shafts. 

We now construct the hook-arrow tree for this dissection, following the method outlined in Definition \ref{hatdef} (see also fig.~\ref{fig:hook}):

\begin{enumerate}
\item Add vertices to the middle of each side of the polygon;
\item Join all vertices that can be connected without crossing the shaft of an arrow;
\item Remove the polygon and arrows and direct tree towards distinguished vertex representing final side of polygon.
\end{enumerate}

\begin{center}
\begin{figure}[!h]
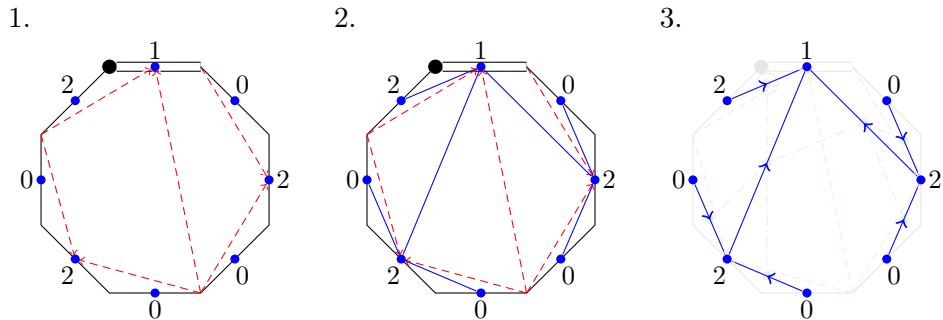

\begin{center}
\begin{tabular}{lll}
1.  & 2. & 3. \\
\small
\octexamplethree{0.5}{100}{0}
\normalsize 
&
\small
\octexamplethree{0.5}{100}{100}
\normalsize 
&
\small
 \octexample{10}{100}{10}{10}{0.5}
\normalsize 
\end{tabular}
\caption{\label{fig:hook}Construction of a hook arrow tree corresponding to a maximal dissection. Details are given in the text.}
\end{center}
\end{figure}
\end{center}

We now reintroduce the dual tree view of a dissection from section 3.1 of ref.~\cite{GGL:2009} as this is also beneficial in finding the symbol attached to a polygon. As with the dissection of a polygon using arrows, the dual tree can easily be seen in the hook-arrow tree view. For clarity we give the dual tree a dash-dotted line.

\begin{center}
\small
\begin{tabular}{ccc}
 \octexample{10}{100}{100}{10}{0.5} &\qquad & \octexamplefour{100}{10}{100}{100}{0.5 } \\[10pt]
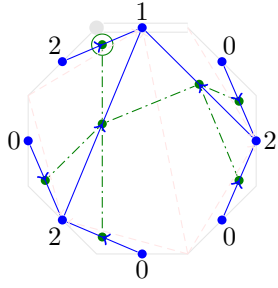
Hook-arrow tree and dual tree && Polygon dissection and dual tree \\[10pt]
\end{tabular}
\normalsize 
\end{center}


We now elaborate on the interpretation of a hook-arrow tree.

The association between a $2$-gon in a polygon dissection, an edge of a hook-arrow tree and a term of a tensor product in the symbol is

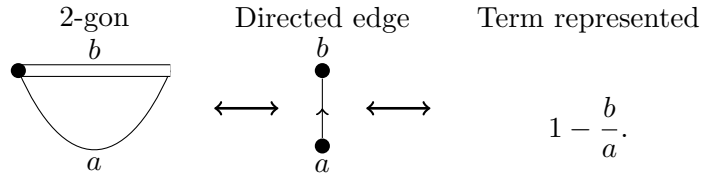
\begin{center}
\begin{tikzpicture}[inner sep=2pt,scale = 1, dot/.style={fill=black,circle,minimum size=1pt}]
\draw node[label=below:$2$-gon] at (1,1){};
\draw node[label=below:Directed edge] at (4,1){};
\draw node[label=below:Term represented] at (7.5,1){};

\draw node[label=above:$b$] at (1,0){};
\draw node[label=below:$a$] at (1,-1){};
\path(0,0) coordinate (P1);
\path(2,0) coordinate (P2);
\draw (0,0) .. controls (0.6,-1.4) and (1.4,-1.4)  .. (2,0); 
\begin{scope}
\draw [double distance = 4pt] (P1)--(P2){};
\end{scope}
\draw node[dot] (A1) at (0,0){};

\draw node (A11) at (2.5,-0.5){};
\draw node (A12) at (3.5,-0.5){};
\draw[<->,line width=1pt] (A11) to (A12);

\draw node[dot,label=above:$b$](a) at (4,0){};
\draw node[dot,label=below:$a$](b) at (4,-1){};
\midarrow{b}{a};

\draw node (A21) at (4.5,-0.5){};
\draw node (A22) at (5.5,-0.5){};
\draw[<->,line width=1pt] (A21) to (A22);

\draw node[label=below:$1- \displaystyle\frac{b}{a}. $] at (7.5,-0.2){};
\draw node at (7.3,-0.5){};
\end{tikzpicture}
\end{center} 

In the dissection of a polygon, we set the coefficient of certain terms in the symbol to $0$ if the sides of $2$-gons have certain combinations of labels, e.g., $\mmu{\twogon{{0}}{{a}}} =0$. We do the same in a hook-arrow tree when we have the corresponding edges
$$
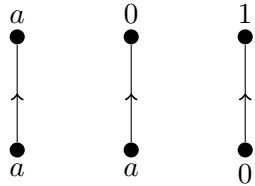
\begin{tikzpicture}[inner sep=2pt,scale = 1.5, dot/.style={fill=black,circle,minimum size=1pt}]
\draw node[dot,label=above:$1$](a) at (1,0){};
\draw node[dot,label=below:$0$](b) at (1,-1){};
\midarrow{b}{a};

\draw node[dot,label=above:$0$](a) at (0,0){};
\draw node[dot,label=below:$a$](b) at (0,-1){};
\midarrow{b}{a};

\draw node[dot,label=above:$a$](a) at (-1,0){};
\draw node[dot,label=below:$a$](b) at (-1,-1){};
\midarrow{b}{a};
\end{tikzpicture}$$
for any $a$. 

To obtain the term in the symbol, the edges of the hook-arrow tree are chosen in the same order as the corresponding $2-$gons in the dissection of a polygon. This can also be seen by viewing the dual tree (this correspondence is shown in the above example). Note that an explicit algorithm for this purely from the view of the hook-arrow tree can be written.

The sign of a maximal dissection is determined, in a similar way to the polygons by first determining the number, $\alpha$, of `backward' edges in the hook-arrow tree. These correspond to edges that, respecting their direction, go from one vertex in the order of the vertices defined by the polygon, to a previous vertex. The sign is then $(-1)^{\alpha}$.

\begin{Remark}
As previously mentioned, it is important to note that for every full dissection of a polygon there is a unique corresponding hook-arrow tree and that the method of extracting the symbol from hook-arrow trees is simply a different view of the procedure for polygons. We include an overview of the construction of an hook-arrow tree because constructing all possible dissections for the functions required in Proposition \ref{mainprop} is a lot easier to view from this perspective. The actual proof then takes place in the dual tree view.
\end{Remark}
Finally, before the proof of Proposition \ref{mainprop} we require the following proposition for which we give a sketch proof involving generating functions.

\begin{Proposition}
\label{genfuncprop}
If $c,n \in \mathbb{Z}^{\ge 0}$ then
$$\displaystyle\sum_{i=0}^{n} (-1)^{i}\binom{n-i+c}{n-i}\binom{n+c+1}{i} = (-1)^{n}$$
\end{Proposition}

\begin{proof} (Sketch)

Let $r=n-i$ and view the right hand side of the identity as coefficients of the generating function
$$\Phi(x) = \displaystyle\sum_{n=0}^{\infty} x^{n} \sum_{r=0}^{n} (-1)^{n-r}\binom{r+c}{r}\binom{n+c+1}{n-r}.$$
Firstly since for $r>n$ we have $\binom{n+c +1}{n-r} = 0$, we can change the summation of $r$ to run over all positive integers. Then by re-ordering the summation signs for small $x$, and applying some basic properties of binomial coefficients we can arrive at
$$\Phi(x) = \displaystyle \sum_{r=0}^{\infty} \frac{(-1)^{r-c-1}}{x^{c+1}} \binom{r+c}{r} \sum_{n=0}^{\infty} \binom{n+c+1}{r+c+1}(-x)^{n+c+1}.$$
Then by using 
$$ \displaystyle\sum_{m}\binom{m}{k} x^{m} = \frac{x^{k}}{(1-x)^{k+1}} \quad \textrm{and then} \quad \displaystyle\sum_{m} \binom{m+k}{k}x^{m} = \frac{1}{(1-x)^{k+1}}$$
we can show that 
$$\Phi(x) = \frac{1}{1+x} = \displaystyle\sum_{n=0}^{\infty}(-x)^{n}$$
which proves the result.
\end{proof}

\begin{proof} (of Proposition \ref{mainprop})
After applying the H\"{o}lder convolution with $p=\infty$, and without loss of generality, we attempt to find all hook-arrow trees relating to the polygon
$$P(\underbrace{0,...,0}_{t_{0}},2,\underbrace{0,...,0}_{t_{1}},2,0,\quad ...\quad,0 ,2,\underbrace{0,...,0}_{t_{m}},1)$$
which do not represent terms with coefficient $0$ in the symbol. After some consideration we see that these hook-arrow trees must take the following form.

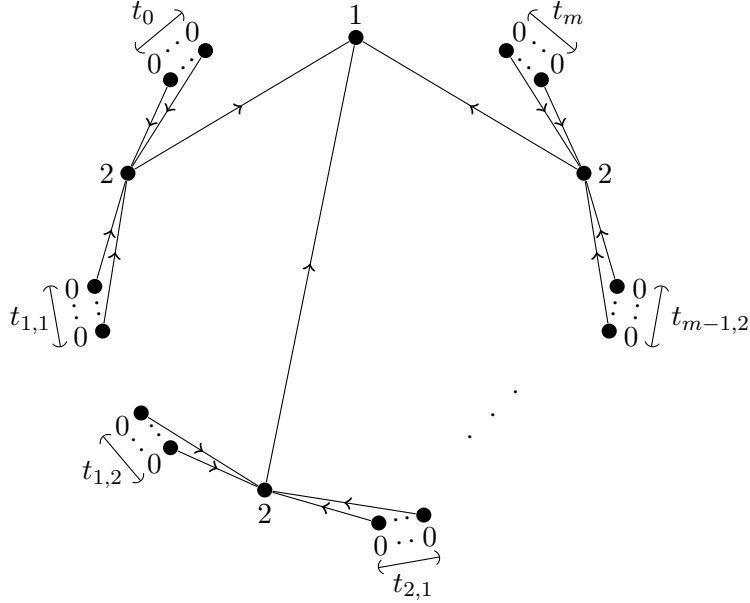
\begin{center}
\begin{tikzpicture}[inner sep=2pt,scale = 3, dot/.style={fill=black,circle,minimum size=1pt}]
\draw node[dot,label = left:$2$] (Px) at (-1,0.4) {};
\draw node[dot,label=below:$2$] (Py) at (-0.4,-1) {};
\draw node[dot,label=right:$2$] (Pz) at (1,0.4) {};
\draw node[dot, label = above:$1$] (P1) at (0,1) {};
\draw node[ label = above:$\cdot$] (dots) at (0.6,-0.75) {};
\draw node[ label = above:$\cdot$] (dots) at (0.7,-0.65) {};
\draw node[ label = above:$\cdot$] (dots) at (0.5,-0.85) {};
\midarrow{Px}{P1};
\midarrow{Py}{P1};
\midarrow{Pz}{P1};
\zerogroup{-50+180}{5}{Px}{$t_{0}$};
\zerogroup{10+180}{5}{Px}{$t_{1,1}$};
\zerogroup{240-20}{5}{Py}{$t_{1,2}$};
\zerogroup{-60-20}{5}{Py}{$t_{2,1}$};
\zerogroup{-10}{5}{Pz}{$\quad \quad t_{m-1,2}$};
\zerogroup{50}{5}{Pz}{$t_{m}$};
\end{tikzpicture}
\end{center}

Each $t_{i,1}$ and $t_{i,2}$, for $i=1,...,m-1$ are chosen integers $0\le t_{i,1},t_{i,2} \le t_i$ such that $t_{i,1}+t_{i,2} = t_{i}$. The choice of the $t_{i,j}$ arises from the fact that we can choose where to partition each group of $t_{i}$ vertices labelled $0$, for $i=1,...,m-1$, and attach them to the vertices labelled $2$, remembering that the vertices must not cross. In the case of the function $P(2,0,2,0,0,2,0,1)$ from the example above, where $m=3, t_{0}=0, t_{1}=1,t_{2}=2$ and $t_{3}=1$, we would have $6$ possible valid dissections, arising from two choices of the $t_{i,j}$ in $t_{1,1}+t_{1,2}=1$ and three choices from $t_{2,1} +t_{2,2} = 2$. We note that the example above explored the particular dissection where $t_{1,1}=0,t_{1,2}=1,t_{2,1}=1$ and $t_{2,2}=1$.

We will now show how it is possible to simplify this tree by effectively removing the edges joining vertices labelled $2$ and $1$ and replacing them with edges connecting vertices labelled $0$ and $2$. For this we return to the dual tree notation described above, and in section~\ref{nutshell}. The dual tree of the above hook-arrow tree above is

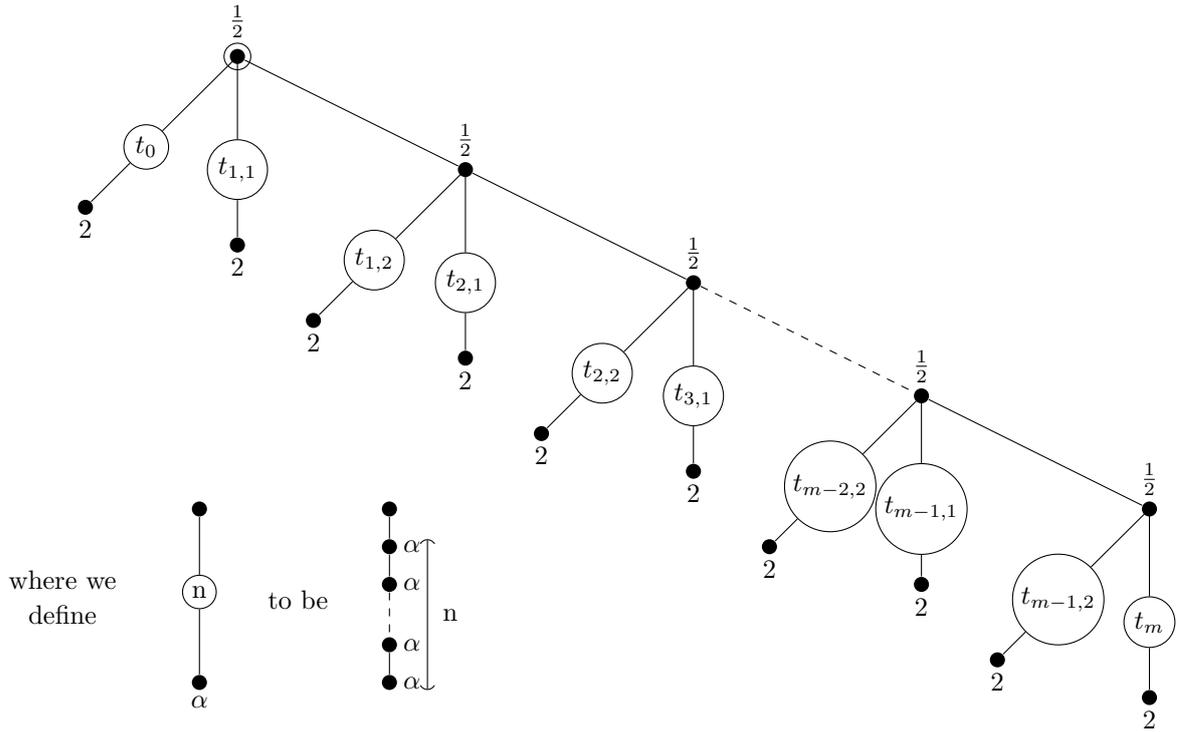
\small
\begin{center}
\begin{tikzpicture}[inner sep=2pt,scale = 1, dot/.style={fill=black,circle,minimum size=1pt}]
\draw node[dot,] (H1) at (0,0) {};
\draw node[draw,circle, minimum size = 10pt,label = above:$\frac{1}{2}$] (C1) at (0,0) {};
\dualpart{0}{0}{$t_{0}$}{$t_{1,1}$}{1};
\draw node[dot,label = above:$\frac{1}{2}$] (H2) at (3,-1.5) {};
\draw [-] (H1) to (H2) {};
\dualpart{3}{-1.5}{$t_{1,2}$}{$t_{2,1}$}{1};
\draw node[dot,label = above:$\frac{1}{2}$] (H3) at (6,-3) {};
\draw [-] (H2) to (H3) {};
\dualpart{6}{-3}{$t_{2,2}$}{$t_{3,1}$}{1};
\draw node[dot,label = above:$\frac{1}{2}$] (H4) at (9,-4.5) {};
\dualpart{9}{-4.5}{$t_{m-2,2}$}{$t_{m-1,1}$}{1};
\draw [dashed] (H3) to (H4) {};
\draw node[dot,label = above:$\frac{1}{2}$] (H5) at (12,-6) {};
\draw [-] (H4) to (H5) {};
\dualpart{12}{-6}{$t_{m-1,2}$}{$t_{m}$}{1};

\draw node (W0) at (-2.3,-6 -1.2) {\begin{tabular}{c}where we \\ define \end{tabular}};
\draw node (W1) at (0.8,-6 -1.2) {to be};

\draw node[dot] (G0) at (2,-6) {};
\draw node[dot,label = right:$\alpha$] (G1) at (2,-6-0.5) {};
\draw node[dot,label = right:$\alpha$] (G2) at (2,-6-1) {};
\draw node[dot,label = right:$\alpha$] (G3) at (2,-6-1.8) {};
\draw node[dot,label = right:$\alpha$] (G4) at (2,-6-2.3) {};
\draw[(-)] (2.5,-6-0.4) to (2.5,-6-2.4) {};
\draw node (H5) at (2.8,-6-1.4) {n};
\draw [dashed] (G2) to (G3) {};
\draw [-] (G0) to (G1) {};
\draw [-] (G1) to (G2) {};
\draw [-] (G3) to (G4) {};

\draw node[dot] (K0) at (-0.5,-6) {};
\draw node[circle,draw] (K1) at (-0.5,-6-1.1) {n};
\draw node[dot,label = below:$\alpha$] (K2) at (-0.5,-6-2.3) {};
\draw [-] (K0) to (K1) {};
\draw [-] (K1) to (K2) {};

\end{tikzpicture}
\end{center}

\normalsize

We claim that 

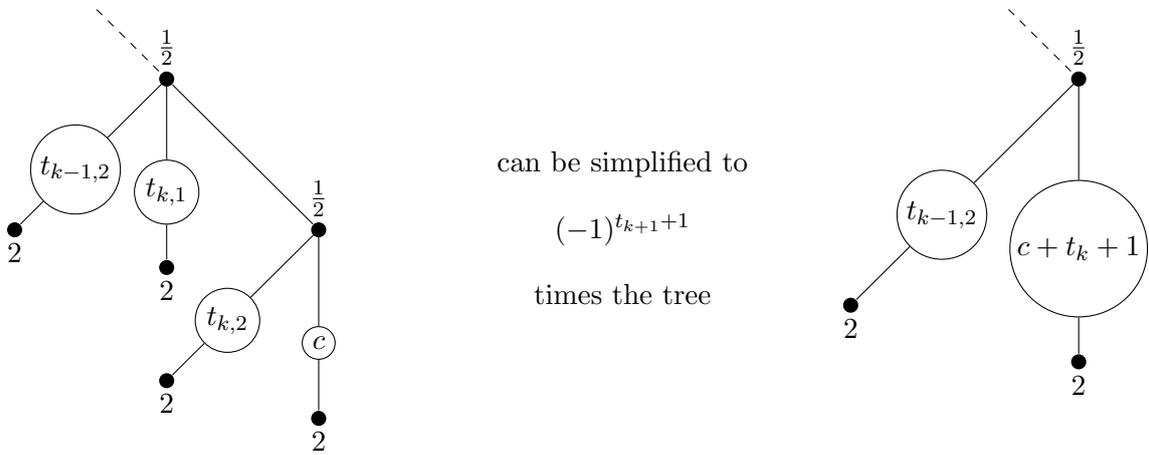
\begin{center}
\begin{tikzpicture}[inner sep=2pt,scale = 1, dot/.style={fill=black,circle,minimum size=1pt}]
\draw node (H0) at (-1,1) {};
\draw node[dot,label = above:$\frac{1}{2}$] (H1) at (0,0) {};
\draw [dashed] (H0) to (H1) {};
\dualpart{0}{0}{$t_{k-1,2}$}{$t_{k,1}$}{1};
\draw node[dot,label = above:$\frac{1}{2}$] (H2) at (2,-2) {};
\draw [-] (H1) to (H2) {};
\dualpart{2}{-2}{$t_{k,2}$}{$c$}{1};

\draw node (W0) at (6,-2) {\begin{tabular}{c}can be simplified to \\[10pt] $(-1)^{t_{k+1}+1}$\\[10pt] times the tree\end{tabular}};

\draw node (I0) at (11,1) {};
\draw node[dot,label = above:$\frac{1}{2}$] (I1) at (12,0) {};
\draw [dashed] (I0) to (I1) {};
\dualpart{12}{0}{$t_{k-1,2}$}{{$c+t_{k}+1$}}{1.5};
\end{tikzpicture}
\end{center}

We will now write the tensor product of the symbol of the left hand dual tree part in the above claim. 
Let us recall out convention~\eqref{eq:shuffle_binding} that shuffles takes precedence over a tensor
$$a^{\ot b} = \underbrace{a \ot ... \ot a}_{b} \quad \textrm{and that} \quad a^{\ot b} \sha c= (\underbrace{a \ot ... \ot a}_{b}) \sha c.$$
The left hand dual tree part in the claim will have the symbol
\begin{align*}
& \displaystyle\sum_{t_{k,1}=0}^{t_{k}} (-1)^{t_{k,1}}\left( \frac{1}{2} \ot 2^{\ot t_{k-1,2}} \sha 2^{\ot t_{k,1}} \sha \left( \frac{1}{2} \ot 2^{\ot t_{k,2}} \sha 2^{\ot c} \right)\right)\\
=&- \displaystyle\sum_{t_{m-1,1} = 0}^{t_{m-1}} (-1)^{t_{m-1,1}} \binom{t_{m-1} - t_{m-1,1} +t_{m}}{t_{m-1}-t_{m-1,1}} \binom{t_{m-1}+t_{m} + 1}{t_{m-1,1}} \left( \frac{1}{2} \ot 2^{\ot t_{k-1,2}} \sha 2^{\ot (t_{m-1} +t_{m} +1)} \right)\\
=& (-1)^{t_{k+1}+1}\left(\frac{1}{2} \ot 2^{\ot t_{k-1,2}} \sha 2^{\ot (c+t_{k}+1)}\right)
\end{align*}
which is exactly the symbol for the tree on the right side. Note that we used Proposition~\ref{genfuncprop} in the last line of the calculation.

By repeated application of this simplification, starting with $c=t_{m}$ and $k=m-1$, we will arrive at a much simplified tree. By noting that $n-a-1 = \displaystyle\sum_{i=1}^{m} (t_{i}+1)$ and recalling that $t_{0} = a$ we see that this tree is
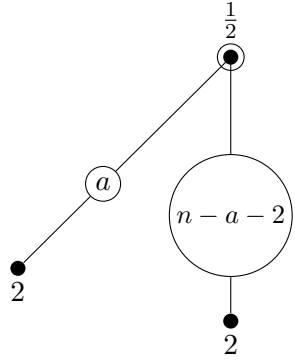
\begin{center}
\begin{tikzpicture}[inner sep=2pt,scale = 1.4, dot/.style={fill=black,circle,minimum size=1pt}]
\draw node[dot] (I1) at (0,0) {};
\draw node[draw,circle, minimum size = 10pt,label = above:$\frac{1}{2}$] (C1) at (0,0) {};

\dualpart{0}{0}{$a$}{\small$n-a-2$}{1};
\end{tikzpicture}
\end{center}
times a factor of $(-1)^{(n-a-1)}$. This represents the symbol
\begin{align*}
¥& (-1)^{n-a-1}\left(\frac{1}{2} \ot 2^{\ot a} \sha 2^{\ot\left( n-a-2 \right)}\right)\\
=& (-1)^{n-a} \binom{n-1}{a} 2^{\ot n}
\end{align*}
Finally, by including the factor of $(-1)^{n}$ from the application of H\"{o}lder involution we find
$$ \lambda_{a,n} = (-1)^{a}\binom{n-1}{a}.$$
\end{proof}

\begin{Proposition}
\label{secondprop}
The polygon $P(\underbrace{0,...,0}_{m_{1}-1},\varepsilon_1,\ldots,\underbrace{0,...,0}_{m_{k}-1},\varepsilon_k,1)$ for $\varepsilon_{i} = \pm 1$ and at least one of the $m_{i}\neq 1$, has a symbol equal to $0$.
\end{Proposition}
\begin{proof} (Sketch)
After applying H\"{o}lder involution~\eqref{eq:Hoelder_inf} with $p=\infty$ we try to find possible hook-arrow trees which do not correspond to terms with coefficient $0$ in the symbol. The vertices of the hook-arrow tree will be labelled corresponding to the sides of the polygon
$$P(\gamma_{1,1},...,\gamma_{t_{0},1},2,\gamma_{1,2},...,\gamma_{t_{1},2},2, ...,2,\gamma_{1,m},...,\gamma_{t_{m},m},1).$$
where all the $\gamma_{i,j}$ are equal to either $0$ or $1$. As in the proof of Proposition~\ref{mainprop}, the vertices labelled $2$ must connect directly to the final $1$ and the vertices labelled $0$ must connect to a $2$. However, there is no way to connect the $1$'s to any other vertex without setting the coefficient of the term to $0$. There is therefore no possible non-trivial dissection, and hence the symbol is zero.
\end{proof}

\section{Some considerations on the implementation of the algorithm}
The algorithmic approach presented in section~\ref{sec:higher_weights} relies on the construction of the sets $\cR^{(k)}$ defined in eq.~\eqref{eq:RTk}. The algorithm then consists in selecting the elements of $\cR$ which satisfy certain factorization properties. Even though this is a mathematically well-defined prescription, implementing this algorithm into a computer program can be hampered by several issues,
\begin{enumerate}
\item The set $\cR$ is infinite, and so we cannot just proceed by `trial and error' to select the elements that have the right factorization properties.
\item Factorization of polynomials is rather slow (at least on most computer algebra systems), leading to serious speed issues.
\end{enumerate}
The first issue can be dealt with by decomposing $\cR$ as
\beq\label{eq:RT_tower}
\cR = \bigcup_{n=0}^\infty \cR_{n}\,,
\eeq
where $\cR_{n}$ is defined as the subset of $\cR$ consisting of those elements $\pm\prod_{i=1}^k\overline{\pi}_i^{n_i}$ such that $|n_1|+\ldots+|n_k|=n$. In order to construct the set $\cR^{(1)}$, which is the basis out of which $\cR^{(k)}$ for arbitrary $k$ is constructed, one can then limit oneself to truncating the tower of sets in eq.~\eqref{eq:RT_tower} to some finite value $N<\infty$. Indeed, in practical applications one does not expect rational functions for large values of the sum $n$ of the exponents\footnote{Empirically, we observe that the set $\cR^{(1)}$ seems to be finite in general.}.

The second item seems to be harder to solve, since it is related to the capabilities of the chosen computer algebra system. It is however possible to circumvent this problem by deriving from the factorization constraints a necessary condition that must be fulfilled and that can be checked in a fast way by a computer. Since in practice most of the elements of $\cR$ will fail this constraint, one can filter out these elements and discard them in an easier way. In the following we discuss the example of $\cR^{(1)}$. The generalization to $\cR^{(k)}$ is immediate.

Let us consider a generic element $R$ in $\cR$. Without loss of generality, we can assume that we can write
\beq
R=s\,{\overline{\pi}_1^{n_1}\ldots \overline{\pi}_\ell^{n_\ell}\over \overline{\pi}_{\ell+1}^{n_{\ell+1}}\ldots \overline{\pi}_k^{n_k}}\,,
\eeq
with $s=\pm1$ and $n_i\in\mathbb{N}$, and $\overline{\pi}_i\in\overline{P}$. Checking if $R\in
\cR^{(1)}$ is equivalent to checking whether $1-R\in\cR$, i.e., whether $1-R$ can be written, up to a sign, as a ratio of elements from the set $\overline{P}$. Writing
\beq
1-R = {\overline{\pi}_{\ell+1}^{n_{\ell+1}}\ldots \overline{\pi}_k^{n_k} - s\, \overline{\pi}_1^{n_1}\ldots \overline{\pi}_\ell^{n_\ell}\over
\overline{\pi}_{\ell+1}^{n_{\ell+1}}\ldots \overline{\pi}_k^{n_k}}\,,
\eeq
it is easy to see that this condition can only be fulfilled if the polynomial in the numerator can be factored into a product of elements in $\overline{P}$. Let us call $\Pi$ this numerator. A necessary condition for $R\in\cR^{(1)}$ is thus that $\Pi$ can be divided by at least one element in $\overline{P}$. We can further simplify this condition by reducing it from a problem of division of polynomials to a problem of division of integers. Indeed, we can choose prime numbers $\{p_1,\ldots,p_m\}$ such that $\overline{\pi}_i(p_1,\ldots,p_m)\neq\overline{\pi}_j(p_1,\ldots,p_m)$, for $i\neq j$. The necessary condition for $R\in\cR^{(1)}$ then reduces to checking that the integer number $\Pi(p_1,\ldots,p_m)$ can be divided by at least one of the numbers $\overline{\pi}_i(p_1,\ldots,p_m)$, which is in general much quicker to test on a computer. Note however that this is a necessary, but not necessarily sufficient, condition for $R\in\cR^{(1)}$.

\section{Analytic continuation of the spanning set of functions for harmonic polylogarithms}
\label{app:anal_cont}
\subsection{Analytic representation inside the unit disc}
The analytic expressions for the spanning set $\{\cB_j^{(i)}(x)\}$ introduced in section~\ref{sec:HPLbasis} are valid for $x\in[0,1]$, but the functions might have different analytic representations in different regions of the complex plane. In particular, since the $\cB_i^{(j)}$ functions are real when $x$ is in the range $[0,1]$, Schwarz's reflection principle implies that the elements of the spanning set must satisfy
\beq
\cB_j^{(i)}(x^\ast) = \cB_j^{(i)}(x)^\ast\,,
\eeq
where $x^\ast$ denotes the complex conjugate of the complex number $x$. We checked numerically that the analytic expressions for the elements of the spanning set are valid everywhere inside the unit disc, except for $\cB_4^{(13)}(x)$, where the correct expression for $|x|<1$ is
  \beq\label{eq:Li41mz2}
  \cB_4^{(13)}(x) = \left\{\begin{array}{ll}
  \textrm{Li}_4(1-x^2)\,, & \textrm{if Re}(x) >0 \textrm{ or (\textrm{Re}}(x)=0 \textrm{ and Im}(x)\ge0)\,,\\ 
  \textrm{Li}_4(1-x^2)-{i\pi\over3}\,\sigma(x)\,\ln^3(1-x^2)\,,&\textrm{otherwise}\,,
  \end{array}\right.
  \eeq
  where $\sigma(x) = \textrm{sign}(\textrm{Im}(x))$. In order to understand this structure, 
  let us look at the simpler case of weight two. Then we get
\beq
\textrm{Li}_2(1-x^2) = -\log  (x^2)  \log  \left(1-x^2\right) -\text{Li}_2\left(x^2\right)+\frac{\pi ^2}{6}\,.
\eeq
The first term in this expression contains $\ln(x^2) = 2\ln x$, which is real for $x>0$, but develops an imaginary part for $x<0$, displaying the complicated  branch cut structure of $\textrm{Li}_2(1-x^2)$. A similar reasoning leads to eq.~(\ref{eq:Li41mz2}).

\subsection{Analytic representation outside the unit disc: inversion relations}
We now have analytic representations of the elements of the spanning set everywhere inside the unit disc, and so we can analytically continue them outside the unit disc via inversion relations, i.e., functional equations of the form
\beq
\cB_j^{(i)}(x) = \sum_{k,l}\,c_{ijkl}\,\cB_k^{(l)}\left({1\over x}\right)+\textrm{products of lower weight.}
\eeq
These functional equations can easily be obtained for the whole spanning set. Below we show the explicit inversion formulas for weight one and two. For the the complete list of inversion formulas for higher weights, we refer to appendix~\ref{app:inversion}. Letting $\sigma(x) = \textrm{sign}(\textrm{Im}(x))$, we get, for $|x|>1$, $x$ not real,
\begin{itemize}
\item for weight one:
\beq\bsp
\cB_{1}^{(1)}\left(x\right) &\,= -\cB_{1}^{(1)}\left(\frac{1}{x}\right)\,,\\
\cB_{1}^{(2)}\left(x\right) &\,= -\cB_{1}^{(1)}\left(\frac{1}{x}\right)+\cB_{1}^{(2)}\left(\frac{1}{x}\right)-i \pi  \sigma (x)\,,\\
\cB_{1}^{(3)}\left(x\right) &\,= \cB_{1}^{(3)}\left(\frac{1}{x}\right)-\cB_{1}^{(1)}\left(\frac{1}{x}\right)\,,\\
\esp\eeq
\item for weight two:
\beq\bsp
\cB_{2}^{(1)}\left(x\right) &\,= -i \pi  \sigma (x) \cB_{1}^{(1)}\left(\frac{1}{x}\right)-\frac{1}{2} \cB_{1}^{(1)}\left(\frac{1}{x}\right)^2-\cB_{2}^{(1)}\left(\frac{1}{x}\right)+\frac{\pi ^2}{3}\,,\\
\cB_{2}^{(2)}\left(x\right) &\,= -\frac{1}{2} \cB_{1}^{(1)}\left(\frac{1}{x}\right)^2-\cB_{2}^{(2)}\left(\frac{1}{x}\right)-\frac{\pi ^2}{6}\,,\\
\cB_{2}^{(3)}\left(x\right) &\,= -\log  2  \cB_{1}^{(1)}\left(\frac{1}{x}\right)-\frac{1}{2} \cB_{1}^{(1)}\left(\frac{1}{x}\right)^2+\cB_{1}^{(2)}\left(\frac{1}{x}\right) \cB_{1}^{(1)}\left(\frac{1}{x}\right)+\cB_{2}^{(1)}\left(\frac{1}{x}\right)\\
&\,-\cB_{2}^{(2)}\left(\frac{1}{x}\right)+\cB_{2}^{(3)}\left(\frac{1}{x}\right)-\frac{\pi ^2}{4}\,,\\
\esp\eeq
\end{itemize}
Note that the value of $\sigma$ is ambiguous for real values of $x$. This ambiguity can be resolved by  the `$i\varepsilon$' prescription commonly used in the physics literature. According to this prescription, we need to assign a small imaginary part to real value of $x$, i.e., if $x$ is real, we need to perform the replacement $x\to x+i\varepsilon$, and this replacement fixes at the same time the value of $\sigma$. Note however that some care is needed when applying the inversion formulas to real values of $x$ because Schwarz' reflection principle implies
\beq\bsp
\cB_j^{(i)}(x\pm i\varepsilon) &\,= \sum_{k,l}\,c_{ijkl}\,\cB_k^{(l)}\left({1\over x\pm i\varepsilon}\right)+\textrm{products of lower weight}\\
&\,= \sum_{k,l}\,c_{ijkl}\,\cB_k^{(l)}\left({1\over x}\mp i\varepsilon\right)+\textrm{products of lower weight}\\
&\,= \sum_{k,l}\,c_{ijkl}\,\cB_k^{(l)}\left({1\over x}\pm i\varepsilon\right)^\ast+\textrm{products of lower weight.}
\esp\eeq

If $x$ lies on the unit circle, $|x|=1$, there is an ambiguity whether to use the expression for $\cB_j^{(i)}(x)$ valid inside or outside the unit circle. We checked numerically that the two values agree in all cases, except for $\cB_5^{(14)}(x) = \textrm{Li}_4\left({4x\over (1+x)^2}\right)$. In this case, we find
\beq
\cB_4^{(15)}\left(e^{i\varphi}\right) = \textrm{Li}_4\left({1\over\cos^2{\varphi\over2}}\right)\,,
\eeq
i.e., the argument of $\textrm{Li}_4$ is real and greater than 1 for every $x$ on the unit circle (except for $x=-1$, where the result is divergent), and we have an ambiguity on the imaginary part of $\cB_4^{(15)}\left(e^{i\varphi}\right)$. This ambiguity can be lifted by requiring the function to be continuous in a neighborhood of the unit circle. To study this, let us consider a circle which is infinitesimally close to the unit circle, i.e., we choose $x=(1-\varepsilon)\,e^{i\varphi}$, for some infinitesimal $\varepsilon$. We then find
\beq
{4x\over(1+x)^2} = {1\over\cos^2{\varphi\over2}}\,\left(1+i\varepsilon\tan{\varphi\over2}+\ord(\eps^2)\right)\,,
\eeq
i.e., we see that for $|x|=1$, we have,
\beq
\cB_4^{(15)}\left(x\right) = \textrm{Li}_4\left({4x\over(1+x)^2}+i\sigma\varepsilon\right)\,,
\eeq
with $\sigma(x) = \textrm{sign}(\textrm{Im}(x))$.

\section{Inversion formulas for the spanning set}
In this appendix we present the inversion formulas for the spanning set for weight three and four, valid for $x\in\mathbb{C}^\times$.
\label{app:inversion}
\subsection{Weight three}
\beq\bsp
\cB_{3}^{(1)}\left(x\right) &\,= \frac{1}{2} i \pi  \sigma (x) \cB_{1}^{(1)}\left(\frac{1}{x}\right)^2+\frac{1}{6} \cB_{1}^{(1)}\left(\frac{1}{x}\right)^3-\frac{1}{3} \pi ^2 \cB_{1}^{(1)}\left(\frac{1}{x}\right)+\cB_{3}^{(1)}\left(\frac{1}{x}\right)\,,\\
\cB_{3}^{(2)}\left(x\right) &\,= \frac{1}{6} \cB_{1}^{(1)}\left(\frac{1}{x}\right)^3+\frac{1}{6} \pi ^2 \cB_{1}^{(1)}\left(\frac{1}{x}\right)+\cB_{3}^{(2)}\left(\frac{1}{x}\right)\,,
\esp\eeq

\beq\bsp
\cB_{3}^{(3)}\left(x\right) &\,= \frac{1}{6} \cB_{1}^{(1)}\left(\frac{1}{x}\right)^3-\frac{1}{2} \cB_{1}^{(2)}\left(\frac{1}{x}\right) \cB_{1}^{(1)}\left(\frac{1}{x}\right)^2+\frac{1}{6} \pi ^2 \cB_{1}^{(1)}\left(\frac{1}{x}\right)-\cB_{3}^{(1)}\left(\frac{1}{x}\right)\\
&\,-\cB_{3}^{(3)}\left(\frac{1}{x}\right)+\zeta_3\,,\\
\cB_{3}^{(4)}\left(x\right) &\,= \frac{1}{3} \cB_{1}^{(3)}\left(\frac{1}{x}\right)^3-\frac{1}{2} \cB_{1}^{(1)}\left(\frac{1}{x}\right) \cB_{1}^{(3)}\left(\frac{1}{x}\right)^2-\frac{1}{6} \pi ^2 \cB_{1}^{(3)}\left(\frac{1}{x}\right)-\cB_{3}^{(2)}\left(\frac{1}{x}\right)\\
&\,-\cB_{3}^{(4)}\left(\frac{1}{x}\right)+\zeta_3\,,\\
\cB_{3}^{(5)}\left(x\right) &\,= i \pi  \sigma (x) \log  2  \cB_{1}^{(1)}\left(\frac{1}{x}\right)-i \pi  \sigma (x) \log  2  \cB_{1}^{(3)}\left(\frac{1}{x}\right)+\frac{1}{2} \log ^2 2  \cB_{1}^{(1)}\left(\frac{1}{x}\right)\\
&\,-\frac{1}{2} \log ^2 2  \cB_{1}^{(3)}\left(\frac{1}{x}\right)+\frac{1}{2} \log  2  \cB_{1}^{(1)}\left(\frac{1}{x}\right)^2-\log  2  \cB_{1}^{(3)}\left(\frac{1}{x}\right) \cB_{1}^{(1)}\left(\frac{1}{x}\right)\\
&\,-\frac{1}{2} \log  2  \cB_{1}^{(2)}\left(\frac{1}{x}\right)^2+\log  2  \cB_{1}^{(2)}\left(\frac{1}{x}\right) \cB_{1}^{(3)}\left(\frac{1}{x}\right)+\frac{1}{2} i \pi  \sigma (x) \cB_{1}^{(1)}\left(\frac{1}{x}\right)^2\\
&\,-i \pi  \sigma (x) \cB_{1}^{(3)}\left(\frac{1}{x}\right) \cB_{1}^{(1)}\left(\frac{1}{x}\right)+\frac{1}{2} i \pi  \sigma (x) \cB_{1}^{(3)}\left(\frac{1}{x}\right)^2+\frac{1}{6} \cB_{1}^{(1)}\left(\frac{1}{x}\right)^3\\
&\,-\frac{1}{2} \cB_{1}^{(3)}\left(\frac{1}{x}\right) \cB_{1}^{(1)}\left(\frac{1}{x}\right)^2-\frac{1}{2} \cB_{1}^{(2)}\left(\frac{1}{x}\right)^2 \cB_{1}^{(1)}\left(\frac{1}{x}\right)\\
&\,+\cB_{1}^{(2)}\left(\frac{1}{x}\right) \cB_{1}^{(3)}\left(\frac{1}{x}\right) \cB_{1}^{(1)}\left(\frac{1}{x}\right)-\frac{1}{3} \pi ^2 \cB_{1}^{(1)}\left(\frac{1}{x}\right)+\frac{1}{6} \cB_{1}^{(2)}\left(\frac{1}{x}\right)^3\\
&\,+\frac{1}{6} \cB_{1}^{(3)}\left(\frac{1}{x}\right)^3-\frac{1}{2} \cB_{1}^{(2)}\left(\frac{1}{x}\right) \cB_{1}^{(3)}\left(\frac{1}{x}\right)^2+\frac{1}{6} \pi ^2 \cB_{1}^{(2)}\left(\frac{1}{x}\right)+\frac{1}{6} \pi ^2 \cB_{1}^{(3)}\left(\frac{1}{x}\right)\\
&\,-\cB_{3}^{(7)}\left(\frac{1}{x}\right)-\cB_{3}^{(8)}\left(\frac{1}{x}\right)+\frac{1}{2} i \pi  \sigma (x) \log ^2 2 +\frac{1}{6} \log ^3 2 -\frac{1}{3} \pi ^2 \log  2 +\zeta_3\,,\\
\cB_{3}^{(6)}\left(x\right) &\,= \frac{1}{2} \log ^2 2  \cB_{1}^{(1)}\left(\frac{1}{x}\right)-\frac{1}{2} \log ^2 2  \cB_{1}^{(2)}\left(\frac{1}{x}\right)+\frac{1}{2} \log  2  \cB_{1}^{(1)}\left(\frac{1}{x}\right)^2\\
&\,-\log  2  \cB_{1}^{(2)}\left(\frac{1}{x}\right) \cB_{1}^{(1)}\left(\frac{1}{x}\right)+\frac{1}{2} \log  2  \cB_{1}^{(2)}\left(\frac{1}{x}\right)^2+\frac{1}{6} \cB_{1}^{(1)}\left(\frac{1}{x}\right)^3\\
&\,-\frac{1}{2} \cB_{1}^{(2)}\left(\frac{1}{x}\right) \cB_{1}^{(1)}\left(\frac{1}{x}\right)^2+\frac{1}{2} \cB_{1}^{(2)}\left(\frac{1}{x}\right)^2 \cB_{1}^{(1)}\left(\frac{1}{x}\right)+\frac{1}{6} \pi ^2 \cB_{1}^{(1)}\left(\frac{1}{x}\right)\\
&\,-\frac{1}{6} \cB_{1}^{(2)}\left(\frac{1}{x}\right)^3-\frac{1}{6} \pi ^2 \cB_{1}^{(2)}\left(\frac{1}{x}\right)+\cB_{3}^{(8)}\left(\frac{1}{x}\right)+\frac{1}{6} \log ^3 2 +\frac{1}{6} \pi ^2 \log  2 \,,\\
\cB_{3}^{(7)}\left(x\right) &\,= -\frac{1}{2} \log ^2 2  \cB_{1}^{(3)}\left(\frac{1}{x}\right)-\frac{1}{2} \log ^2 2  \cB_{1}^{(2)}\left(\frac{1}{x}\right)+\log  2  \cB_{1}^{(2)}\left(\frac{1}{x}\right) \cB_{1}^{(3)}\left(\frac{1}{x}\right)\\
&\,+\frac{1}{6} \cB_{1}^{(3)}\left(\frac{1}{x}\right)^3-\frac{1}{2} \cB_{1}^{(2)}\left(\frac{1}{x}\right) \cB_{1}^{(3)}\left(\frac{1}{x}\right)^2+\frac{1}{6} \pi ^2 \cB_{1}^{(3)}\left(\frac{1}{x}\right)-\cB_{3}^{(5)}\left(\frac{1}{x}\right)\\
&\,-\cB_{3}^{(6)}\left(\frac{1}{x}\right)+\frac{1}{3} \log ^3 2 -\frac{1}{6} \pi ^2 \log  2 +\zeta_3\,,
\nonumber
\esp\eeq
\beq\bsp\nonumber
\cB_{3}^{(8)}\left(x\right) &\,= i \pi  \sigma (x) \log  2  \cB_{1}^{(2)}\left(\frac{1}{x}\right)+\frac{1}{2} \log ^2 2  \cB_{1}^{(2)}\left(\frac{1}{x}\right)-\frac{1}{2} \log  2  \cB_{1}^{(2)}\left(\frac{1}{x}\right)^2\\
&\,-\frac{1}{2} i \pi  \sigma (x) \cB_{1}^{(2)}\left(\frac{1}{x}\right)^2+\frac{1}{6} \cB_{1}^{(2)}\left(\frac{1}{x}\right)^3-\frac{1}{3} \pi ^2 \cB_{1}^{(2)}\left(\frac{1}{x}\right)+\cB_{3}^{(6)}\left(\frac{1}{x}\right)\\
&\,-\frac{1}{2} i \pi  \sigma (x) \log ^2 2 -\frac{1}{6} \log ^3 2 +\frac{1}{3} \pi ^2 \log  2 \,.
\esp\eeq
\subsection{Weight four}
\beq\bsp
\cB_{4}^{(1)}\left(x\right) &\,= -\frac{1}{6} i \pi  \sigma (x) \cB_{1}^{(1)}\left(\frac{1}{x}\right)^3-\frac{1}{24} \cB_{1}^{(1)}\left(\frac{1}{x}\right)^4+\frac{1}{6} \pi ^2 \cB_{1}^{(1)}\left(\frac{1}{x}\right)^2\\
&\,-\cB_{4}^{(1)}\left(\frac{1}{x}\right)+\frac{\pi ^4}{45}\,,\\
\cB_{4}^{(2)}\left(x\right) &\,= -\frac{1}{24} \cB_{1}^{(1)}\left(\frac{1}{x}\right)^4-\frac{1}{12} \pi ^2 \cB_{1}^{(1)}\left(\frac{1}{x}\right)^2-\cB_{4}^{(2)}\left(\frac{1}{x}\right)-\frac{7 \pi ^4}{360}\,,\\
\cB_{4}^{(3)}\left(x\right) &\,= -\frac{1}{24} \cB_{1}^{(1)}\left(\frac{1}{x}\right)^4+\frac{1}{6} \cB_{1}^{(2)}\left(\frac{1}{x}\right) \cB_{1}^{(1)}\left(\frac{1}{x}\right)^3-\frac{1}{4} \cB_{1}^{(2)}\left(\frac{1}{x}\right)^2 \cB_{1}^{(1)}\left(\frac{1}{x}\right)^2\\
&\,-\frac{1}{12} \pi ^2 \cB_{1}^{(1)}\left(\frac{1}{x}\right)^2+\frac{1}{6} \cB_{1}^{(2)}\left(\frac{1}{x}\right)^3 \cB_{1}^{(1)}\left(\frac{1}{x}\right)+\frac{1}{6} \pi ^2 \cB_{1}^{(2)}\left(\frac{1}{x}\right) \cB_{1}^{(1)}\left(\frac{1}{x}\right)\\
&\,-\frac{1}{24} \cB_{1}^{(2)}\left(\frac{1}{x}\right)^4-\frac{1}{12} \pi ^2 \cB_{1}^{(2)}\left(\frac{1}{x}\right)^2-\cB_{4}^{(5)}\left(\frac{1}{x}\right)-\frac{7 \pi ^4}{360}\,,\\
\cB_{4}^{(4)}\left(x\right) &\,= \cB_{4}^{(6)}\left(\frac{1}{x}\right)\,,\\
\cB_{4}^{(5)}\left(x\right) &\,= \frac{1}{6} i \pi  \sigma (x) \cB_{1}^{(2)}\left(\frac{1}{x}\right)^3-\frac{1}{24} \cB_{1}^{(2)}\left(\frac{1}{x}\right)^4+\frac{1}{6} \pi ^2 \cB_{1}^{(2)}\left(\frac{1}{x}\right)^2\\
&\,-\cB_{4}^{(3)}\left(\frac{1}{x}\right)+\frac{\pi ^4}{45}\,,\\
\cB_{4}^{(6)}\left(x\right) &\,= \cB_{4}^{(4)}\left(\frac{1}{x}\right)\,,\\
\cB_{4}^{(7)}\left(x\right) &\,= -\frac{1}{2} i \pi  \sigma (x) \log ^2 2  \cB_{1}^{(1)}\left(\frac{1}{x}\right)+\frac{1}{2} i \pi  \sigma (x) \log ^2 2  \cB_{1}^{(3)}\left(\frac{1}{x}\right)\\
&\,-\frac{1}{2} i \pi  \sigma (x) \log  2  \cB_{1}^{(1)}\left(\frac{1}{x}\right)^2+i \pi  \sigma (x) \log  2  \cB_{1}^{(3)}\left(\frac{1}{x}\right) \cB_{1}^{(1)}\left(\frac{1}{x}\right)\\
&\,-\frac{1}{2} i \pi  \sigma (x) \log  2  \cB_{1}^{(3)}\left(\frac{1}{x}\right)^2-\frac{1}{6} \log ^3 2  \cB_{1}^{(1)}\left(\frac{1}{x}\right)+\frac{1}{6} \log ^3 2  \cB_{1}^{(3)}\left(\frac{1}{x}\right)\\
&\,-\frac{1}{4} \log ^2 2  \cB_{1}^{(1)}\left(\frac{1}{x}\right)^2+\frac{1}{2} \log ^2 2  \cB_{1}^{(3)}\left(\frac{1}{x}\right) \cB_{1}^{(1)}\left(\frac{1}{x}\right)-\frac{1}{4} \log ^2 2  \cB_{1}^{(3)}\left(\frac{1}{x}\right)^2\\
&\,-\frac{1}{6} \log  2  \cB_{1}^{(1)}\left(\frac{1}{x}\right)^3+\frac{1}{2} \log  2  \cB_{1}^{(3)}\left(\frac{1}{x}\right) \cB_{1}^{(1)}\left(\frac{1}{x}\right)^2
\esp\eeq
\beq\bsp\nonumber
&\,-\frac{1}{2} \log  2  \cB_{1}^{(3)}\left(\frac{1}{x}\right)^2 \cB_{1}^{(1)}\left(\frac{1}{x}\right)+\frac{1}{3} \pi ^2 \log  2  \cB_{1}^{(1)}\left(\frac{1}{x}\right)+\frac{1}{6} \log  2  \cB_{1}^{(3)}\left(\frac{1}{x}\right)^3\\
&\,-\frac{1}{3} \pi ^2 \log  2  \cB_{1}^{(3)}\left(\frac{1}{x}\right)-\frac{1}{6} i \pi  \sigma (x) \cB_{1}^{(1)}\left(\frac{1}{x}\right)^3+\frac{1}{2} i \pi  \sigma (x) \cB_{1}^{(3)}\left(\frac{1}{x}\right) \cB_{1}^{(1)}\left(\frac{1}{x}\right)^2\\
&\,-\frac{1}{2} i \pi  \sigma (x) \cB_{1}^{(3)}\left(\frac{1}{x}\right)^2 \cB_{1}^{(1)}\left(\frac{1}{x}\right)+\frac{1}{6} i \pi  \sigma (x) \cB_{1}^{(3)}\left(\frac{1}{x}\right)^3-\frac{1}{24} \cB_{1}^{(1)}\left(\frac{1}{x}\right)^4\\
&\,+\frac{1}{6} \cB_{1}^{(3)}\left(\frac{1}{x}\right) \cB_{1}^{(1)}\left(\frac{1}{x}\right)^3-\frac{1}{4} \cB_{1}^{(3)}\left(\frac{1}{x}\right)^2 \cB_{1}^{(1)}\left(\frac{1}{x}\right)^2+\frac{1}{6} \pi ^2 \cB_{1}^{(1)}\left(\frac{1}{x}\right)^2\\
&\,+\frac{1}{6} \cB_{1}^{(3)}\left(\frac{1}{x}\right)^3 \cB_{1}^{(1)}\left(\frac{1}{x}\right)-\frac{1}{3} \pi ^2 \cB_{1}^{(3)}\left(\frac{1}{x}\right) \cB_{1}^{(1)}\left(\frac{1}{x}\right)-\frac{1}{24} \cB_{1}^{(3)}\left(\frac{1}{x}\right)^4\\
&\,+\frac{1}{6} \pi ^2 \cB_{1}^{(3)}\left(\frac{1}{x}\right)^2-\cB_{4}^{(11)}\left(\frac{1}{x}\right)-\frac{1}{6} i \pi  \sigma (x) \log ^3 2 -\frac{1}{24} \log ^4 2 +\frac{1}{6} \pi ^2 \log ^2 2 +\frac{\pi ^4}{45}\,,
\esp\eeq
\beq\bsp\nonumber
\cB_{4}^{(8)}\left(x\right) &\,= -\frac{1}{6} \log ^3 2  \cB_{1}^{(1)}\left(\frac{1}{x}\right)+\frac{1}{6} \log ^3 2  \cB_{1}^{(2)}\left(\frac{1}{x}\right)-\frac{1}{4} \log ^2 2  \cB_{1}^{(1)}\left(\frac{1}{x}\right)^2\\
&\,+\frac{1}{2} \log ^2 2  \cB_{1}^{(2)}\left(\frac{1}{x}\right) \cB_{1}^{(1)}\left(\frac{1}{x}\right)-\frac{1}{4} \log ^2 2  \cB_{1}^{(2)}\left(\frac{1}{x}\right)^2-\frac{1}{6} \log  2  \cB_{1}^{(1)}\left(\frac{1}{x}\right)^3\\
&\,+\frac{1}{2} \log  2  \cB_{1}^{(2)}\left(\frac{1}{x}\right) \cB_{1}^{(1)}\left(\frac{1}{x}\right)^2-\frac{1}{2} \log  2  \cB_{1}^{(2)}\left(\frac{1}{x}\right)^2 \cB_{1}^{(1)}\left(\frac{1}{x}\right)\\
&\,-\frac{1}{6} \pi ^2 \log  2  \cB_{1}^{(1)}\left(\frac{1}{x}\right)+\frac{1}{6} \log  2  \cB_{1}^{(2)}\left(\frac{1}{x}\right)^3+\frac{1}{6} \pi ^2 \log  2  \cB_{1}^{(2)}\left(\frac{1}{x}\right)\\
&\,-\frac{1}{24} \cB_{1}^{(1)}\left(\frac{1}{x}\right)^4+\frac{1}{6} \cB_{1}^{(2)}\left(\frac{1}{x}\right) \cB_{1}^{(1)}\left(\frac{1}{x}\right)^3-\frac{1}{4} \cB_{1}^{(2)}\left(\frac{1}{x}\right)^2 \cB_{1}^{(1)}\left(\frac{1}{x}\right)^2\\
&\,-\frac{1}{12} \pi ^2 \cB_{1}^{(1)}\left(\frac{1}{x}\right)^2+\frac{1}{6} \cB_{1}^{(2)}\left(\frac{1}{x}\right)^3 \cB_{1}^{(1)}\left(\frac{1}{x}\right)+\frac{1}{6} \pi ^2 \cB_{1}^{(2)}\left(\frac{1}{x}\right) \cB_{1}^{(1)}\left(\frac{1}{x}\right)\\
&\,-\frac{1}{24} \cB_{1}^{(2)}\left(\frac{1}{x}\right)^4-\frac{1}{12} \pi ^2 \cB_{1}^{(2)}\left(\frac{1}{x}\right)^2-\cB_{4}^{(12)}\left(\frac{1}{x}\right)-\frac{1}{24} \log ^4 2\\
&\, -\frac{1}{12} \pi ^2 \log ^2 2 -\frac{7 \pi ^4}{360}\,,\\
\cB_{4}^{(9)}\left(x\right) &\,= \cB_{4}^{(10)}\left(\frac{1}{x}\right)\,,\\
\cB_{4}^{(10)}\left(x\right) &\,= \cB_{4}^{(9)}\left(\frac{1}{x}\right)\,,\\
\cB_{4}^{(11)}\left(x\right) &\,= -\frac{1}{2} i \pi  \sigma (x) \log ^2 2  \cB_{1}^{(3)}\left(\frac{1}{x}\right)+\frac{1}{2} i \pi  \sigma (x) \log  2  \cB_{1}^{(3)}\left(\frac{1}{x}\right)^2+\frac{1}{6} \log ^3 2  \cB_{1}^{(3)}\left(\frac{1}{x}\right)\\
&\,-\frac{1}{4} \log ^2 2  \cB_{1}^{(3)}\left(\frac{1}{x}\right)^2+\frac{1}{6} \log  2  \cB_{1}^{(3)}\left(\frac{1}{x}\right)^3-\frac{1}{3} \pi ^2 \log  2  \cB_{1}^{(3)}\left(\frac{1}{x}\right)\\
&\,-\frac{1}{6} i \pi  \sigma (x) \cB_{1}^{(3)}\left(\frac{1}{x}\right)^3-\frac{1}{24} \cB_{1}^{(3)}\left(\frac{1}{x}\right)^4+\frac{1}{6} \pi ^2 \cB_{1}^{(3)}\left(\frac{1}{x}\right)^2-\cB_{4}^{(7)}\left(\frac{1}{x}\right)\\
&\,+\frac{1}{6} i \pi  \sigma (x) \log ^3 2 -\frac{1}{24} \log ^4 2 +\frac{1}{6} \pi ^2 \log ^2 2 +\frac{\pi ^4}{45}\,,
\esp\eeq
\beq\bsp\nonumber
\cB_{4}^{(12)}\left(x\right) &\,= \frac{1}{2} i \pi  \sigma (x) \log ^2 2  \cB_{1}^{(2)}\left(\frac{1}{x}\right)-\frac{1}{2} i \pi  \sigma (x) \log  2  \cB_{1}^{(2)}\left(\frac{1}{x}\right)^2+\frac{1}{6} \log ^3 2  \cB_{1}^{(2)}\left(\frac{1}{x}\right)\\
&\,-\frac{1}{4} \log ^2 2  \cB_{1}^{(2)}\left(\frac{1}{x}\right)^2+\frac{1}{6} \log  2  \cB_{1}^{(2)}\left(\frac{1}{x}\right)^3-\frac{1}{3} \pi ^2 \log  2  \cB_{1}^{(2)}\left(\frac{1}{x}\right)\\
&\,+\frac{1}{6} i \pi  \sigma (x) \cB_{1}^{(2)}\left(\frac{1}{x}\right)^3-\frac{1}{24} \cB_{1}^{(2)}\left(\frac{1}{x}\right)^4+\frac{1}{6} \pi ^2 \cB_{1}^{(2)}\left(\frac{1}{x}\right)^2-\cB_{4}^{(8)}\left(\frac{1}{x}\right)\\
&\,-\frac{1}{6} i \pi  \sigma (x) \log ^3 2 -\frac{1}{24} \log ^4 2 +\frac{1}{6} \pi ^2 \log ^2 2 +\frac{\pi ^4}{45}\,,\\
\cB_{4}^{(13)}\left(x\right) &\,= -\frac{2}{3} \cB_{1}^{(1)}\left(\frac{1}{x}\right)^4+\frac{4}{3} \cB_{1}^{(2)}\left(\frac{1}{x}\right) \cB_{1}^{(1)}\left(\frac{1}{x}\right)^3+\frac{4}{3} \cB_{1}^{(3)}\left(\frac{1}{x}\right) \cB_{1}^{(1)}\left(\frac{1}{x}\right)^3\\
&\,-\cB_{1}^{(2)}\left(\frac{1}{x}\right)^2 \cB_{1}^{(1)}\left(\frac{1}{x}\right)^2-\cB_{1}^{(3)}\left(\frac{1}{x}\right)^2 \cB_{1}^{(1)}\left(\frac{1}{x}\right)^2\\
&\,-2 \cB_{1}^{(2)}\left(\frac{1}{x}\right) \cB_{1}^{(3)}\left(\frac{1}{x}\right) \cB_{1}^{(1)}\left(\frac{1}{x}\right)^2-\frac{1}{3} \pi ^2 \cB_{1}^{(1)}\left(\frac{1}{x}\right)^2\\
&\,+\frac{1}{3} \cB_{1}^{(2)}\left(\frac{1}{x}\right)^3 \cB_{1}^{(1)}\left(\frac{1}{x}\right)+\frac{1}{3} \cB_{1}^{(3)}\left(\frac{1}{x}\right)^3 \cB_{1}^{(1)}\left(\frac{1}{x}\right)\\
&\,+\cB_{1}^{(2)}\left(\frac{1}{x}\right) \cB_{1}^{(3)}\left(\frac{1}{x}\right)^2 \cB_{1}^{(1)}\left(\frac{1}{x}\right)+\frac{1}{3} \pi ^2 \cB_{1}^{(2)}\left(\frac{1}{x}\right) \cB_{1}^{(1)}\left(\frac{1}{x}\right)\\
&\,+\cB_{1}^{(2)}\left(\frac{1}{x}\right)^2 \cB_{1}^{(3)}\left(\frac{1}{x}\right) \cB_{1}^{(1)}\left(\frac{1}{x}\right)+\frac{1}{3} \pi ^2 \cB_{1}^{(3)}\left(\frac{1}{x}\right) \cB_{1}^{(1)}\left(\frac{1}{x}\right)\\
&\,-\frac{1}{24} \cB_{1}^{(2)}\left(\frac{1}{x}\right)^4-\frac{1}{24} \cB_{1}^{(3)}\left(\frac{1}{x}\right)^4-\frac{1}{6} \cB_{1}^{(2)}\left(\frac{1}{x}\right) \cB_{1}^{(3)}\left(\frac{1}{x}\right)^3\\
&\,-\frac{1}{12} \pi ^2 \cB_{1}^{(2)}\left(\frac{1}{x}\right)^2-\frac{1}{4} \cB_{1}^{(2)}\left(\frac{1}{x}\right)^2 \cB_{1}^{(3)}\left(\frac{1}{x}\right)^2-\frac{1}{12} \pi ^2 \cB_{1}^{(3)}\left(\frac{1}{x}\right)^2\\
&\,-\frac{1}{6} \cB_{1}^{(2)}\left(\frac{1}{x}\right)^3 \cB_{1}^{(3)}\left(\frac{1}{x}\right)-\frac{1}{6} \pi ^2 \cB_{1}^{(2)}\left(\frac{1}{x}\right) \cB_{1}^{(3)}\left(\frac{1}{x}\right)\\
&\,-\cB_{4}^{(14)}\left(\frac{1}{x}\right)-\frac{7 \pi ^4}{360}\,,\\
\cB_{4}^{(14)}\left(x\right) &\,= \frac{1}{6} i \pi  \sigma (x) \cB_{1}^{(2)}\left(\frac{1}{x}\right)^3+\frac{1}{2} i \pi  \sigma (x) \cB_{1}^{(3)}\left(\frac{1}{x}\right) \cB_{1}^{(2)}\left(\frac{1}{x}\right)^2\\
&\,+\frac{1}{2} i \pi  \sigma (x) \cB_{1}^{(3)}\left(\frac{1}{x}\right)^2 \cB_{1}^{(2)}\left(\frac{1}{x}\right)+\frac{1}{6} i \pi  \sigma (x) \cB_{1}^{(3)}\left(\frac{1}{x}\right)^3\\
&\,-\frac{1}{24} \cB_{1}^{(2)}\left(\frac{1}{x}\right)^4-\frac{1}{6} \cB_{1}^{(3)}\left(\frac{1}{x}\right) \cB_{1}^{(2)}\left(\frac{1}{x}\right)^3-\frac{1}{4} \cB_{1}^{(3)}\left(\frac{1}{x}\right)^2 \cB_{1}^{(2)}\left(\frac{1}{x}\right)^2\\
&\,+\frac{1}{6} \pi ^2 \cB_{1}^{(2)}\left(\frac{1}{x}\right)^2-\frac{1}{6} \cB_{1}^{(3)}\left(\frac{1}{x}\right)^3 \cB_{1}^{(2)}\left(\frac{1}{x}\right)+\frac{1}{3} \pi ^2 \cB_{1}^{(3)}\left(\frac{1}{x}\right) \cB_{1}^{(2)}\left(\frac{1}{x}\right)\\
&\,-\frac{1}{24} \cB_{1}^{(3)}\left(\frac{1}{x}\right)^4+\frac{1}{6} \pi ^2 \cB_{1}^{(3)}\left(\frac{1}{x}\right)^2-\cB_{4}^{(13)}\left(\frac{1}{x}\right)+\frac{\pi ^4}{45}\,,
\esp\eeq
\beq\bsp\nonumber
\cB_{4}^{(15)}\left(x\right) &\,= 4 i \pi  \sigma (x) \log ^2 2  \cB_{1}^{(1)}\left(\frac{1}{x}\right)-8 i \pi  \sigma (x) \log ^2 2  \cB_{1}^{(3)}\left(\frac{1}{x}\right)\\
&\,+2 i \pi  \sigma (x) \log  2  \cB_{1}^{(1)}\left(\frac{1}{x}\right)^2-8 i \pi  \sigma (x) \log  2  \cB_{1}^{(3)}\left(\frac{1}{x}\right) \cB_{1}^{(1)}\left(\frac{1}{x}\right)\\
&\,+8 i \pi  \sigma (x) \log  2  \cB_{1}^{(3)}\left(\frac{1}{x}\right)^2+\frac{1}{3} i \pi  \sigma (x) \cB_{1}^{(1)}\left(\frac{1}{x}\right)^3\\
&\,-2 i \pi  \sigma (x) \cB_{1}^{(3)}\left(\frac{1}{x}\right) \cB_{1}^{(1)}\left(\frac{1}{x}\right)^2+4 i \pi  \sigma (x) \cB_{1}^{(3)}\left(\frac{1}{x}\right)^2 \cB_{1}^{(1)}\left(\frac{1}{x}\right)\\
&\,-\frac{8}{3} i \pi  \sigma (x) \cB_{1}^{(3)}\left(\frac{1}{x}\right)^3+\cB_{4}^{(15)}\left(\frac{1}{x}\right)+\frac{8}{3} i \pi  \sigma (x) \log ^3 2 \,,\\
\cB_{4}^{(16)}\left(x\right) &\,= \frac{1}{12} i \pi ^3 \sigma (x) \cB_{1}^{(1)}\left(\frac{1}{x}\right)-\frac{3}{2} \zeta_3 \cB_{1}^{(1)}\left(\frac{1}{x}\right)+\frac{1}{24} \cB_{1}^{(1)}\left(\frac{1}{x}\right)^4\\
&\,+\frac{1}{2} \cB_{2}^{(1)}\left(\frac{1}{x}\right) \cB_{1}^{(1)}\left(\frac{1}{x}\right)^2+\frac{1}{12} \pi ^2 \cB_{1}^{(1)}\left(\frac{1}{x}\right)^2-2 \cB_{3}^{(1)}\left(\frac{1}{x}\right) \cB_{1}^{(1)}\left(\frac{1}{x}\right)\\
&\,+\frac{1}{6} \pi ^2 \cB_{2}^{(1)}\left(\frac{1}{x}\right)+3 \cB_{4}^{(1)}\left(\frac{1}{x}\right)+\cB_{4}^{(2)}\left(\frac{1}{x}\right)+\cB_{4}^{(16)}\left(\frac{1}{x}\right)-\frac{37 \pi ^4}{720}\,,\\
\cB_{4}^{(17)}\left(x\right) &\,= -\frac{19}{24} \cB_{1}^{(3)}\left(\frac{1}{x}\right)^4+\frac{7}{6} \cB_{1}^{(1)}\left(\frac{1}{x}\right) \cB_{1}^{(3)}\left(\frac{1}{x}\right)^3+\frac{5}{2} \log  2  \cB_{1}^{(3)}\left(\frac{1}{x}\right)^3\\
&\,-\frac{9}{4} \log ^2 2  \cB_{1}^{(3)}\left(\frac{1}{x}\right)^2-\cB_{1}^{(1)}\left(\frac{1}{x}\right) \cB_{1}^{(2)}\left(\frac{1}{x}\right) \cB_{1}^{(3)}\left(\frac{1}{x}\right)^2\\
&\,-\frac{1}{2} \cB_{2}^{(1)}\left(\frac{1}{x}\right) \cB_{1}^{(3)}\left(\frac{1}{x}\right)^2+\frac{1}{2} \cB_{2}^{(2)}\left(\frac{1}{x}\right) \cB_{1}^{(3)}\left(\frac{1}{x}\right)^2\\
&\,-\cB_{2}^{(3)}\left(\frac{1}{x}\right) \cB_{1}^{(3)}\left(\frac{1}{x}\right)^2-\cB_{1}^{(1)}\left(\frac{1}{x}\right) \log  2  \cB_{1}^{(3)}\left(\frac{1}{x}\right)^2\\
&\,-\cB_{1}^{(2)}\left(\frac{1}{x}\right) \log  2  \cB_{1}^{(3)}\left(\frac{1}{x}\right)^2+\frac{19}{24} \pi ^2 \cB_{1}^{(3)}\left(\frac{1}{x}\right)^2-\frac{1}{3} \cB_{1}^{(2)}\left(\frac{1}{x}\right)^3 \cB_{1}^{(3)}\left(\frac{1}{x}\right)\\
&\,+\frac{5}{6} \log ^3 2  \cB_{1}^{(3)}\left(\frac{1}{x}\right)+\cB_{1}^{(1)}\left(\frac{1}{x}\right) \cB_{1}^{(2)}\left(\frac{1}{x}\right)^2 \cB_{1}^{(3)}\left(\frac{1}{x}\right)\\
&\,+\frac{3}{2} \cB_{1}^{(1)}\left(\frac{1}{x}\right) \log ^2 2  \cB_{1}^{(3)}\left(\frac{1}{x}\right)-\frac{1}{4} \pi ^2 \cB_{1}^{(1)}\left(\frac{1}{x}\right) \cB_{1}^{(3)}\left(\frac{1}{x}\right)\\
&\,-\frac{1}{6} \pi ^2 \cB_{1}^{(2)}\left(\frac{1}{x}\right) \cB_{1}^{(3)}\left(\frac{1}{x}\right)+\cB_{1}^{(1)}\left(\frac{1}{x}\right) \cB_{2}^{(3)}\left(\frac{1}{x}\right) \cB_{1}^{(3)}\left(\frac{1}{x}\right)\\
&\,+2 \cB_{3}^{(1)}\left(\frac{1}{x}\right) \cB_{1}^{(3)}\left(\frac{1}{x}\right)-2 \cB_{3}^{(5)}\left(\frac{1}{x}\right) \cB_{1}^{(3)}\left(\frac{1}{x}\right)+2 \cB_{3}^{(7)}\left(\frac{1}{x}\right) \cB_{1}^{(3)}\left(\frac{1}{x}\right)\\
&\,+2 \cB_{3}^{(8)}\left(\frac{1}{x}\right) \cB_{1}^{(3)}\left(\frac{1}{x}\right)+\cB_{1}^{(2)}\left(\frac{1}{x}\right)^2 \log  2  \cB_{1}^{(3)}\left(\frac{1}{x}\right)\\
&\,-\cB_{1}^{(1)}\left(\frac{1}{x}\right) \cB_{1}^{(2)}\left(\frac{1}{x}\right) \log  2  \cB_{1}^{(3)}\left(\frac{1}{x}\right)-\frac{2}{3} \pi ^2 \log  2  \cB_{1}^{(3)}\left(\frac{1}{x}\right)\\
&\,+\frac{1}{2} i \pi  \log ^2 2  \sigma (x) \cB_{1}^{(3)}\left(\frac{1}{x}\right)-\frac{1}{12} i \pi ^3 \sigma (x) \cB_{1}^{(3)}\left(\frac{1}{x}\right)+3 \zeta_3 \cB_{1}^{(3)}\left(\frac{1}{x}\right)
\esp\eeq
\beq\bsp\nonumber
&\,+\frac{1}{8} \log ^4 2 -\frac{1}{3} \cB_{1}^{(1)}\left(\frac{1}{x}\right) \log ^3 2 +\frac{1}{2} \cB_{2}^{(2)}\left(\frac{1}{x}\right)^2+\frac{1}{12} \pi ^2 \log ^2 2\\
&\, +\frac{1}{6} \pi ^2 \cB_{2}^{(3)}\left(\frac{1}{x}\right)+2 \cB_{1}^{(1)}\left(\frac{1}{x}\right) \cB_{3}^{(5)}\left(\frac{1}{x}\right)-\frac{3}{2} \cB_{4}^{(1)}\left(\frac{1}{x}\right)-\frac{1}{2} \cB_{4}^{(2)}\left(\frac{1}{x}\right)\\
&\,-2 \cB_{4}^{(4)}\left(\frac{1}{x}\right)-2 \cB_{4}^{(6)}\left(\frac{1}{x}\right)-3 \cB_{4}^{(7)}\left(\frac{1}{x}\right)+3 \cB_{4}^{(11)}\left(\frac{1}{x}\right)-\frac{3}{4} \cB_{4}^{(15)}\left(\frac{1}{x}\right)\\
&\,-\cB_{4}^{(16)}\left(\frac{1}{x}\right)-\cB_{4}^{(17)}\left(\frac{1}{x}\right)+\frac{1}{6} \pi ^2 \cB_{1}^{(1)}\left(\frac{1}{x}\right) \log  2 -\frac{1}{6} \pi ^2 \cB_{1}^{(2)}\left(\frac{1}{x}\right) \log  2\\
&\, -\frac{1}{2} i \pi  \log ^3 2  \sigma (x)+\frac{1}{12} i \pi ^3 \log  2  \sigma (x)-\frac{7}{4} \cB_{1}^{(1)}\left(\frac{1}{x}\right) \zeta_3+6 \text{Li}_4\left(\frac{1}{2}\right)-\frac{\pi ^4}{160}\,,\\
\cB_{4}^{(18)}\left(x\right) &\,= \frac{35}{96} \cB_{1}^{(2)}\left(\frac{1}{x}\right)^4-\frac{5}{12} \cB_{1}^{(1)}\left(\frac{1}{x}\right) \cB_{1}^{(2)}\left(\frac{1}{x}\right)^3-\frac{1}{24} \cB_{1}^{(3)}\left(\frac{1}{x}\right) \cB_{1}^{(2)}\left(\frac{1}{x}\right)^3\\
&\,-\frac{1}{2} \log  2  \cB_{1}^{(2)}\left(\frac{1}{x}\right)^3-\frac{1}{6} i \pi  \sigma (x) \cB_{1}^{(2)}\left(\frac{1}{x}\right)^3-\frac{1}{16} \cB_{1}^{(3)}\left(\frac{1}{x}\right)^2 \cB_{1}^{(2)}\left(\frac{1}{x}\right)^2\\
&\,+\frac{1}{4} \log ^2 2  \cB_{1}^{(2)}\left(\frac{1}{x}\right)^2+\frac{1}{4} \cB_{1}^{(1)}\left(\frac{1}{x}\right) \cB_{1}^{(3)}\left(\frac{1}{x}\right) \cB_{1}^{(2)}\left(\frac{1}{x}\right)^2\\
&\,+\frac{1}{2} \cB_{2}^{(1)}\left(\frac{1}{x}\right) \cB_{1}^{(2)}\left(\frac{1}{x}\right)^2-\frac{1}{2} \cB_{2}^{(2)}\left(\frac{1}{x}\right) \cB_{1}^{(2)}\left(\frac{1}{x}\right)^2+\cB_{2}^{(3)}\left(\frac{1}{x}\right) \cB_{1}^{(2)}\left(\frac{1}{x}\right)^2\\
&\,+\frac{1}{48} \pi ^2 \cB_{1}^{(2)}\left(\frac{1}{x}\right)^2-\frac{1}{24} \cB_{1}^{(3)}\left(\frac{1}{x}\right)^3 \cB_{1}^{(2)}\left(\frac{1}{x}\right)-\frac{1}{6} \log ^3 2  \cB_{1}^{(2)}\left(\frac{1}{x}\right)\\
&\,+\frac{1}{4} \cB_{1}^{(1)}\left(\frac{1}{x}\right) \cB_{1}^{(3)}\left(\frac{1}{x}\right)^2 \cB_{1}^{(2)}\left(\frac{1}{x}\right)+\frac{1}{2} \cB_{1}^{(1)}\left(\frac{1}{x}\right) \log ^2 2  \cB_{1}^{(2)}\left(\frac{1}{x}\right)\\
&\,-\frac{1}{12} \pi ^2 \cB_{1}^{(1)}\left(\frac{1}{x}\right) \cB_{1}^{(2)}\left(\frac{1}{x}\right)-\frac{1}{24} \pi ^2 \cB_{1}^{(3)}\left(\frac{1}{x}\right) \cB_{1}^{(2)}\left(\frac{1}{x}\right)\\
&\,-\cB_{1}^{(1)}\left(\frac{1}{x}\right) \cB_{2}^{(3)}\left(\frac{1}{x}\right) \cB_{1}^{(2)}\left(\frac{1}{x}\right)-2 \cB_{3}^{(6)}\left(\frac{1}{x}\right) \cB_{1}^{(2)}\left(\frac{1}{x}\right)\\
&\,-2 \cB_{3}^{(8)}\left(\frac{1}{x}\right) \cB_{1}^{(2)}\left(\frac{1}{x}\right)-i \pi  \cB_{2}^{(3)}\left(\frac{1}{x}\right) \sigma (x) \cB_{1}^{(2)}\left(\frac{1}{x}\right)\\
&\,+\frac{7}{2} \zeta_3 \cB_{1}^{(2)}\left(\frac{1}{x}\right)+\frac{55}{96} \cB_{1}^{(3)}\left(\frac{1}{x}\right)^4+\frac{1}{8} \log ^4 2 -\frac{7}{12} \cB_{1}^{(1)}\left(\frac{1}{x}\right) \cB_{1}^{(3)}\left(\frac{1}{x}\right)^3\\
&\,-\frac{1}{3} \cB_{1}^{(1)}\left(\frac{1}{x}\right) \log ^3 2 -\cB_{1}^{(3)}\left(\frac{1}{x}\right) \log ^3 2 -\frac{29}{48} \pi ^2 \cB_{1}^{(3)}\left(\frac{1}{x}\right)^2+\frac{1}{2} \cB_{2}^{(1)}\left(\frac{1}{x}\right)^2\\
&\,+\frac{3}{2} \cB_{1}^{(3)}\left(\frac{1}{x}\right)^2 \log ^2 2 -\frac{1}{12} \pi ^2 \log ^2 2 -\cB_{2}^{(1)}\left(\frac{1}{x}\right) \cB_{2}^{(2)}\left(\frac{1}{x}\right)-\frac{1}{6} \pi ^2 \cB_{2}^{(3)}\left(\frac{1}{x}\right)\\
&\,-2 \cB_{1}^{(3)}\left(\frac{1}{x}\right) \cB_{3}^{(1)}\left(\frac{1}{x}\right)+2 \cB_{1}^{(1)}\left(\frac{1}{x}\right) \cB_{3}^{(6)}\left(\frac{1}{x}\right)+\frac{3}{2} \cB_{4}^{(1)}\left(\frac{1}{x}\right)+\frac{1}{2} \cB_{4}^{(2)}\left(\frac{1}{x}\right)\\
&\,-2 \cB_{4}^{(3)}\left(\frac{1}{x}\right)+4 \cB_{4}^{(4)}\left(\frac{1}{x}\right)+2 \cB_{4}^{(5)}\left(\frac{1}{x}\right)+4 \cB_{4}^{(6)}\left(\frac{1}{x}\right)+6 \cB_{4}^{(7)}\left(\frac{1}{x}\right)\\
&\,+3 \cB_{4}^{(8)}\left(\frac{1}{x}\right)-6 \cB_{4}^{(11)}\left(\frac{1}{x}\right)-3 \cB_{4}^{(12)}\left(\frac{1}{x}\right)+\frac{1}{4} \cB_{4}^{(13)}\left(\frac{1}{x}\right)-\frac{1}{4} \cB_{4}^{(14)}\left(\frac{1}{x}\right)
\esp\eeq
\beq\bsp\nonumber
&\,+\frac{3}{4} \cB_{4}^{(15)}\left(\frac{1}{x}\right)+\cB_{4}^{(16)}\left(\frac{1}{x}\right)-\cB_{4}^{(18)}\left(\frac{1}{x}\right)-\cB_{1}^{(3)}\left(\frac{1}{x}\right)^3 \log  2\\
&\, +\frac{1}{6} \pi ^2 \cB_{1}^{(1)}\left(\frac{1}{x}\right) \log  2 +\frac{1}{2} \pi ^2 \cB_{1}^{(3)}\left(\frac{1}{x}\right) \log  2 +\frac{1}{6} i \pi  \log ^3 2  \sigma (x)+2 i \pi  \cB_{3}^{(6)}\left(\frac{1}{x}\right) \sigma (x)\\
&\,+\frac{1}{12} i \pi ^3 \log  2  \sigma (x)-\frac{7}{4} i \pi  \zeta_3 \sigma (x)-\frac{7}{4} \cB_{1}^{(1)}\left(\frac{1}{x}\right) \zeta_3-\frac{3}{2} \cB_{1}^{(3)}\left(\frac{1}{x}\right) \zeta_3\\
&\,-6 \text{Li}_4\left(\frac{1}{2}\right)-\frac{37 \pi ^4}{1440}\,.
\esp\eeq

\section{Expression of HPL's in terms of the spanning set}
\label{app:results}
In this appendix we present the results for expressing all HPL's up to weight four in terms of the spanning set $\{\cB_i^{(j)}\}$. We restrict ourselves to giving the expression for a minimal set of HPL's out of which all other cases can be obtained via shuffle relations\footnote{A set of text files containing the expressions for all HPL's up to weight four (for $x\in[0,1]$) in Mathematica is included in the arXiv distribution.}.
\subsection{Results for weight two}

\beq\bsp\nonumber
H(-1,1; x) &\,= \log  2  \log  (1-x) -\log  (1-x)  \log  (1+x) -\frac{1}{2} \log ^2 2 -\text{Li}_2\left(\frac{1-x}{2}\right)+\frac{\pi ^2}{12}\,,\\
H(0,-1; x) &\,= -\text{Li}_2(-x)\,,\\
H(0,1; x) &\,= \text{Li}_2(x)\,.
\esp\eeq

\subsection{Results for weight three}

\beq\bsp\nonumber
H(-1,1,-1; x) &\,= -\text{Li}_2\left(\frac{1-x}{2}\right) \log  (1+x) -\frac{3}{2} \log ^2 2  \log  (1+x) \\
&\,+\log  2  \log ^2 (1+x) -\log  (1-x)  \log ^2 (1+x) +\log  2  \log  (1-x)  \log  (1+x)\\
&\, +\frac{1}{4} \pi ^2 \log  (1+x) +\frac{1}{3} \log ^3 2 -\frac{1}{6} \pi ^2 \log  2 -2 \text{Li}_3\left(\frac{1+x}{2}\right)+\frac{7 \zeta_3}{4}\,,\\
H(-1,1,1; x) &\,= \text{Li}_2\left(\frac{1-x}{2}\right) \log  (1-x) -\frac{1}{2} \log  2  \log ^2 (1-x) +\frac{1}{2} \log ^2 (1-x)  \log  (1+x)\\
&\, +\frac{1}{6} \log ^3 2 -\frac{1}{12} \pi ^2 \log  2 -\text{Li}_3\left(\frac{1-x}{2}\right)+\frac{7 \zeta_3}{8}\,,\\
H(0,-1,-1; x) &\,= -\text{Li}_2(-x) \log  (1+x) +\frac{1}{6} \log ^3 (1+x) -\frac{1}{2} \log  x  \log ^2 (1+x) -\frac{\pi ^2}{6} \log  (1+x)\\
&\, -\text{Li}_3\left(\frac{1}{1+x}\right)+\zeta_3\,,\\
H(0,-1,1; x) &\,= \text{Li}_2(-x) \log  (1-x) -\frac{1}{6} \log ^3 (1-x) -\frac{1}{2} \log ^2 2  \log  (1-x) +\frac{1}{2} \log  2  \log ^2 (1-x)\\
&\, +\frac{1}{2} \log ^2 (1-x)  \log  x -\frac{\pi ^2}{12}  \log  (1-x) +\frac{1}{6} \log ^3 2 -\frac{\pi ^2}{12}  \log  2 -\text{Li}_3\left(\frac{1-x}{2}\right)\\
&\,+\text{Li}_3(1-x)-\text{Li}_3(-x)+\text{Li}_3(x)+\text{Li}_3\left(\frac{2 x}{x-1}\right)-\frac{1}{8}\zeta_3\,,
\esp\eeq
\beq\bsp\nonumber
H(0,0,-1; x) &\,= -\text{Li}_3(-x)\,,\\
H(0,0,1; x) &\,= \text{Li}_3(x)\,,\\
H(0,1,-1; x) &\,= \text{Li}_2(x) \log  (1+x) +\frac{1}{6} \log ^3 (1-x) -\frac{1}{2} \log ^2 2  \log  (1+x) -\frac{1}{2} \log  2  \log ^2 (1-x)\\
&\, -\frac{1}{2} \log  (1-x)  \log ^2 (1+x) -\frac{1}{2} \log ^2 (1-x)  \log  x +\log  2  \log  (1-x)  \log  (1+x)\\
&\, +\frac{\pi ^2}{6} \log  (1-x) +\log  (1-x)  \log  x  \log  (1+x) +\frac{\pi ^2}{12} \log  (1+x) +\frac{1}{6} \log ^3 2\\
&\, -\frac{\pi ^2}{12} \log  2 +\text{Li}_3(-x)-\text{Li}_3(x)-\text{Li}_3\left(\frac{2 x}{x-1}\right)+\text{Li}_3\left(\frac{1}{1+x}\right)\\
&\,-\text{Li}_3\left(\frac{1-x}{1+x}\right)-\text{Li}_3\left(\frac{1+x}{2}\right)+\frac{7}{8}\zeta_3\,,\\
H(0,1,1; x) &\,= -\text{Li}_2(x) \log  (1-x) -\frac{1}{2} \log  x  \log ^2 (1-x) +\frac{\pi ^2}{6}  \log  (1-x) -\text{Li}_3(1-x)+\zeta_3\,.
\esp\eeq

\subsection{Results for weight four}

\beq\bsp\nonumber
H(-1,1,-1,-1; x) &\,= - \frac{1}{2} \text{Li}_2\left(\frac{1-x}{2}\right) \log ^2 (1+x) -2 \text{Li}_3\left(\frac{1+x}{2}\right) \log  (1+x)\\
&\, -\frac{7}{8} \zeta_3 \log  (1+x) -\frac{1}{6} \log ^3 2  \log  (1+x) +\frac{1}{2} \log  2  \log ^3 (1+x)\\
&\, -\frac{1}{2} \log  (1-x)  \log ^3 (1+x) -\frac{1}{2} \log ^2 2  \log ^2 (1+x)\\
&\, +\frac{1}{2} \log  2  \log  (1-x)  \log ^2 (1+x) +\frac{\pi ^2}{12} \log ^2 (1+x) +\frac{\pi ^2}{12} \log  2  \log  (1+x)\\
&\, +3 \text{Li}_4\left(\frac{1+x}{2}\right)-3 \text{Li}_4\left(\frac{1}{2}\right)\,,
\\
H(-1,1,-1,1; x) &\,= \frac{1}{2}  \text{Li}_2\left(\frac{1-x}{2}\right) \log ^2 2 -\text{Li}_2\left(\frac{1-x}{2}\right) \log  2  \log  (1-x)\\
&\, +\text{Li}_2\left(\frac{1-x}{2}\right) \log  (1-x)  \log  (1+x) +2 \text{Li}_3\left(\frac{1+x}{2}\right) \log  (1-x)\\
&\, -2 \zeta_3 \log  (1-x) +\frac{1}{4} \zeta_3 \log  (1+x) +\frac{1}{12} \log ^4 (1+x) -\frac{7}{6} \log ^3 2  \log  (1-x)\\
&\, +\frac{1}{3} \log ^3 2  \log  (1+x) -\frac{1}{3} \log  (1-x)  \log ^3 (1+x) +\log ^2 2  \log ^2 (1-x)\\
&\, -\frac{1}{2} \log ^2 2  \log ^2 (1+x) +\frac{3}{2} \log ^2 2  \log  (1-x)  \log  (1+x)\\
&\, -2 \log  2  \log ^2 (1-x)  \log  (1+x) +\log ^2 (1-x)  \log ^2 (1+x) +\frac{\pi ^2}{6}  \log ^2 (1+x)\\
&\, +\frac{5}{12} \pi ^2 \log  2  \log  (1-x) -\frac{\pi ^2}{6} \log  2  \log  (1+x) -\frac{5}{12} \pi ^2 \log  (1-x)  \log  (1+x)\\
&\, +\frac{1}{8} \log ^4 2 -\frac{\pi ^2}{24} \log ^2 2 +\frac{1}{2}\text{Li}_2\left(\frac{1-x}{2}\right)^2-\frac{\pi ^2}{12}\text{Li}_2\left(\frac{1-x}{2}\right)+2 \text{Li}_4\left(\frac{1-x}{2}\right)\\
&\,+2 \text{Li}_4\left(\frac{x-1}{x+1}\right)-2 \text{Li}_4\left(\frac{1+x}{2}\right)+\frac{11}{480}\pi ^4\,,
\esp\eeq
\beq\bsp\nonumber
H(0,-1,-1,-1; x) &\,= -\frac{1}{2} \text{Li}_2(-x) \log ^2 (1+x) -\text{Li}_3\left(\frac{1}{1+x}\right) \log  (1+x) +\frac{1}{8} \log ^4 (1+x)\\
&\, -\frac{1}{3} \log  x  \log ^3 (1+x) -\frac{\pi ^2}{12}  \log ^2 (1+x) -\text{Li}_4\left(\frac{1}{1+x}\right)+\frac{\pi ^4}{90}\,,
\\
H(0,-1,-1,1; x) &\,= \text{Li}_{2,2}(-1,x)+\text{Li}_{2,2}\left(\frac{1}{2},\frac{2 x}{x+1}\right)+\frac{1}{2} \text{Li}_2\left(\frac{1-x}{2}\right) \log ^2 (1+x)\\
&\, -\frac{1}{2} \text{Li}_2(-x) \log ^2 (1+x) +\frac{1}{2} \text{Li}_2(x) \log ^2 (1+x) +\frac{1}{2} \text{Li}_2(-x) \log ^2 2\\
&\, +\text{Li}_2(-x) \log  (1-x)  \log  (1+x) -2 \text{Li}_3(x) \log  (1+x)\\
&\, -2 \text{Li}_3\left(\frac{2 x}{x-1}\right) \log  (1+x) -2 \text{Li}_3\left(\frac{1-x}{1+x}\right) \log  (1+x)\\
&\, -\text{Li}_2(-x) \log  2  \log  (1-x) +\text{Li}_3\left(\frac{1}{1+x}\right) \log  (1-x) -\frac{5}{4} \zeta_3 \log  (1+x)\\
&\, -\frac{7}{8} \zeta_3 \log  (1-x) +\frac{5}{8} \log ^4 (1+x) -\frac{3}{2} \log  2  \log ^3 (1+x)\\
&\, -\frac{2}{3} \log  (1-x)  \log ^3 (1+x) -\log  x  \log ^3 (1+x) -\frac{1}{2} \log ^3 2  \log  (1+x)\\
&\, +\frac{1}{3} \log ^3 (1-x)  \log  (1+x) +\log ^2 2  \log ^2 (1+x) +\frac{3}{2} \log  2  \log  (1-x)  \log ^2 (1+x)\\
&\, +2 \log  (1-x)  \log  x  \log ^2 (1+x) -\frac{5}{12} \pi ^2 \log ^2 (1+x)\\
&\, -\log  2  \log ^2 (1-x)  \log  (1+x) -\log ^2 (1-x)  \log  x  \log  (1+x)\\
&\, +\frac{\pi ^2}{4}  \log  2  \log  (1+x) +\frac{5}{12} \pi ^2 \log  (1-x)  \log  (1+x) -\frac{1}{4}\text{Li}_4\left(1-x^2\right)\\
&\,-\frac{1}{2}\text{Li}_2(-x)^2+\text{Li}_2\left(\frac{1-x}{2}\right) \text{Li}_2(-x)-\frac{\pi ^2 }{12}\text{Li}_2(-x)+\text{Li}_4(1-x)\\
&\,+\text{Li}_4(-x)+\text{Li}_4(x)+\frac{1}{2}\text{Li}_4\left(\frac{4 x}{(x+1)^2}\right)-\frac{1}{2}\text{Li}_4\left(\frac{1-x}{1+x}\right)+\frac{1}{2}\text{Li}_4\left(\frac{x-1}{x+1}\right)\\
&\,+2 \text{Li}_4\left(\frac{x}{x+1}\right)-2 \text{Li}_4\left(\frac{2 x}{x+1}\right)+3 \text{Li}_4\left(\frac{1+x}{2}\right)-3 \text{Li}_4\left(\frac{1}{2}\right)+\frac{\pi^4}{480}\,,
\\
H(0,-1,1,-1; x) &\,= -\text{Li}_{2,2}(-1,x)-\text{Li}_{2,2}\left(\frac{1}{2},\frac{2 x}{x+1}\right)-\frac{1}{2} \text{Li}_2\left(\frac{1-x}{2}\right) \log ^2 (1+x)\\
&\, +\frac{1}{2} \text{Li}_2(-x) \log ^2 (1+x) -\frac{1}{2} \text{Li}_2(x) \log ^2 (1+x) -\frac{1}{2} \text{Li}_2(-x) \log ^2 2\\
&\, -\text{Li}_3\left(\frac{1-x}{2}\right) \log  (1+x) +\text{Li}_3(1-x) \log  (1+x) -\text{Li}_3(-x) \log  (1+x)\\
&\, +3 \text{Li}_3(x) \log  (1+x) +3 \text{Li}_3\left(\frac{2 x}{x-1}\right) \log  (1+x)\\
&\, +2 \text{Li}_3\left(\frac{1-x}{1+x}\right) \log  (1+x) +\text{Li}_2(-x) \log  2  \log  (1-x) +\frac{17}{8} \zeta_3 \log  (1+x)\\
&\, -\frac{19}{24} \log ^4 (1+x) +2 \log  2  \log ^3 (1+x) +\frac{1}{2} \log  (1-x)  \log ^3 (1+x)\\
&\, +\frac{7}{6} \log  x  \log ^3 (1+x) +\frac{7}{6} \log ^3 2  \log  (1+x) -\frac{1}{2} \log ^3 (1-x)  \log  (1+x)\\
&\, -\frac{7}{4} \log ^2 2  \log ^2 (1+x) -\frac{3}{2} \log  2  \log  (1-x)  \log ^2 (1+x)
\esp\eeq
\beq\bsp\nonumber
&\, -2 \log  (1-x)  \log  x  \log ^2 (1+x) +\frac{17}{24} \pi ^2 \log ^2 (1+x)\\
&\, +\frac{3}{2} \log  2  \log ^2 (1-x)  \log  (1+x) -\frac{1}{2} \log ^2 2  \log  (1-x)  \log  (1+x)\\
&\, +\frac{3}{2} \log ^2 (1-x)  \log  x  \log  (1+x) -\frac{7}{12} \pi ^2 \log  2  \log  (1+x)\\
&\, -\frac{5}{12} \pi ^2 \log  (1-x)  \log  (1+x) +\frac{1}{2}\text{Li}_2(-x)^2-\text{Li}_2\left(\frac{1-x}{2}\right) \text{Li}_2(-x)\\
&\,+\frac{\pi ^2}{12}\text{Li}_2(-x)-\frac{1}{2}\text{Li}_4(-x)-\frac{3}{2}\text{Li}_4(x)-\frac{3}{4}\text{Li}_4\left(\frac{4 x}{(x+1)^2}\right)\\
&\,-\text{Li}_4\left(\frac{1}{1+x}\right)-2 \text{Li}_4\left(\frac{x}{x+1}\right)+3 \text{Li}_4\left(\frac{2 x}{x+1}\right)-6 \text{Li}_4\left(\frac{1+x}{2}\right)\\
&\,+6 \text{Li}_4\left(\frac{1}{2}\right)+\frac{\pi ^4}{90}\,,
\\
H(0,-1,1,1; x) &\,= -\frac{1}{2} \text{Li}_2(-x) \log ^2 (1-x) +\text{Li}_3\left(\frac{1-x}{2}\right) \log  (1-x) -\text{Li}_3(1-x) \log  (1-x)\\
&\, +\text{Li}_3(-x) \log  (1-x) -\text{Li}_3(x) \log  (1-x) -\text{Li}_3\left(\frac{2 x}{x-1}\right) \log  (1-x)\\
&\, +\frac{7}{4} \zeta_3 \log  (1-x) +\frac{19}{96} \log ^4 (1-x) +\frac{23}{96} \log ^4 (1+x) -\frac{1}{3} \log  2  \log ^3 (1-x)\\
&\, -\frac{7}{12} \log  x  \log ^3 (1-x) -\frac{1}{24} \log  (1+x)  \log ^3 (1-x) -\frac{1}{24} \log ^3 (1+x)  \log  (1-x)\\
&\, -\frac{1}{3} \log  2  \log ^3 (1+x) -\frac{1}{4} \log  x  \log ^3 (1+x) -\frac{1}{3} \log ^3 2  \log  (1+x)\\
&\, +\frac{1}{4} \log ^2 2  \log ^2 (1-x) -\frac{1}{16} \log ^2 (1+x)  \log ^2 (1-x) +\frac{1}{4} \log  x  \log  (1+x)  \log ^2 (1-x)\\
&\, +\frac{3}{16} \pi ^2 \log ^2 (1-x) +\frac{1}{4} \log  x  \log ^2 (1+x)  \log  (1-x) +\frac{1}{2} \log ^2 2  \log ^2 (1+x)\\
&\, -\frac{13}{48} \pi ^2 \log ^2 (1+x) -\frac{\pi ^2}{24}  \log  (1+x)  \log  (1-x) +\frac{\pi ^2}{6} \log  2  \log  (1+x)\\
&\, +\frac{1}{4}\text{Li}_4\left(1-x^2\right)-\frac{1}{4}\text{Li}_4\left(\frac{x^2}{x^2-1}\right)-\text{Li}_4\left(\frac{1-x}{2}\right)-\text{Li}_4(1-x)\\
&\,-\frac{1}{2}\text{Li}_4(-x)+\frac{1}{2}\text{Li}_4(x)+2 \text{Li}_4\left(\frac{x}{x-1}\right)-\text{Li}_4\left(\frac{2 x}{x-1}\right)+\frac{1}{4}\text{Li}_4\left(\frac{4 x}{(x+1)^2}\right)\\
&\,+2 \text{Li}_4\left(\frac{1}{1+x}\right)+2 \text{Li}_4\left(\frac{x}{x+1}\right)-2 \text{Li}_4\left(\frac{2 x}{x+1}\right)+2 \text{Li}_4\left(\frac{1+x}{2}\right)\\
&\,-\text{Li}_4\left(\frac{1}{2}\right)-\frac{\pi ^4}{72}\,,
\\
H(0,0,-1,-1; x) &\,= -\text{Li}_3(-x) \log  (1+x) +\zeta_3 \log  (1+x) +\frac{1}{12} \log ^4 (1+x) -\frac{1}{6} \log  x  \log ^3 (1+x)\\
&\, -\frac{\pi ^2}{12} \log ^2 (1+x) +\text{Li}_4(-x)+\text{Li}_4\left(\frac{1}{1+x}\right)+\text{Li}_4\left(\frac{x}{x+1}\right)-\frac{\pi ^4}{90}\,,
\\
H(0,0,-1,1; x) &\,= \text{Li}_3(-x) \log  (1-x) +\frac{3}{4} \zeta_3 \log  (1-x) +\frac{1}{32} \log ^4 (1-x) +\frac{23}{96} \log ^4 (1+x)\\
&\, -\frac{1}{12} \log  x  \log ^3 (1-x) -\frac{1}{24} \log  (1+x)  \log ^3 (1-x) -\frac{1}{24} \log ^3 (1+x)  \log  (1-x)
\esp\eeq
\beq\bsp\nonumber
&\, -\frac{1}{3} \log  2  \log ^3 (1+x) -\frac{1}{4} \log  x  \log ^3 (1+x) -\frac{1}{3} \log ^3 2  \log  (1+x)\\
&\, -\frac{1}{16} \log ^2 (1+x)  \log ^2 (1-x) +\frac{1}{4} \log  x  \log  (1+x)  \log ^2 (1-x) +\frac{\pi ^2}{16} \log ^2 (1-x)
\\
&\, +\frac{1}{4} \log  x  \log ^2 (1+x)  \log  (1-x) +\frac{1}{2} \log ^2 2  \log ^2 (1+x) -\frac{13}{48} \pi ^2 \log ^2 (1+x)\\
&\, -\frac{\pi ^2}{24} \log  (1+x)  \log  (1-x) +\frac{\pi ^2}{6} \log  2  \log  (1+x) +\frac{1}{4}\text{Li}_4\left(1-x^2\right)\\
&\,-\frac{1}{4}\text{Li}_4\left(\frac{x^2}{x^2-1}\right)-\text{Li}_4(1-x)-\frac{3 }{2}\text{Li}_4(-x)+\frac{1}{2}\text{Li}_4(x)+\text{Li}_4\left(\frac{x}{x-1}\right)\\
&\,+\frac{1}{4}\text{Li}_4\left(\frac{4 x}{(x+1)^2}\right)+2 \text{Li}_4\left(\frac{1}{1+x}\right)+2 \text{Li}_4\left(\frac{x}{x+1}\right)-2 \text{Li}_4\left(\frac{2 x}{x+1}\right)\\
&\,+2 \text{Li}_4\left(\frac{1+x}{2}\right)-2 \text{Li}_4\left(\frac{1}{2}\right)-\frac{\pi ^4}{72}\,,
\\
H(0,0,0,-1; x) &\,= -\text{Li}_4(-x)\,,
\\
H(0,0,0,1; x) &\,= \text{Li}_4(x)\,,
\\
H(0,0,1,-1; x) &\,= \text{Li}_3(x) \log  (1+x) +\frac{3}{4} \zeta_3 \log  (1+x) -\frac{1}{6} \log ^4 (1+x) +\frac{1}{3} \log  2  \log ^3 (1+x)\\
&\, +\frac{1}{6} \log  x  \log ^3 (1+x) +\frac{1}{3} \log ^3 2  \log  (1+x) -\frac{1}{2} \log ^2 2  \log ^2 (1+x) +\frac{\pi ^2}{6} \log ^2 (1+x)\\
&\, -\frac{\pi ^2}{6} \log  2  \log  (1+x) +\frac{1}{2}\text{Li}_4(-x)-\frac{3 }{2}\text{Li}_4(x)-\frac{1}{4}\text{Li}_4\left(\frac{4 x}{(x+1)^2}\right)-\text{Li}_4\left(\frac{1}{1+x}\right)\\
&\,-\text{Li}_4\left(\frac{x}{x+1}\right)+2 \text{Li}_4\left(\frac{2 x}{x+1}\right)-2 \text{Li}_4\left(\frac{1+x}{2}\right)+2 \text{Li}_4\left(\frac{1}{2}\right)+\frac{\pi ^4}{90}\,,
\\
H(0,1,-1,-1; x) &\,= \frac{1}{2} \text{Li}_2(x) \log ^2 (1+x) +\text{Li}_3(-x) \log  (1+x) -\text{Li}_3(x) \log  (1+x)\\
&\, -\text{Li}_3\left(\frac{2 x}{x-1}\right) \log  (1+x) +\text{Li}_3\left(\frac{1}{1+x}\right) \log  (1+x) -\text{Li}_3\left(\frac{1-x}{1+x}\right) \log  (1+x)\\
&\, -\text{Li}_3\left(\frac{1+x}{2}\right) \log  (1+x) -\frac{3}{4} \zeta_3 \log  (1+x) +\frac{1}{6} \log ^4 (1+x) -\frac{1}{2} \log  2  \log ^3 (1+x)\\
&\, -\frac{1}{2} \log  (1-x)  \log ^3 (1+x) -\frac{1}{6} \log  x  \log ^3 (1+x) -\frac{1}{3} \log ^3 2  \log  (1+x)\\
&\, +\frac{1}{6} \log ^3 (1-x)  \log  (1+x) +\frac{1}{4} \log ^2 2  \log ^2 (1+x) +\log  2  \log  (1-x)  \log ^2 (1+x)\\
&\, +\log  (1-x)  \log  x  \log ^2 (1+x) -\frac{\pi ^2}{8} \log ^2 (1+x) -\frac{1}{2} \log  2  \log ^2 (1-x)  \log  (1+x)\\
&\, -\frac{1}{2} \log ^2 (1-x)  \log  x  \log  (1+x) +\frac{\pi ^2}{6} \log  2  \log  (1+x) +\frac{\pi ^2}{6} \log  (1-x)  \log  (1+x)\\
&\, -\frac{1}{2}\text{Li}_4(-x)+\frac{1}{2}\text{Li}_4(x)+\frac{1}{4}\text{Li}_4\left(\frac{4 x}{(x+1)^2}\right)+\text{Li}_4\left(\frac{1}{1+x}\right)-\text{Li}_4\left(\frac{2 x}{x+1}\right)\\
&\,+3 \text{Li}_4\left(\frac{1+x}{2}\right)-3 \text{Li}_4\left(\frac{1}{2}\right)-\frac{\pi ^4}{90}\,,
\esp\eeq
\beq\bsp\nonumber
H(0,1,-1,1; x) &\,= -\frac{47}{96} \log ^4 (1-x) +\log  2  \log ^3 (1-x) +\frac{11}{12} \log  x  \log ^3 (1-x)\\
&\, +\frac{1}{24} \log  (1+x)  \log ^3 (1-x) -\frac{1}{4} \log ^2 2  \log ^2 (1-x) +\frac{9}{16} \log ^2 (1+x)  \log ^2 (1-x)\\
&\, -\log  2  \log  (1+x)  \log ^2 (1-x) -\frac{5}{4} \log  x  \log  (1+x)  \log ^2 (1-x)\\
&\, -\frac{1}{2} \text{Li}_2\left(\frac{1-x}{2}\right) \log ^2 (1-x) +\frac{1}{2} \text{Li}_2(-x) \log ^2 (1-x) -\frac{1}{2} \text{Li}_2(x) \log ^2 (1-x)\\
&\, -\frac{17}{48} \pi ^2 \log ^2 (1-x) -\frac{1}{6} \log ^3 2  \log  (1-x) +\frac{1}{24} \log ^3 (1+x)  \log  (1-x)\\
&\, -\frac{1}{4} \log  x  \log ^2 (1+x)  \log  (1-x) +\frac{\pi ^2}{12} \log  2  \log  (1-x)\\
&\, +\frac{1}{2} \log ^2 2  \log  (1+x)  \log  (1-x) -\frac{\pi ^2}{24} \log  (1+x)  \log  (1-x) \\
&\,+\log  2\,  \text{Li}_2(x) \log  (1-x) -\log  (1+x)  \text{Li}_2(x) \log  (1-x) -\text{Li}_3(-x) \log  (1-x)\\
&\, +\text{Li}_3(x) \log  (1-x) +3 \text{Li}_3\left(\frac{2 x}{x-1}\right) \log  (1-x) -\text{Li}_3\left(\frac{1}{1+x}\right) \log  (1-x)\\
&\, +\text{Li}_3\left(\frac{1-x}{1+x}\right) \log  (1-x) +\text{Li}_3\left(\frac{1+x}{2}\right) \log  (1-x) -\frac{29}{8} \zeta_3 \log  (1-x)\\
&\, -\frac{55}{96} \log ^4 (1+x) +\log  2  \log ^3 (1+x) +\frac{7}{12} \log  x  \log ^3 (1+x) -\frac{3}{2} \log ^2 2  \log ^2 (1+x)\\
&\, +\frac{29}{48} \pi ^2 \log ^2 (1+x) -\frac{\text{Li}_2(x)^2}{2}-\text{Li}_{2,2}(-1,x)+\text{Li}_{2,2}\left(\frac{1}{2},\frac{2 x}{x-1}\right)\\
&\,+\log ^3 2  \log  (1+x) -\frac{\pi ^2}{2} \log  2  \log  (1+x) -\frac{1}{2} \log ^2 2\,  \text{Li}_2(x)-\text{Li}_2\left(\frac{1-x}{2}\right) \text{Li}_2(x)\\
&\,+\text{Li}_2(-x) \text{Li}_2(x)+\frac{\pi ^2}{12}\text{Li}_2(x)+2 \log  (1+x)  \text{Li}_3(x)+3 \text{Li}_4(1-x)\\
&\,-\frac{1}{2}\text{Li}_4(-x)-\frac{3}{2} \text{Li}_4(x)-2 \text{Li}_4\left(\frac{x}{x-1}\right)+3 \text{Li}_4\left(\frac{2 x}{x-1}\right)\\
&\,-\frac{3}{4} \text{Li}_4\left(\frac{4 x}{(x+1)^2}\right)-4 \text{Li}_4\left(\frac{1}{1+x}\right)-4 \text{Li}_4\left(\frac{x}{x+1}\right)+6 \text{Li}_4\left(\frac{2 x}{x+1}\right)\\
&\,-6 \text{Li}_4\left(\frac{1+x}{2}\right)-\frac{1}{4}\text{Li}_4\left(1-x^2\right)+\frac{1}{4}\text{Li}_4\left(\frac{x^2}{x^2-1}\right)+\frac{3}{2} \log  (1+x)  \zeta_3\\
&\,+6 \text{Li}_4\left(\frac{1}{2}\right)+\frac{\pi ^4}{72}\,,
\\
H(0,1,0,-1; x) &\,= -\text{Li}_{2,2}(-1,x)\,,
\\
H(0,1,0,1; x) &\,= 2 \text{Li}_3(x) \log  (1-x) -2 \zeta_3 \log  (1-x) -\frac{1}{12} \log ^4 (1-x) +\frac{1}{3} \log  x  \log ^3 (1-x)\\
&\, -\frac{\pi^2}{6} \log ^2 (1-x) +\frac{\text{Li}_2(x)^2}{2}+2 \text{Li}_4(1-x)-2 \text{Li}_4(x)-2 \text{Li}_4\left(\frac{x}{x-1}\right)-\frac{\pi ^4}{45}\,,
\esp\eeq
\beq\bsp\nonumber
H(0,1,1,-1; x) &\,= \text{Li}_{2,2}(-1,x)-\text{Li}_{2,2}\left(\frac{1}{2},\frac{2 x}{x-1}\right)+\frac{1}{2} \text{Li}_2\left(\frac{1-x}{2}\right) \log ^2 (1-x)\\
&\, -\frac{1}{2} \text{Li}_2(-x) \log ^2 (1-x) +\frac{1}{2} \text{Li}_2(x) \log ^2 (1-x) +\frac{1}{2} \text{Li}_2(x) \log ^2 2\\
&\, -\text{Li}_2(x) \log  2  \log  (1-x) -2 \text{Li}_3\left(\frac{2 x}{x-1}\right) \log  (1-x) -\text{Li}_3(1-x) \log  (1+x)\\
&\, -2 \text{Li}_3(x) \log  (1+x) +\frac{7}{4} \zeta_3 \log  (1-x) -\frac{5}{8} \zeta_3 \log  (1+x) +\frac{7}{24} \log ^4 (1-x)\\
&\, +\frac{3}{8} \log ^4 (1+x) -\frac{2}{3} \log  2  \log ^3 (1-x) -\frac{1}{3} \log  x  \log ^3 (1-x) -\frac{1}{6} \log ^3 2  \log  (1-x)\\
&\, -\frac{2}{3} \log  2  \log ^3 (1+x) -\frac{1}{3} \log  x  \log ^3 (1+x) -\frac{2}{3} \log ^3 2  \log  (1+x)\\
&\, +\frac{1}{2} \log ^2 2  \log ^2 (1-x) +\frac{\pi^2}{8} \log ^2 (1-x) +\log ^2 2  \log ^2 (1+x) -\frac{3}{8} \pi ^2 \log ^2 (1+x)\\
&\, +\frac{\pi^2}{12} \log  2  \log  (1-x) +\frac{\pi^2}{3} \log  2  \log  (1+x) +\frac{1}{4}\text{Li}_4\left(1-x^2\right)+\frac{1}{2}\text{Li}_2(x)^2\\
&\,+\text{Li}_2\left(\frac{1-x}{2}\right) \text{Li}_2(x)-\text{Li}_2(-x) \text{Li}_2(x)-\frac{\pi ^2}{12}\text{Li}_2(x)+\text{Li}_4\left(\frac{1-x}{2}\right)\\
&\,-2 \text{Li}_4(1-x)+\text{Li}_4(-x)+\text{Li}_4(x)-2 \text{Li}_4\left(\frac{2 x}{x-1}\right)+\frac{1}{2}\text{Li}_4\left(\frac{4 x}{(x+1)^2}\right)\\
&\,+3 \text{Li}_4\left(\frac{1}{1+x}\right)-\frac{1}{2}\text{Li}_4\left(\frac{1-x}{1+x}\right)+\frac{1}{2}\text{Li}_4\left(\frac{x-1}{x+1}\right)+2 \text{Li}_4\left(\frac{x}{x+1}\right)\\
&\,-4 \text{Li}_4\left(\frac{2 x}{x+1}\right)+4 \text{Li}_4\left(\frac{1+x}{2}\right)-5 \text{Li}_4\left(\frac{1}{2}\right)-\frac{\pi ^4}{288}\,,
\\
H(0,1,1,1; x) &\,= \frac{1}{2} \text{Li}_2(x) \log ^2 (1-x) +\text{Li}_3(1-x) \log  (1-x) +\frac{1}{3} \log  x  \log ^3 (1-x)\\
&\, -\frac{\pi^2}{12} \log ^2 (1-x) -\text{Li}_4(1-x)+\frac{\pi ^4}{90}\,,
\esp\eeq

\end{document}